\documentclass[draftclsnofoot,onecolumn,letterpaper]{IEEEtran}

\usepackage[utf8]{inputenc}
\usepackage[T1]{fontenc}
\usepackage{url}
\usepackage{ifthen}
\usepackage{cite}
\usepackage{enumitem}
\usepackage{comment}
\usepackage{graphicx}
\usepackage{multirow}
\usepackage{makecell}

\usepackage[caption=false,font=footnotesize]{subfig}

\usepackage[english]{babel}
\usepackage[cmex10]{amsmath}
\usepackage{amssymb}

\usepackage{mathrsfs}
\usepackage{tikz}
\usepackage[utf8]{inputenc}
\usepackage{amsfonts} 
\usepackage{amsthm}
\usepackage{epstopdf}
\usepackage{bbm}
\usepackage{enumerate}
\usepackage{enumitem}
\usepackage{array}
\usepackage{makecell}
\usepackage{multirow}
\usepackage{booktabs}
\usepackage{psfrag,xspace}

\usepackage{bm}
\usepackage{hhline}
\usepackage{braket}
\usepackage{amssymb}
\let\origcap\cap
\let\origcup\cup
\let\origvarnothing\varnothing

\usepackage{mathabx}
\let\cap\origcap
\let\cup\origcup
\let\varnothing\origvarnothing

\usepackage{xcolor}
\definecolor{darkblue}{rgb}{0.1,0.1,0.8}
\definecolor{brickred}{rgb}{0.8, 0.25, 0.33}
\definecolor{DarkGreen}{rgb}{0,0.6,0}

\newtheorem{theorem}{Theorem}

\newtheorem{lemma}[theorem]{Lemma}
\newtheorem{proposition}[theorem]{Proposition}
\newtheorem{corollary}[theorem]{Corollary}

\newtheorem{definition}[theorem]{Definition}

\newtheorem{remark}[theorem]{Remark}
\newtheorem{example}[theorem]{Example}

\newenvironment{subproof}[1][Proof of claim]{%
  \par\addvspace{\parskip}
  \noindent\textit{#1.}\ 
}{%
  \hfill$\square$\par\addvspace{\parskip}
}

\allowdisplaybreaks 
\makeatletter
\def\old@comma{,}
\catcode`\,=13
\def,{%
  \ifmmode%
    \old@comma\discretionary{}{}{}%
  \else%
    \old@comma%
  \fi%
}
\makeatother

\usepackage{etoolbox}
\providetoggle{Blue_revision}
\settoggle{Blue_revision}{true}

 \usepackage{hyperref}
\hypersetup{
colorlinks=true, %
 pdfstartview={FitH},
    linkcolor=red,
    citecolor=blue, 
    urlcolor={blue!80!black}
}

\usepackage{graphicx}
\usepackage{stmaryrd }
\usepackage{stackengine}

\usepackage{relsize}
\usepackage{upgreek}

\newcommand{\supp}{\mathrm{supp}}

\DeclareMathOperator*{\argmax}{argmax}

\newcommand{\tP}{\tilde{P}}
\newcommand{\tQ}{\tilde{Q}}
\newcommand{\bQ}{\widebar{Q}}
\newcommand{\tPx}{\tP_{X^n|\mC}}
\newcommand{\tPy}{\tP_{Y^n|\mC}}
\newcommand{\tPyQ}{\tP_{Y^n|\mC\sim q}}
\newcommand{\bW}{\widebar{W}}
\newcommand{\bV}{\widebar{V}}
\newcommand{\by}{\widebar{y}}
\newcommand{\hy}{\hat{y}}
\newcommand{\hY}{\hat{Y}}
\newcommand{\carX}{|\mX|}
\newcommand{\carY}{|\mY|}
\newcommand{\mT}{\mathcal{T}}
\newcommand{\mP}{\mathcal{P}}
\newcommand{\mQ}{\mathcal{Q}}
\newcommand{\mX}{\mathcal{X}}
\newcommand{\mY}{\mathcal{Y}}
\newcommand{\mM}{\mathcal{M}}
\newcommand{\mA}{\mathcal{A}}
\newcommand{\mB}{\mathcal{B}}
\newcommand{\mG}{\mathcal{G}}
\newcommand{\mC}{\mathcal{C}}
\newcommand{\mV}{\mathcal{V}}

\newcommand{\mS}{\mathcal{S}}
\newcommand{\mE}{\mathcal{E}}
\newcommand{\mD}{\mathcal{D}}

\newcommand{\mJ}{\mathcal{J}}
\newcommand{\hmY}{\hat{\mY}}

\newcommand{\iQ}{\iota(Q_{XY})}

\newcommand{\Eu}{\overline{E}}
\newcommand{\El}{\underline{E}}
\newcommand{\Ga}{\mathit{\Gamma}}
\newcommand{\Gu}{\overline{\Ga}}
\newcommand{\Gl}{\underline{\Ga}}
\newcommand{\Gll}{\underline{\underline{\Ga}}}

\newcommand{\su}{\mathrm{uni}}
\newcommand{\sn}{\mathrm{non}}
\newcommand{\sr}{{-\infty}}
\newcommand{\rc}{\mathrm{rc}}
\newcommand{\nl}{\mathrm{nl}}
\newcommand{\lcm}{\mathrm{lcm}}
\newcommand{\Pe}{\mathrm{Pr_e}}

\newcommand{\minI}{\min_{P_X\in\mS}I(P_X;W)}
\newcommand{\maxI}{\max_{P_X\in\mS}I(P_X;W)}

\newcommand{\TVinEq}{\frac12 \left\|\tPy - P_Y^n \right\|_1}
\newcommand{\TVinText}{\frac12 \big\|\tPy - P_Y^n \big\|_1}
\newcommand{\TVQinText}{\frac12 \big\|\tPyQ - P_Y^n \big\|_1}
\newcommand{\EcTVinEq}{\frac12 \mathbb{E}_{\mC} \left\|\tP_{Y^n|\mC} - P_Y^n\right\|_1}

\newcommand{\equiva}{\ref{lem-equiv}\hyperlink{lem-equiv-a}{(a)}}
\newcommand{\equivb}{\ref{lem-equiv}\hyperlink{lem-equiv-b}{(b)}}
\newcommand{\addtvt}{\ref{lem-Jab}\hyperlink{lem-Jab-add}{(c)}}
\newcommand{\postvt}{\ref{lem-Jab}\hyperlink{lem-Jab-pos}{(b)}}

\begin{document}
\title{On the Strong Converse Exponent and Error Exponent of the Classical Soft Covering\vspace{2ex}}

\author{%
 \IEEEauthorblockN{Xingyi He and S. Sandeep Pradhan}\\
 \IEEEauthorblockA{Department of Electrical Engineering and Computer Science, University of Michigan, USA\\
 Email: \{xingyihe, pradhanv\}@umich.edu} \\
 \vspace{1.5em}
 \IEEEauthorblockN{Andreas Winter} \\
 \IEEEauthorblockA{Department Mathematik/Informatik-Abteilung Informatik, Universit\"at zu K\"oln, K\"oln, Germany \\
 ICREA \& Universitat Aut\`onoma de Barcelona, Barcelona, Spain\\ 
 Institute for Advanced Study, Technische Universit\"at M\"unchen, Garching, Germany\\
    Email: andreas.winter@uni-koeln.de}
\thanks{A preliminary version of this work was presented in part at the 2025 IEEE International Symposium on Information Theory (ISIT), 22-27 June 2025, Ann Arbor MI, doi:10.1109/ISIT63088.2025.11195309 \cite{HePradhanWinter:ISIT}.}
\thanks{XH's and SSP's work was supported in part by NSF grant CCF-2132815. AW's work was supported by the European Commission QuantERA project ExTRaQT (Spanish MICIN grant no.~PCI2022-132965); by the Spanish MICIN (project PID2022-141283NB-I00) with the support of FEDER funds; by the Spanish MICIN with funding from European Union NextGenerationEU (PRTR-C17.I1) and the Generalitat de Catalunya; by the Spanish MTDFP through the QUANTUM ENIA project: Quantum Spain, funded by the European Union NextGenerationEU within the framework of the ``Digital Spain 2026 Agenda''; by the Alexander von Humboldt Foundation; and by the Institute for Advanced Study of the Technical University Munich.}}

\thispagestyle{empty} 

\maketitle

\vspace{-1.5\baselineskip}

\begin{abstract}
    This paper establishes the exact strong converse exponent of the soft covering problem in the classical setting. This exponent characterizes the slowest achievable convergence speed of the total variation to one when a code of rate below mutual information is applied to a discrete memoryless channel for synthesizing a product output distribution. The proposed exponent is expressed through a new two-parameter information quantity, differing from the more commonly studied R\'enyi divergence or R\'enyi mutual information. In addition, we demonstrate the non-tightness of random coding for rates both below and above mutual information. Discussions on the latter start with noiseless channels, where we develop a deterministic code construction that outperforms random codes in error exponents. We further observe that the conventional formulation, which assumes a uniform distribution over messages, inherently introduces a discrepancy in error exponents depending on whether the components of the target distribution are rational or irrational numbers. To eliminate this discrepancy, we propose a new formulation in which messages are allowed to be distributed non-uniformly, and the rate is 
    given by the logarithm of the smallest nonzero message probability (corresponding to R\'enyi entropy $H_{-\infty}$ of order $-\infty$). The exact error exponent is characterized in this formulation for noiseless channels. Furthermore, for noisy channels, we provide a high-rate improvement in achievability and derive a converse bound on the error exponent.
\end{abstract}

\section{Introduction}
\label{sec:intro}

The soft covering lemma is a fundamental lemma used in various information-theoretic problems, such as channel resolvability, channel simulation, and lossy source coding. It first appeared in \cite[Thm.~6.3]{wyner1975common} where Wyner used the normalized (with a factor $1/n$) relative entropy to quantify probabilistic closeness and derived the achievability proof of the common information as the optimal rate of shared randomness between two agents. The concept of soft covering has been consequently applied in wiretap channels \cite{wyner1975wire,hayashi2006general,parizi2016exact,yu2018renyi} and in other problems including channel synthesis \cite{cuff2010coordination,cuff2013distributed,hsieh2016channel,yu2019exact} and channel resolvability \cite{han1993approximation,bloch2013strong,watanabe2014strong,liu2016e_}\cite[Ch.~6]{koga2013information} with tighter measures of probabilistic closeness such as total variation, relative entropy, and R\'enyi divergence \cite{yu2018renyi}\cite{liu2016e_}\cite{winter2005secret,hou2013informational,hou2014effective,cuff2015stronger,cuff2016soft,yu2025renyi,li2025two}. 
It is worth noting that channel resolvability is closely related to soft covering in the sense that the former deals with simulating all output distributions including non-product ones, while the latter focuses solely on the product ones. More recently, developments in quantum information theory have prompted extensive research into the soft covering \cite{ahlswede2002strong,hayashi2015quantum,cheng2023error,atif2023lossy,atif2024quantum,shen2024optimal,he2024quantum,hayashi2025resolvability}\cite[Ch.~17]{wilde2011classical}\cite[Sec.~9.4]{hayashi2016quantum} and channel simulation \cite{massar2000amount,winter2001compression,winter2004extrinsic,luo2009channel,wilde2012information,bennett2014quantum,radhakrishnan2017one,anshu2019convex} for quantum channels.

The general idea of soft covering is to simulate a given output distribution using a specific channel $n$ times. Suppose we are given a discrete memoryless channel (DMC) with transition probability distribution $W_{Y|X}$ and a desired output distribution $P_Y$. Consider any code $\mC = \{X^n(1),X^n(2),\dots,X^n(M)\}$, with $|\mC|=M=:2^{nR}$, and with each encoded message being drawn from the uniform distribution on $\{1,2,\ldots,M\}$. The output distribution induced by the code $\mC$ is
$$ \tP_{Y^n|\mC}(\cdot) = \frac{1}{M} \sum_{i=1}^M W_{Y|X}^n(\cdot|X^n(i)). $$
We aim for $\tP_{Y^n|\mC}$ to effectively cover the space of $Y^n \sim P_Y^n$. In other words, we want the non-product distribution $\tP_{Y^n|\mC}$ to asymptotically approximate the product distribution $P_Y^n$ under the criterion of the total variation $\TVinText$.
Let $\mS:=\{P_X: P_X W = P_Y\}$ be the set of all 1-shot input distributions whose outputs are $P_Y$ under the DMC $W_{Y|X}$. According to the soft covering lemma, e.g.\cite[Ch.~19]{moser2019advanced}, when $R> \minI$, we can achieve a good covering: there exists a code that ensures $\TVinText$ to vanish exponentially fast. The achievable error exponent has been studied extensively in the literature both in the classical and the classical-quantum settings \cite{hayashi2006general,parizi2016exact,yu2018renyi}\cite{yu2025renyi}\cite{li2025two}\cite{cheng2023error}\cite{yassaee2019almost}\cite{yagli2019exact}. All of these studies employ a random coding strategy when investigating the decaying behavior of the total variation and relative entropy at high rates. However, it is not guaranteed that random coding yields a tight exponent; even if it does, it only establishes an achievability result. An important open problem is to develop a converse bound showing that no code can make the total variation decay too fast. In this work, we demonstrate the non-tightness of random coding and provide a converse bound that holds for all codes.

On the other hand, when $R<\minI$, we expect that the covering error $\TVinText$ should approach $1$ exponentially fast: $\TVinText = 1 - 2^{-n\Ga(\mC,n)}$.
We are particularly interested in the minimal achievable exponent:
$$ \Ga(R) := \liminf_{n\to\infty} \min_{\mC:|\mC|= 2^{nR}} \Ga(\mC,n), $$
which characterizes the slowest convergence of $\TVinText$ to $1$. It can be understood as an optimal performance among all the soft covering codes: given a low rate, how well a code can possibly behave to avoid poor covering. Therefore, for an arbitrary code $\mC$, we can deduce that $\TVinText \geq 1 - 2^{-n\Ga(R)}$. In this sense, we refer to $\Ga(R)$ as the strong converse exponent, as it holds for all codes, not just for random codes.

The dual problem of the soft covering lemma is the packing lemma in channel coding, where a similar exponent, called the reliability  function \cite{shannon1959probability}, has been thoroughly studied at rates both below \cite{gallager1965simple,haroutunian1968estimates,sibson1969information,blahut1974hypothesis,csiszar1972class,arimoto1977information,augustin1978noisy,verdu2021error}\cite[Ch.~5]{gallager1968information}\cite[Ch.~10]{csiszar2011information} and above \cite{arimoto1973converse,dueck1979reliability,oohama2015two,mosonyi2017strong} the capacity. The latter case is known as the strong converse of channel coding \cite{wolfowitz1957coding,wolfowitz1960note,winter1999coding}, as the corresponding exponent holds for all codes when the probability of error approaches one. However, in the soft covering, the strong converse exponent $\Ga(R)$ has not been explored in the literature in terms of either lower or upper bounds. It may be noted that \cite{cheng2023error} provides a lower bound on the performance of random code ensemble at rates below mutual information. \cite{watanabe2014strong} and \cite{hayashi2025resolvability} prove the asymptotic strong converse for classical and classical-quantum channels but do not characterize the exponent.

In this work, we characterize the exact strong converse exponent $\Ga(R)$ of the classical soft covering problem. It is stated in Theorem \ref{thm-sc}, and a novel two-parameter R\'enyi-type information quantity $J_{\alpha,\beta}(W_{Y|X}\|P_Y)$ is introduced in the process. The converse is derived using a hypothesis testing perspective, and the achievability is established through a deterministic code construction technique combined with the type covering lemma. In addition, we provide a random coding achievability bound for the strong converse exponent in Theorem \ref{thm-rc} that is expressed using the R\'enyi mutual information. By comparing the exact $\Ga(R)$, the proposed random coding achievability, and the random coding converse in \cite{cheng2023error}, we conclude that random coding fails to produce the tight strong converse exponent. 

Building on the observation of non-tightness of random coding in the strong converse exponents, we provide a new characterization of the error exponents for $R>\minI$. In particular, we establish a new lower bound on the error exponents (see Theorem \ref{thm-nl-El}) in terms of R\'enyi entropy for noiseless channels using a new deterministic code construction. We also derive a new upper bound on the error exponents (see Theorem \ref{thm-nl-Eu}) that match the lower bound until they both reach a slope of unity. The lower bound strictly outperforms random codes for noiseless channels. Furthermore, this lower bound can be extended to noisy channels (see Theorem \ref{thm-ee-achi}), achieving a high-rate improvement over the random coding exponent given in \cite[Thm.~1]{yagli2019exact}. We also provide a new sphere-packing style upper bound on the error exponents of the noisy channels, given in Theorem \ref{thm-ee-con}. 

While investigating the behavior of error exponents for noiseless channels, we uncover an interesting phenomenon: the conventional formulation that assumes a uniform distribution $1/M$ on messages yields different exponents depending on whether the components of the target distribution $P_Y$ are rational or irrational numbers. These discrepancies are characterized in Theorems \ref{thm-irr} and \ref{thm-ra}. The reason for this is that the code-induced distribution $\tPy$ takes values in multiples of $1/M$, so soft covering essentially approximates $P_Y^n$ within a quantization grid of step size $1/M$. If $P_Y^n$ is rational, a perfect covering with zero error is achievable at sufficiently high rates. In contrast, if $P_Y^n$ is irrational, a nonzero covering error may persist due to limitations imposed by Diophantine approximation \cite[Ch.~3]{queffelec2013diophantine}. To address this discrepancy, we propose a new formulation of the soft covering problem, in which the distribution of the messages is allowed to be non-uniform, and the covering rate is defined not by the number of messages but by the inverse of the smallest nonzero message probability. We call this formulation $H_{-\infty}$-constrained soft covering (see Definition \ref{def-ren}). We provide 
an exact characterization of the error exponents for this formulation without any discrepancy for all rates for noiseless channels (see Theorem \ref{thm-nl-ren}).
It aligns with that of the uniform distribution at low rates.  
In summary, the error exponents of the soft covering problems exhibit a very complex behavior.

This paper is organized as follows. Section \ref{sec-pre} introduces useful definitions. All the main results are presented in Section \ref{sec-result}. Section \ref{sec-sc} provides a detailed proof of the characterization of  strong converse exponent. In Section \ref{sec-ee-nl}, we provide a detailed discussion on error exponents for noiseless channels under different formulations and generalize the corresponding results to noisy channels in Section \ref{sec-ee-ny}.

\section{Preliminaries}
\label{sec-pre}

\subsection{Notation}
This paper applies the method of types, as well as basic notions from discrete probability theory and a range of information quantities. Basic properties of types can be found in many papers and textbooks, e.g.~\cite[Ch.~2]{csiszar2011information} and \cite{csiszar1998method}.
Basic notations we will use are collected below.

Let $\mM$ denote a finite message set, $\mX$ and $\mY$ denote the finite input and output alphabets of a channel, respectively, and $\mP(\mM)$, $\mP(\mX)$, and $\mP(\mX\mY)$ denote the sets of all probability distributions on $\mM$, $\mX$, and $\mX\times\mY$, respectively. 
Likewise, let $\mP(\mY|\mX)$ denote the set of all conditional distributions on $\mY$ given $\mX$. 
For any block length $n$, let $\mP_n(\mX)$ and $\mP_n(\mX\mY)$ denote the sets of all $n$-types and joint $n$-types on $\mX^n$ and $\mX^n\times\mY^n$, respectively, 
and let $\mP_n(\mY|Q_X)$ denote the set of all conditional $n$-types on $\mY^n$ given an input type $Q_X\in\mP_n(\mX)$. 
Denote by $\mT_Q$ the type class corresponding to an $n$-type $Q$, 
and by $\mT_V(x^n)$ the conditional type class (or $V$-shell) given $x^n\in\mT_{Q_X}$ and a conditional type $V\in\mP_n(\mY|Q_X)$. 
Let $H(V|Q_X)$ and $I(Q_X;V)$ denote the conditional entropy $H(Y|X)$ and the mutual information $I(X;Y)$ under the joint distribution $Q_{XY}=Q_X V_{Y|X}$, respectively, 
and define $D(V\|W|Q_X) := D(Q_XV_{Y|X}\|Q_XW_{Y|X}) $.
Let $\mS := \{P_X\in\mP(\mX): P_X W = P_Y\}$ denote the set of all input distributions whose induced output distribution under $W_{Y|X}$ is $P_Y$. 
For a distribution $V$, let $\supp(V)$ denote its support, and write $V\ll W$ if $\supp(V)\subseteq\supp(W)$. 
Finally, define $|x|^+ := \max\{x,0\}$.

Consider a joint distribution $Q_{XY}\in\mP(\mX\mY)$ with marginal distributions $Q_X\in\mP(\mX)$ and $Q_Y\in\mP(\mY)$. In this paper, we follow the notation that $Q_{XY} = Q_X V_{Y|X} = Q_Y \bV_{X|Y}$, where $V_{Y|X}:= Q_{XY}/Q_X\in\mP(\mY|\mX)$ and $\bV_{X|Y} = Q_{XY}/Q_Y\in\mP(\mX|\mY)$ are the forward and backward transition probabilities, respectively. In other words, the letters $Q$ and $V$ are consistently associated throughout this paper's notations: whenever a joint distribution $Q_{XY}$ appears, the corresponding marginals $Q_X,Q_Y$ and conditionals $V_{Y|X},\bV_{X|Y}$ are all implicitly defined and need not be restated. Moreover, we write $Q_Y = Q_X V$. 

In this work, the desired output distribution is denoted by $P_Y\in\mP(\mY)$ and the DMC used for covering is denoted by $W_{Y|X}\in\mP(\mY|\mX)$. Definitions of some useful information-theoretic quantities are listed below.

\begin{definition}[Information density]\label{def-iXY}
    The information density is defined as
    $$ \iota_{X;Y}(x,y) := \log \frac{W_{Y|X}(y|x)}{P_Y(y)}. $$
\end{definition}

\begin{definition}[Expectation of the information density]\label{def-iQ}
    Given $Q_{XY}\in\mP(\mX\mY)$, the expectation of the information density under $Q_{XY}$ is defined as
    $$ \iota(Q_{XY}) := \mathbb{E}_{Q_{XY}} [\iota_{X;Y}] = 
    \begin{cases}
        \displaystyle{ \sum_{x,y} Q_{XY}(x,y) \log\frac{W_{Y|X}(y|x)}{P_Y(y)} } 
            & \text{if } V\ll W \\
        -\infty & \text{otherwise}
    \end{cases}. $$
\end{definition}

\begin{lemma}\label{lem-Deq} One can easily verify the following identity, and hence we claim it without proof.
    $$ D(V\|W|Q_X) + \iQ = D(Q_X V\|P_Y) + I(Q_X;V). $$ 
\end{lemma}

\begin{definition}[$\alpha$-R\'enyi entropy]\label{def-Ha}
    Given $Q_X\in\mP(\mX)$, the $\alpha$-R\'enyi entropy is defined as 
    $$ H_\alpha(Q_X) 
    := \frac{1}{1-\alpha} \log \left(\sum_x Q_X^\alpha(x) \right). $$
\end{definition}

\begin{remark}\label{rmk-Hinf}
    $H_{-\infty}(Q_X) = -\log \displaystyle{\min_{x\in\supp(Q_X)}} Q_X(x)$.
\end{remark}

\begin{definition}[$\alpha$-R\'enyi divergence]\label{def-Da}
    Given $Q_X,P_X\in\mP(\mX)$, their $\alpha$-R\'enyi divergence is defined as 
    $$ D_\alpha(Q_X \| P_X) 
    := \frac{1}{\alpha-1} \log \left( \sum_x Q_X^\alpha(x) P_X^{1-\alpha}(x) \right). $$
\end{definition}

\begin{definition}[$\alpha$-R\'enyi mutual information]\label{def-Ia}
    Given $P_X\in\mP(\mX)$, the $\alpha$-R\'enyi mutual information is defined as 
    $$ I_\alpha(P_X;W) 
    := \frac{\alpha}{\alpha-1} \log \left[ \sum_y \left(\sum_x P_X(x) W^\alpha_{Y|X}(y|x) \right)^\frac{1}{\alpha} \right]. $$
\end{definition}

\begin{lemma}[\text{\cite[Rmk.~24]{yagli2019exact}}]\label{lem-Ia-yagli}
    Let $P_X\in\mS$ and define $\bW_{X|Y} := P_X W_{Y|X}/{P_Y} \in \mP(\mX|\mY)$. Then 
    $$ I_\alpha(P_X;W) 
    = \frac{\alpha}{\alpha-1} \log \left(\mathbb{E}_{P_Y} \left[\mathbb{E}_{\bW_{X|Y}}^{1/\alpha} \left[2^{(\alpha-1)\iota_{X;Y}}|Y\right] \right] \right). $$
\end{lemma}

\begin{remark}\label{rmk-H=I}
    If $W_{Y|X}$ is a noiseless, i.e., for each $x\in\mX$, there exists a symbol $y\in\mY$ such that $W_{Y|X}(y|x) = 1$. Then $I_\alpha(P_X;W) = H_{\frac{1}{\alpha}}(P_X W)$ \cite[Eq.~(85)]{verdu2021error}. 
\end{remark}

Next, we define a novel two-parameter information-theoretic quantity below, which we will use to characterize the strong converse exponent (see Theorem \ref{thm-sc}). Some properties of this quantity are summarized in Lemma \ref{lem-Jab} and proved in Appendix \ref{app-Jab}.

\begin{definition}[A two-parameter information-theoretic quantity] \label{def-Jab}
    Let $P_Y\in\mP(\mY)$ and $W_{Y|X}\in\mP(\mY|\mX)$ be given such that $\mS$ is nonempty. 
    Let $\alpha,\beta\in[0,1]$. 
    We define a two-parameter information-theoretic quantity as follows:
    $$ J_{\alpha,\beta}\left(W_{Y|X}\|P_Y\right)
        := -\log \left[ \max_{Q_X\in\mP(\mX)} \sum_y P_Y^{1-\alpha\beta}(y) \left(\sum_x Q_X(x) W_{Y|X}^\alpha(y|x) \right)^\beta \right]. $$
\end{definition}

\begin{lemma}[Properties of $J_{\alpha,\beta}$] \label{lem-Jab}
    The quantity $J_{\alpha,\beta}\left(W_{Y|X}\|P_Y\right)$ defined above has the following properties.
    \begin{itemize}

        \item [(a)] Given $Q_X\in\mP(\mX)$, let $\hat{Q}_Y(y) \ \propto \ \big[\sum_x Q_X(x) W_{Y|X}^\alpha(y|x) \big]^{\frac{1}{\alpha}}$ be a normalized distribution on $\mY$. Then
        $$ J_{\alpha,\beta}\left(W_{Y|X}\|P_Y\right) 
        = \min_{Q_X\in\mP(\mX)} \left[ 
            (1 - \alpha\beta) D_{\alpha\beta}(\hat{Q}_Y\|P_Y) +
            \beta(1-\alpha) I_\alpha(Q_X;W) \right], $$
        where $D_{\alpha\beta}(\hat{Q}_Y\|P_Y)$ is defined in Definition \ref{def-Da} and $I_\alpha(Q_X;W)$ is defined in Definition \ref{def-Ia}.

        \item [(b)] \hypertarget{lem-Jab-pos}{}
        (Non-negativity)
        If $\alpha,\beta\in[0,1]$, then
        $J_{\alpha,\beta}\left(W_{Y|X}\|P_Y\right) \geq 
        J_{1,\beta}\left(W_{Y|X}\|P_Y\right) = 0$.

        \item [(c)] \hypertarget{lem-Jab-add}{}
        (Additivity)
        $J_{\alpha,\beta}\big(W_{Y|X}^n\|P_Y^n\big) = n J_{\alpha,\beta}\left(W_{Y|X}\|P_Y\right).$
    \end{itemize}
\end{lemma}

\subsection{Notions of soft covering}
The precise formulations of the soft covering problem are given in this subsection. We begin with two cases: uniform and non-uniform, depending on whether the message set $\mM$ has a uniform distribution or not. 

\begin{definition}[Uniform soft covering]\label{def-uni}
    An $(R,n,P_Y,W_{Y|X})$ (uniform) soft-covering scheme consists of a message set $\mM = \{1,\dots,M\}$ with $M=:2^{nR}$, where each message is drawn from uniform distribution, and a code $\mC = \{X^n(1),X^n(2),\dots,X^n(M)\}$.
    This soft-covering scheme induces the following distribution at the output of the DMC $W_{Y|X}$:
    \begin{equation}\label{tPy}
        \tP_{Y^n|\mC}(y^n) := \frac{1}{M} \sum_{i=1}^M W_{Y|X}^n(y^n|X^n(i)), 
        \quad \forall y^n\in\mY^n.
    \end{equation}
    The covering error is defined as the total variation between the achieved non-product output distribution $\tPy$ and the product desired one $P_Y^n$:
    $$ \frac12 \left\| \tPy - P_Y^n \right\|_1 
        := \frac12 \sum_{y^n} \left| \tPy(y^n) - P_Y^n(y^n) \right|. $$
\end{definition}

\begin{definition}[Non-uniform soft covering]\label{def-non}
    An $(R,n,P_Y,W_{Y|X},q)$ (non-uniform) soft-covering scheme consists of a message set $\mM = \{1,\dots,M\}$ with $M=:2^{nR}$, where messages are drawn from distribution $q\in\mP(\mM)$, and a code $\mC = \{X^n(1),X^n(2),\dots,X^n(M)\}$.
    This soft-covering scheme induces the following distribution at the output of the DMC $W_{Y|X}$:
    \begin{equation}\label{tPyQ}
        \tPyQ(y^n) := \sum_{i=1}^M q(i) \ W_{Y|X}^n(y^n|X^n(i)), 
        \quad \forall y^n\in\mY^n.
    \end{equation}
    The covering error is defined as the total variation between the achieved non-product output distribution $\tPyQ$ and the product desired one $P_Y^n$:
    \begin{equation}\label{l1-non}
        \frac12 \left\| \tPyQ - P_Y^n \right\|_1 
        := \frac12 \sum_{y^n} \left| \tPyQ(y^n) - P_Y^n(y^n) \right|.
    \end{equation}
\end{definition}

Besides these two formulations, we introduce a variation based on the non-uniform case: in addition to assuming that the messages have a non-uniform distribution $q$, we further impose an extra condition that the smallest nonzero probability among all messages is at least $1/M =: 2^{-nR}$. This condition can be equivalently written as
\begin{equation}\label{H-inf<nR}
    H_{-\infty}(q) = -\log\min_{i\in\mM:q(i)>0} q(i) \leq nR,
\end{equation}
according to Remark \ref{rmk-Hinf}. Under this condition, the size of the message set need not be $M$; instead, $|\mM| = |\supp(q)| \leq M$. Moreover, given a rate $R$, define
\begin{equation}\label{QRn}
    \mQ(R,n) := \left\{q\in\mP(\{1,\dots,M\}): H_{-\infty}(q) \leq nR \right\}
\end{equation}
to be the set of all probabilities satisfying condition \eqref{H-inf<nR}. Soft covering under this condition is referred to as $H_{-\infty}$-constrained formulation, and is given by the following definition.

\begin{definition}[$H_{-\infty}$-constrained soft covering]
\label{def-ren}
An $(R,n,P_Y,W_{Y|X},q)_{-\infty}$ ($H_{-\infty}$-constrained) soft-covering scheme consists of a message set $\mM$, where messages are drawn from distribution $q\in\mQ(R,n)$, and a code $\mC = \{X^n(1),X^n(2),\dots,X^n(|\mM|)\}$.
    This soft-covering scheme induces distribution $\tPyQ$ at the output of the DMC $W_{Y|X}$, which has the same expression in \eqref{tPyQ}, and the covering error is defined the same as \eqref{l1-non}.
\end{definition}

The reason for this new formulation of the soft covering problem is that, when $R\geq\minI$, the conventional uniform formulation would lead to an inevitable discrepancy depending on whether $P_Y(y)$'s are rational or irrational. This discrepancy arises even in the simplest case, when the DMC $W_{Y|X}$ is noiseless, because in this case, the uniform formulation attempts to approximate $P_Y$ using only rational probabilities which are multiples of $1/M$. Switching to the non-uniform formulation in Definition \ref{def-non} can certainly eliminate this rational-irrational discrepancy, but alters the problem substantially and thereby shifts the error exponent relative to the uniform case. The $H_{-\infty}$-constrained formulation in Definition \ref{def-ren}, on the other hand, resolves the discrepancy while still sharing an overlapping error-exponent region with the uniform formulation. A detailed discussion of this point is presented in Section \ref{sec-result-ee-nl} and Section \ref{sec-ee-nl}.

In information theory, proofs of achievability are commonly facilitated by the technique of random coding; accordingly, we also define the random-coding soft covering below. In this definition, the message set $\mM$ is uniformly distributed.

\begin{definition}[Random-coding soft covering]\label{def-rc}
    An $(R,n,W_{Y|X})_{P_X}$ (random-coding) soft-covering scheme consists of a message set $\mM = \{1,\dots,M\}$ with $M=:2^{nR}$, where each message is drawn from uniform distribution, and a random code $\mC = \{X^n(1),X^n(2),\dots,X^n(M)\}$, where each codeword $X^n(i)$ is generated randomly from i.i.d. $P_X$, i.e., $X^n(i)\sim P_X^n$ for all $i\in\mM$. The induced output distribution $\tPy$ is given by \eqref{tPy}, and the covering error is defined as the expectation of the total variation between $\tP_{Y^n|\mC}$ and $\mathbb{E}_\mC \tP_{Y^n|\mC}$:
    $$ \frac12 \mathbb{E}_\mC \left\| \tPy - \mathbb{E}_\mC \tP_{Y^n|\mC} \right\|_1 
    := \frac12 \mathbb{E}_\mC \sum_{y^n} \left| \tPy(y^n) - \mathbb{E}_\mC \tP_{Y^n|\mC}(y^n) \right|, $$
    where the expectation $\mathbb{E}_\mC$ is taken with respect to i.i.d. $P_X$. 
\end{definition}

\begin{remark}
\label{rmk-CsimQ}
The notation $\mC\sim q$ indicates that $\mC$ is the code corresponding to a message set with distribution $q$. If $q(i) = 1/M$ is uniform, we omit the notation $\sim q$, as in Definitions \ref{def-uni} and \ref{def-rc} above. This convention will be used throughout the paper: specifically, every occurrence of $\mC\sim q$ or $\tPyQ$ refers to the non-uniform formulation (including Definitions \ref{def-non} and \ref{def-ren}), whereas $\tPy$ refers to the uniform formulation (including Definitions \ref{def-uni} and \ref{def-rc}), unless otherwise specified.
\end{remark}

Next, we define the error exponent (denoted by $E$) and the strong converse exponent (denoted by $\Ga$) for the above four formulations of the soft covering problem.

\begin{definition}[Error exponent and strong converse exponent for soft covering] \label{def-exp} \ 
\begin{itemize}
\item [(a)] For uniform soft-covering:
\begin{align*}
    E_\su(R) &:= \limsup_{n\to\infty} \ \max_{\mC:|\mC|= 2^{nR}} \left[ -\frac1n \log \left( \frac12 \left\| \tPy - P_Y^n \right\|_1 \right) \right], \\
    \Ga_\su(R) &:= \liminf_{n\to\infty} \ \min_{\mC:|\mC|= 2^{nR}} \left[ -\frac1n \log \left( 1 - \TVinEq \right) \right].
\end{align*}

\item [(b)] For non-uniform soft-covering:
\begin{align*}
    E_\sn(R) &:= \limsup_{n\to\infty} \ \max_{q\in\mP(\{1,\dots,2^{nR}\})} \ \max_{\mC:\mC\sim q} \left[ -\frac1n \log \left( \frac12 \left\| \tPyQ - P_Y^n \right\|_1 \right) \right], \\
    \Ga_\sn(R) &:= \liminf_{n\to\infty} \ \min_{q\in\mP(\{1,\dots,2^{nR}\})} \ \min_{\mC:\mC\sim q} \left[ -\frac1n \log \left( 1 - \frac12 \left\| \tPyQ - P_Y^n \right\|_1 \right) \right],
\end{align*}
where the notation $\mC\sim q$ is defined in Remark \ref{rmk-CsimQ}.

\item [(c)] For $H_{-\infty}$-constrained soft covering:
\begin{align*}
    E_\sr(R) &:= \limsup_{n\to\infty} \ \max_{q\in\mQ(R,n)} \ \max_{\mC:\mC\sim q} \left[ -\frac1n \log \left( \frac12 \left\| \tPyQ - P_Y^n \right\|_1 \right) \right], \\
    \Ga_\sr(R) &:= \liminf_{n\to\infty} \ \min_{q\in\mQ(R,n)} \ \min_{\mC:\mC\sim q} \left[ -\frac1n \log \left( 1 - \frac12 \left\| \tPyQ - P_Y^n \right\|_1 \right) \right],
\end{align*}
where $\mQ(R,n)$ is defined in \eqref{QRn} and the notation $\mC\sim q$ is defined in Remark \ref{rmk-CsimQ}.

\item [(d)] For random-coding soft covering with codewords randomly drawn from i.i.d. $P_X$:
\begin{align*}
    E_\rc(R,P_X) &:= \limsup_{n\to\infty} \left[ -\frac1n \log \left( \frac12 \mathbb{E}_\mC \left\| \tPy - \mathbb{E}_\mC\tPy \right\|_1 \right) \right], \\
    \Ga_\rc(R,P_X) &:= \liminf_{n\to\infty} \left[ -\frac1n \log \left( 1 -\frac12 \mathbb{E}_\mC \left\| \tPy - \mathbb{E}_\mC\tPy \right\|_1 \right) \right].
\end{align*}
\end{itemize}
\end{definition}

\begin{remark}\label{rmk-relation}
    It is clear that $E_\su(R) \leq  E_\sr(R) \leq E_\sn(R)$ 
    and $\Ga_\su(R) \geq \Ga_\sr(R) \geq \Ga_\sn(R)$, since the $H_{-\infty}$-constrained formulation is a special case of the non-uniform formulation, and the uniform formulation is a special case of the $H_{-\infty}$-constrained formulation.
\end{remark}

\begin{remark}\label{rmk-cuff}
    It is shown in \cite[Thm.~1]{yagli2019exact} that for all $P_X\in\mS$, 
    \begin{align*}
        E_\rc(R,P_X) &= \min_{Q_{XY}\in\mP(\mX\mY)} \left[ D(Q_{XY}\|P_X W_{Y|X}) + \frac12 \big|R - D(Q_{XY}\|P_X Q_Y)\big|^+ \right] \\
        &= \max_{\alpha\in[1,2]} \left\{ \frac{\alpha-1}{\alpha} \left[R - I_\alpha(P_X;W)\right] \right\}, 
    \end{align*}
    where $I_\alpha(P_X;W)$ is the $\alpha$-R\'enyi mutual information, defined in Definition \ref{def-Ia}. Consequently, the optimal random coding error exponent is
    $ E_\rc(R) := \max_{P_X\in\mS} E_\rc(R,P_X)$. 
    In particular, by Remark \ref{rmk-H=I}, when the DMC $W_{Y|X}$ is noiseless, $E_\rc(R)$ reduces to
    \begin{align*}
      E_\rc^\nl(R)
        &= \min_{Q_Y\in\mP(\mY)} 
        \left\{ D(Q_Y\|P_Y) + \frac12 \big| R - D(Q_Y\|P_Y) - H(Q_Y) \big|^+ \right\} \\
        &=\ \max_{\alpha\in[1,2]} \left\{ \frac{\alpha-1}{\alpha} \left[R - H_{\frac{1}{\alpha}}(P_Y)\right] \right\}.
    \end{align*}
\end{remark}

\section{Main Results}
\label{sec-result}
This section summarizes the main results of this work, including the strong converse exponent in Section \ref{sec-result-sc} and bounds for the error exponent in Sections \ref{sec-result-ee-nl} and \ref{sec-result-ee-ny}. A summary table of these results is presented in Section \ref{sec-result-summary}.

\subsection{Results on the strong converse exponent}
\label{sec-result-sc}
To begin with, the following Theorem \ref{thm-sc} characterizes the exact strong converse exponents of the soft covering problem in the uniform, non-uniform, and $H_{-\infty}$-constrained formulations, and shows that they are all identical.

\begin{theorem}[The exact strong converse exponent]
\label{thm-sc}
The exact strong converse exponents in Definition \ref{def-exp} for the uniform, non-uniform, and the $H_{-\infty}$-constrained soft-covering formulations coincide, and are given by
    $$ \Ga_\su(R) = \Ga_\sn(R) = \Ga_\sr(R) = \Ga(R) 
    := \max_{\alpha,\beta\in[0,1]} 
        \left[ J_{\alpha,\beta}\left(W_{Y|X}\|P_Y\right) 
        + \beta(\alpha-1) R \right], $$
    where $J_{\alpha,\beta}\left(W_{Y|X}\|P_Y\right)$ is defined in Definition \ref{def-Jab}.
\end{theorem}

Theorem \ref{thm-sc} is proved in Section \ref{sec-sc-dual}. The quantity $J_{\alpha,\beta}\left(W_{Y|X}\|P_Y\right)$ involves two parameters, which is uncommon in the literature of error exponents. Similar two-parameter exponents have appeared in the R\'enyi resolvability problem, e.g., \cite{yu2025renyi}\cite{li2025two}\cite{hayashi2016equivocations}, in which the R\'enyi divergence, itself parameterized, serves as the error criterion in the formulation.
Moreover, in the context of lossy source coding, e.g., \cite{blahut1974hypothesis}\cite{jitsumatsu2025computation}, exponents taking two parameters have also been observed; however, lossy source coding and soft covering address fundamentally different problems. The former employs a distortion function as a criterion for loss and directly evaluates the probability of error, whereas the latter centers on a forward channel and quantifies the error through a probabilistic divergence. Although the notion of covering is intrinsically revealed in source coding, the two settings are structurally distinct. Typically, when a channel is involved in a parameter-free problem formulation, the error exponent takes the form of $\max_\alpha \frac{1-\alpha}{\alpha}\left(I_\alpha - R\right)$, where $I_\alpha$ is some R\'enyi mutual information with respect to the given channel, with different domains of maximization over $\alpha$. This form is observed in both packing \cite{gallager1965simple}\cite{sibson1969information}\cite{arimoto1973converse}\cite{arimoto1977information}\cite{augustin1978noisy}\cite{csiszar1995generalized}\cite{verdu2015alpha}\cite{verdu2021error} and covering (when $R\geq\minI$) \cite{hayashi2006general}\cite{parizi2016exact}\cite{yassaee2019almost}\cite{yagli2019exact} problems. In this work, we give an achievability bound for random codes that exactly takes this form, as formally stated in Theorem \ref{thm-rc}.

\begin{theorem}[A random coding achievability bound]\label{thm-rc} The strong converse exponent in Definition \ref{def-exp} for the random coding formulation has the following upper bound. For all $P_X\in\mS$, we have
    \begin{align}
        \Ga_\rc(R,P_X) \leq \Gu_\rc(R,P_X)
        := & \min_{Q_{XY}\in\mP(\mX\mY)} \left[ D(Q_{XY}\|P_X W_{Y|X}) + \big|D(Q_{XY}\|P_X Q_Y) - R\big|^+ \right] \label{thm-rc-a} \\
        = & \max_{\alpha\in[\frac12,1]} \left\{ \frac{1-\alpha}{\alpha} \left[I_\alpha(P_X;W) - R \right] \right\}, \label{thm-rc-b}
    \end{align}
    where $I_\alpha(P_X;W)$ is the $\alpha$-R\'enyi mutual information, defined in Definition \ref{def-Ia}.
\end{theorem}

See Appendix \ref{app-ran} for a proof of Theorem \ref{thm-rc}. On the other hand, a lower bound for $\Ga_\rc(R,P_X)$ is given in \cite[Thm.~3]{cheng2023error}:
$$ \Ga_\rc(R,P_X) \geq \Gl_\rc(R,P_X):= 
    \max_{\alpha\in[\frac12,1]} \left\{ \frac{1-\alpha}{\alpha} \left[D_{2-\frac{1}{\alpha}}\left(P_X W_{Y|X}\|P_X P_Y\right) - R \right] \right\} $$
for all $P_X\in\mS$, where $D_{2-\frac{1}{\alpha}}\left(P_X W_{Y|X}\|P_X P_Y\right)$ is the R\'enyi divergence in Definition \ref{def-Da}. 
Hence, defining 
$$ \Ga_\rc(R) := \min_{P_X\in\mS}\Ga_\rc(R,P_X), $$ 
it can be further bounded via
\begin{equation}\label{Ga_rc-bound}
    \Gl_\rc(R) \leq \Ga_\rc(R) \leq \Gu_\rc(R)
\end{equation}
with 
$$ \Gl_\rc(R) := \min_{P_X\in\mS} \Gl_\rc(R,P_X), \quad 
\Gu_\rc(R) := \min_{P_X\in\mS} \Gu_\rc(R,P_X). $$
It is also noteworthy that $\Ga_\rc(R)$ provides an achievability bound for $\Ga_\su(R)$; specifically, $\Ga_\su(R) \leq \Ga_\rc(R)$. Examples of the exact exponent $\Ga(R)$ (in Theorem \ref{thm-sc}) and aforementioned random coding bounds $\Gl_\rc(R)$ and $\Gu_\rc(R)$ are illustrated in Figure \ref{fig-rc}, where the matrix entry $W(x,y)$ represents $W_{Y|X}(y|x)$. These examples include channels with fully noisy, fully noiseless, and hybrid input symbols (i.e., a mix of noisy and noiseless inputs). Figure~\ref{fig-rc} clearly shows that the random coding is not tight in general in the converse regime of the soft covering problem: due to \eqref{Ga_rc-bound}, the random-coding exponent $\Ga_\rc(R)$ must lie between $\Gl_\rc(R)$ and $\Gu_\rc(R)$, while the exact exponent $\Ga(R)$ is lower and hence tighter.

\begin{figure}[!htbp]
    \centering
    \subfloat[A fully noisy binary channel]{%
        \includegraphics[width=0.49\linewidth]{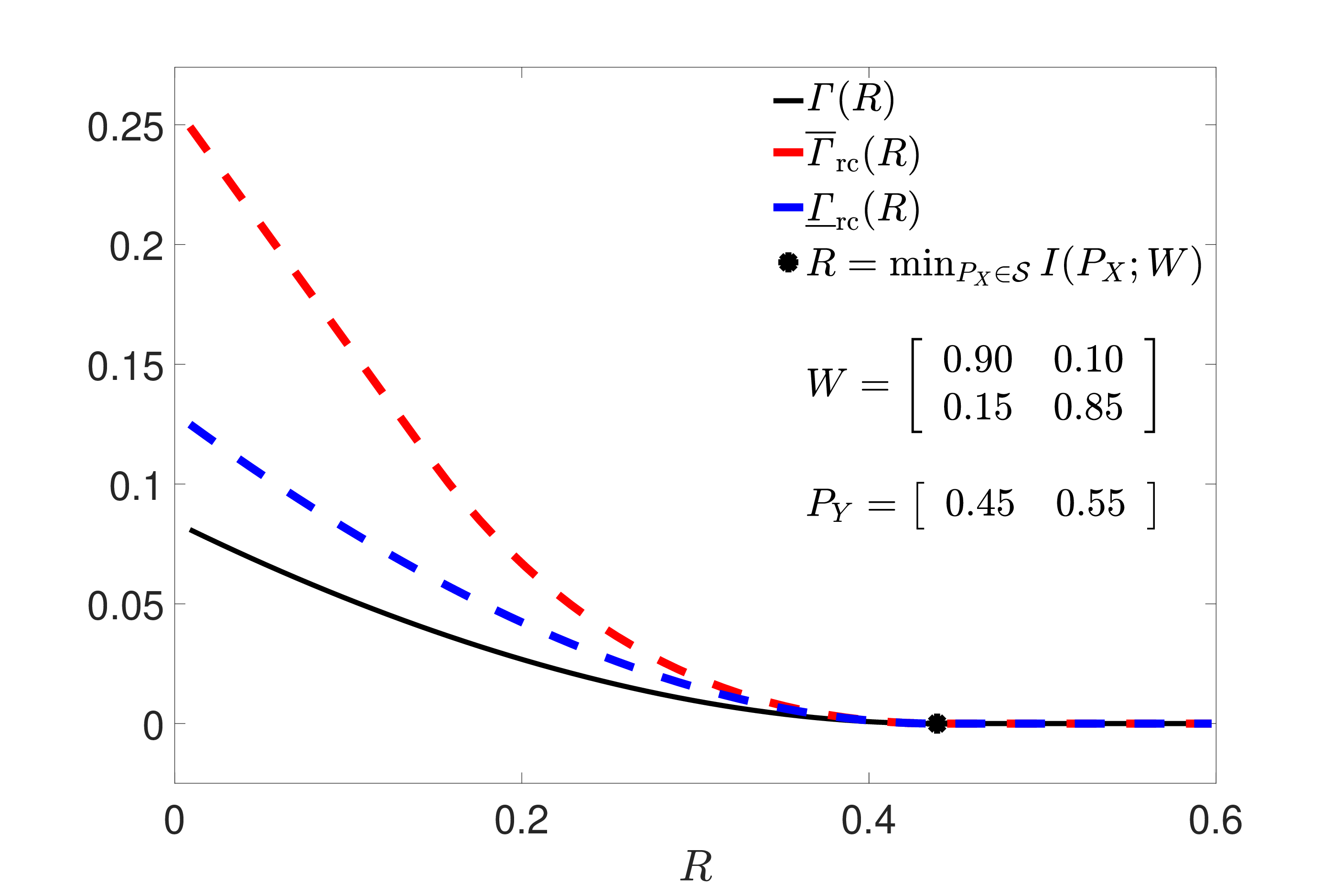}%
    }\hfill
    \subfloat[A fully noisy ternary channel]{%
        \includegraphics[width=0.49\linewidth]{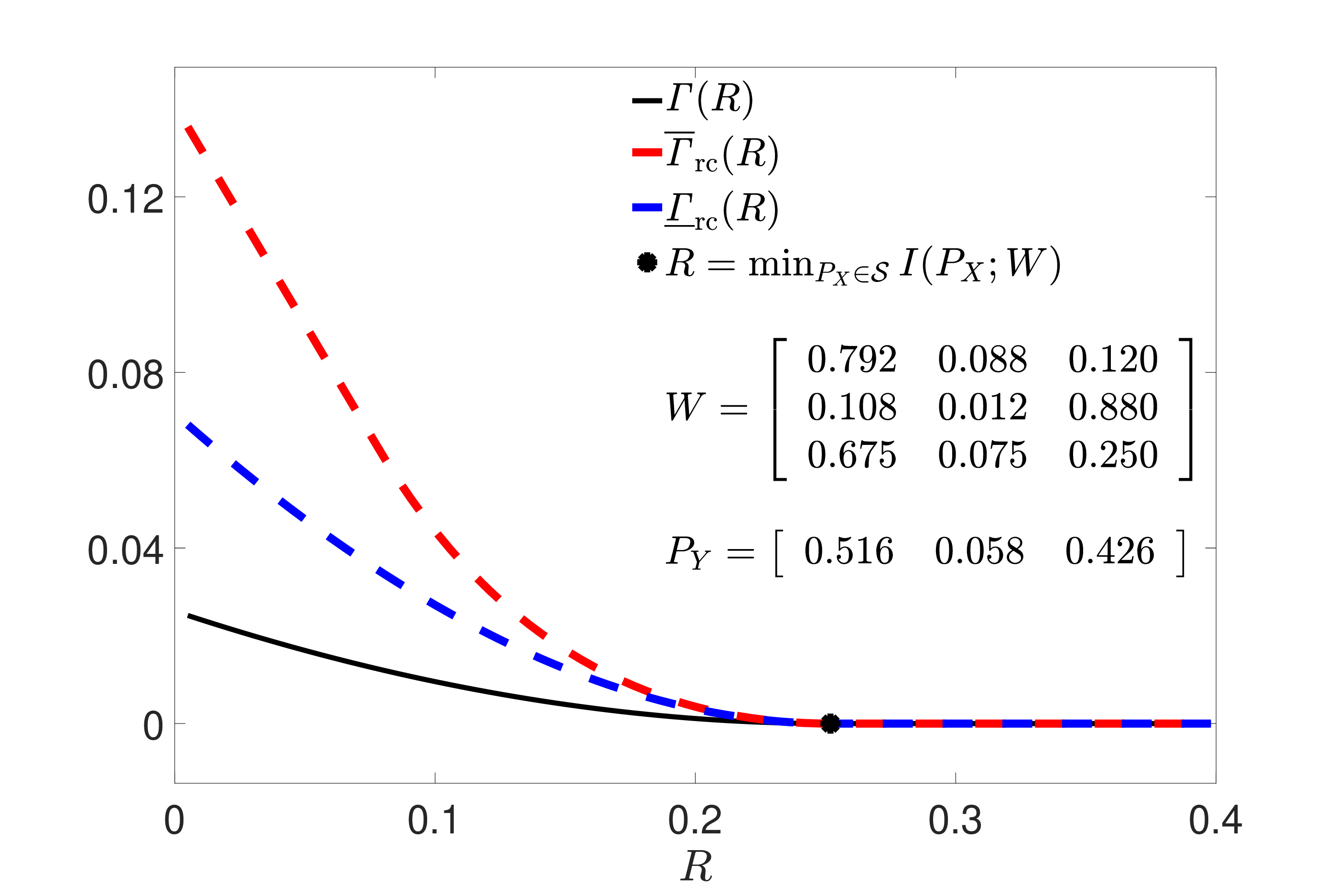}%
    }\\[1ex]
    \subfloat[A fully noiseless ternary channel]{%
        \includegraphics[width=0.49\linewidth]{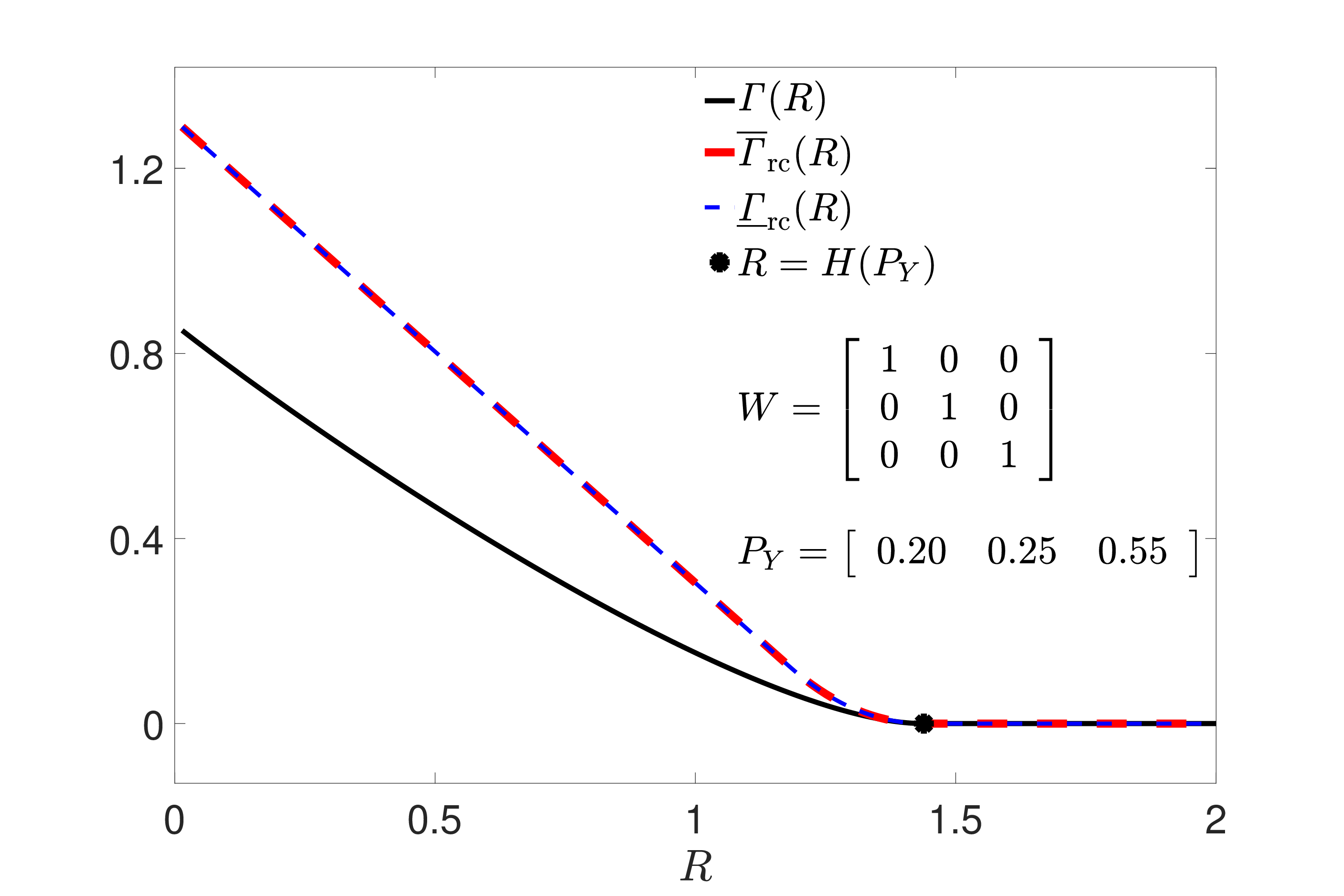}%
    }\hfill
    \subfloat[A hybrid ternary channel]{%
        \includegraphics[width=0.49\linewidth]{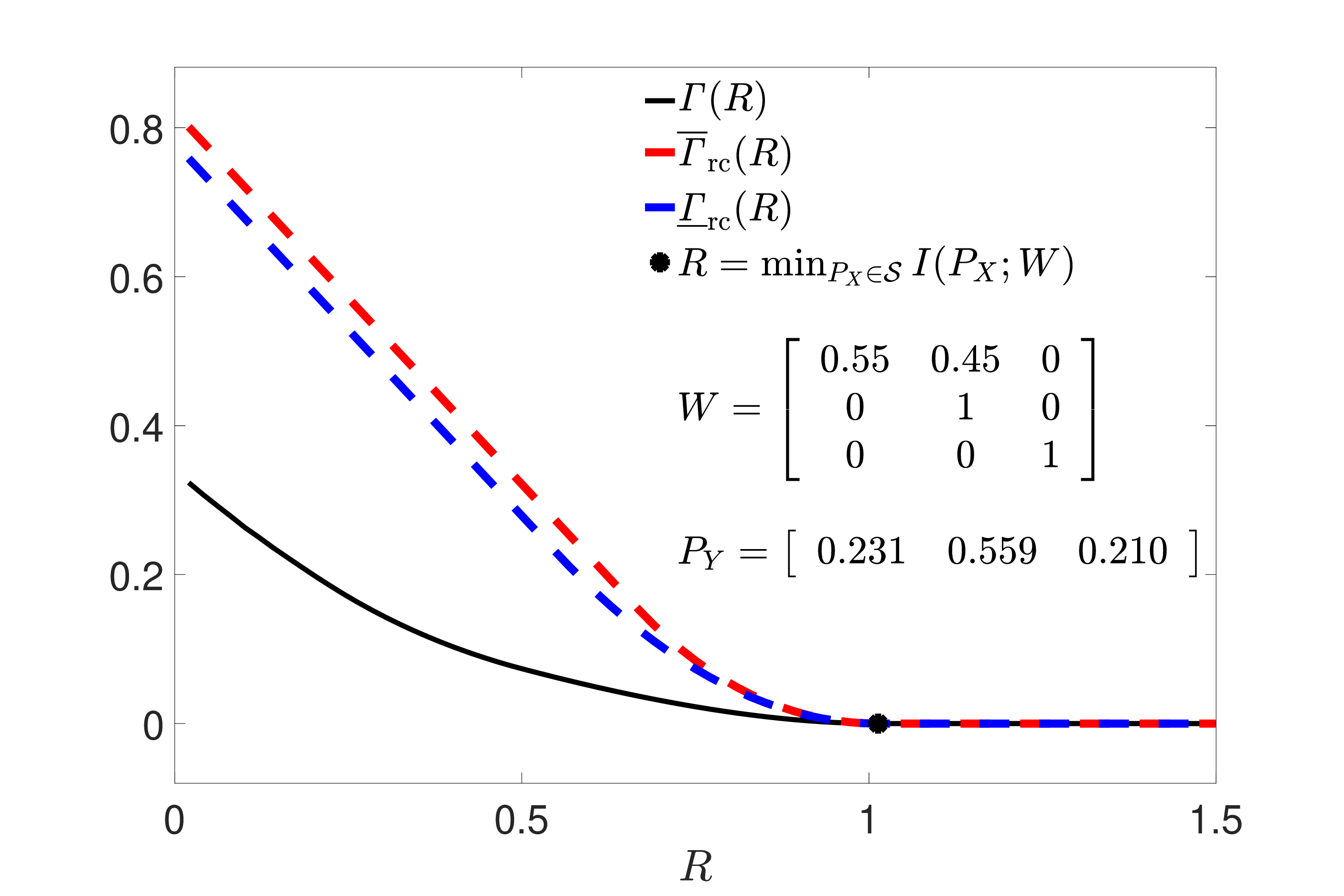}%
    }
    \caption{Examples of the exact strong converse exponent $\Ga(R)$,
    the random coding achievability $\Gu_{\rc}(R)$,
    and the random coding converse $\Gl_{\rc}(R)$.}
    \label{fig-rc}
\end{figure}

\subsection{Results on error exponents for noiseless channels}\label{sec-result-ee-nl}

The previous subsection reveals an intriguing phenomenon: when $R<\minI$, random coding fails to achieve a tight strong converse exponent. A natural question is whether such non-tightness also arises in the error exponent regime for $R\geq\minI$. In this subsection, we show that the answer is yes when $W_{Y|X}$ is noiseless. The following Theorems \ref{thm-nl-Eu}, \ref{thm-nl-El}, \ref{thm-nl-non}, and \ref{thm-nl-ren} summarize our results on error exponents for noiseless channels, where exponents for different formulations are defined in Definition \ref{def-exp}, with the superscript `nl' indicating `noiseless'.

\begin{theorem}[Converse for noiseless channels] 
\label{thm-nl-Eu} 
For noiseless channels under the uniform formulation, we have the following upper bound on $E_\su^\nl(R)$.
\begin{align*}
        E_\su^\nl(R) \leq \Eu^\nl(R) 
        := \ & \min_{Q_Y\in\mP(\mY): D(Q_Y\|P_Y) + H(Q_Y) \geq R } D(Q_Y\|P_Y) \\
        = \ & \max_{\alpha\in(-\infty,0)\cup[1,\infty)} \left\{ \frac{\alpha-1}{\alpha} \left[R - H_{\frac{1}{\alpha}}(P_Y)\right] \right\}.
\end{align*}
\end{theorem}

\begin{theorem}[Achievability for noiseless channels]\label{thm-nl-El} For noiseless channels under the uniform formulation, we have the following lower bound on $E_\su^\nl(R)$.
    \begin{align*}
        E_\su^\nl(R) \geq \El^\nl(R)
        := \ & \min_{Q_Y\in\mP(\mY)} 
        \left\{ D(Q_Y\|P_Y) + \big| R - D(Q_Y\|P_Y) - H(Q_Y) \big|^+ \right\} \\
        = \ & \max_{\alpha\geq1} \left\{ \frac{\alpha-1}{\alpha} \left[R - H_{\frac{1}{\alpha}}(P_Y)\right] \right\},
    \end{align*}
\end{theorem}

\begin{theorem}[Error exponent for non-uniform case]\label{thm-nl-non} For noiseless channels under the non-uniform formulation, the exact error exponent is given by
    \begin{align*}
        E_\sn^\nl(R)
        &= \min_{Q_Y\in\mP(\mY): H(Q_Y) \geq R} D(Q_Y\|P_Y) 
        = \max_{\alpha\geq1} \left\{ (\alpha-1) \left[R - H_{\frac{1}{\alpha}}(P_Y)\right] \right\}.
    \end{align*}
\end{theorem}

\begin{theorem}[Error exponent for $H_{-\infty}$-constrained case]\label{thm-nl-ren} For noiseless channels under the $H_{-\infty}$-constrained formulation, the exact error exponent is given by
    $ E_\sr^\nl(R) = \Eu^\nl(R), $
    where $\Eu^\nl(R)$ is defined in Theorem \ref{thm-nl-Eu}.
\end{theorem}

\begin{figure}[!htbp]
    \centering
    \subfloat[A binary output]{%
        \includegraphics[width=0.49\linewidth]{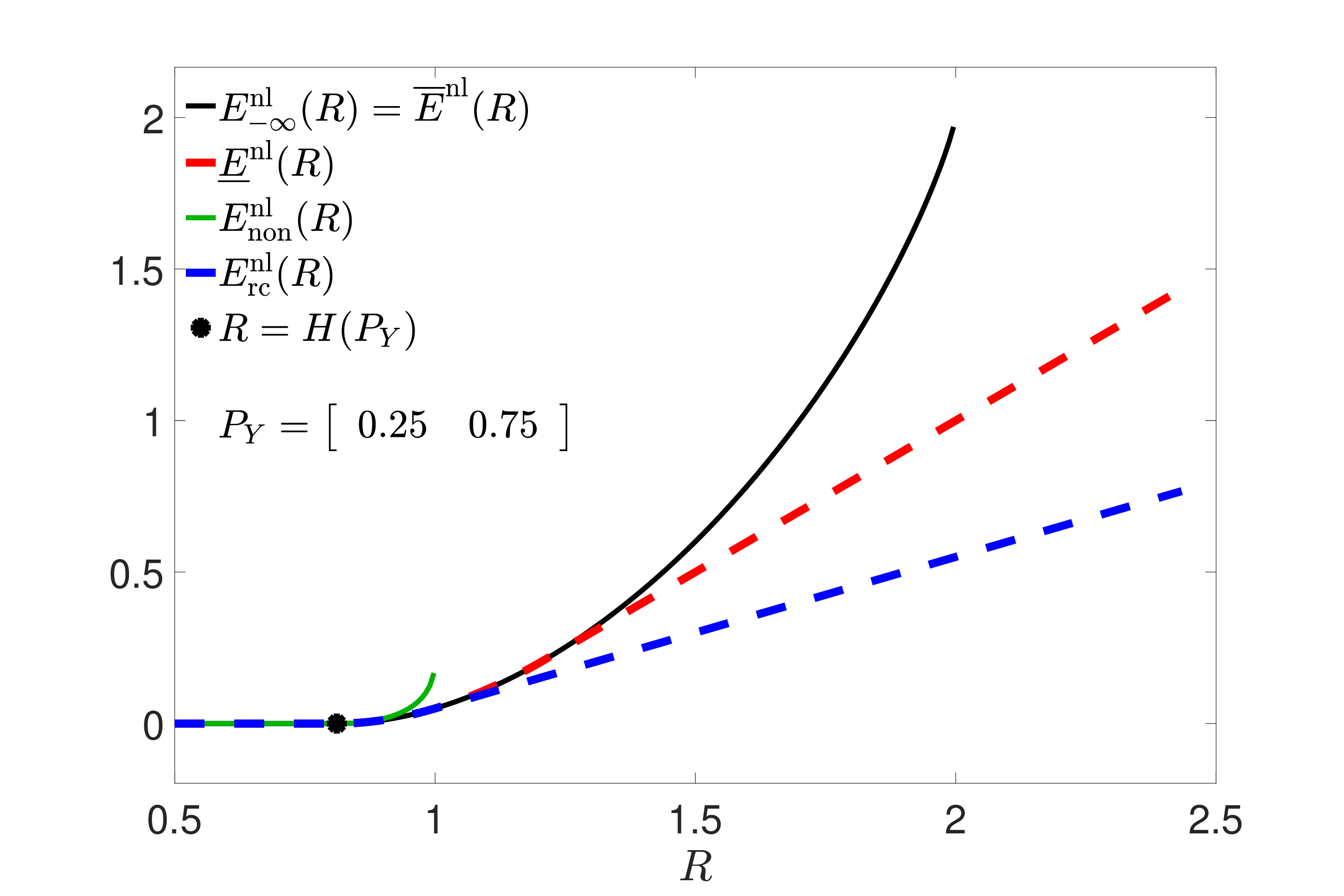}%
    }\hfill
    \subfloat[A ternary output]{%
        \includegraphics[width=0.49\linewidth]{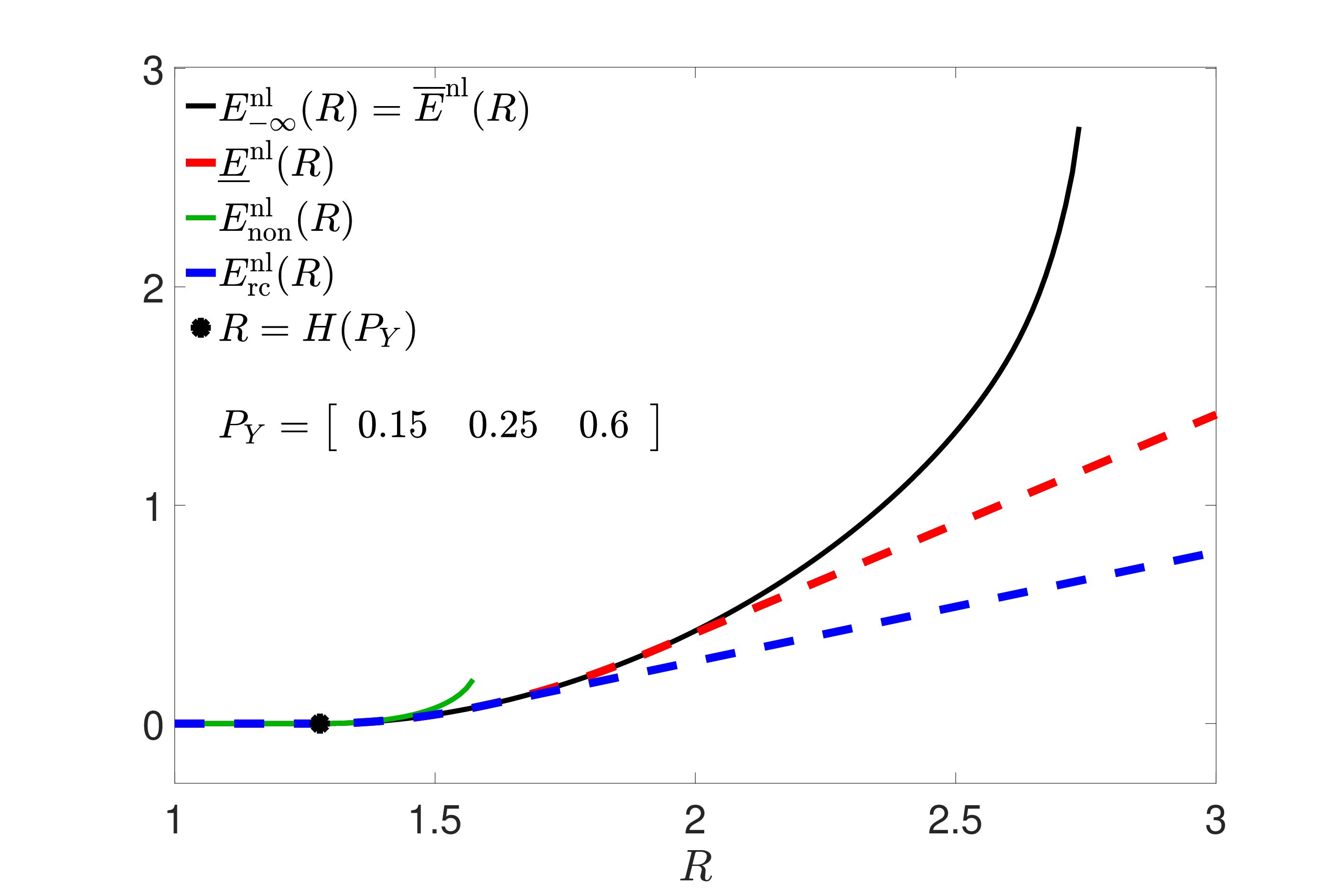}%
    }
    \caption{Examples of error exponents for noiseless channels under the uniform, non-uniform, and $H_{-\infty}$-constrained formulations.}
    \label{fig-nl}
\end{figure}

Theorem \ref{thm-nl-Eu} and Theorem \ref{thm-nl-El} are proved in Section \ref{sec-nl-uni}; Theorem \ref{thm-nl-non} is proved in Section \ref{sec-nl-non}; and Theorem \ref{thm-nl-ren} is proved in Section \ref{sec-nl-ren}. Figure \ref{fig-nl} presents plots of these proposed bounds or exponents as well as the random coding exponent $E_\rc^\nl(R)$ (see Remark \ref{rmk-cuff}) using a binary and a ternary target output distribution, respectively. Under the non-uniform and the $H_{-\infty}$-constrained formulation, an exact error exponent is established. Under the uniform formulation, the error exponent $E_\su^\nl(R)$ lies between the two bounds $\Eu^\nl(R)$ and $\El^\nl(R)$. However, these two bounds overlap in the vicinity of $R = H(P_Y)$, which is given by the following proposition.

\begin{proposition}
\label{prop-nl-exact}
Let $R_\alpha^s$ be the value where the function $\Eu^\nl(\cdot)$ has a tangent line of slope $\alpha$. Then
\begin{align*}
        \El^\nl(R) &= 
        \begin{cases}
            \Eu^\nl(R)     & \text{if } 0\leq R \leq R_1^s, \\
            \Eu^\nl(R_1^s) + R - R_1^s & \text{if } R>R_1^s,
        \end{cases} \\
        E_{\rc}^\nl(R) &= 
        \begin{cases}
            \Eu^\nl(R)     & \text{if } 0\leq R \leq R_{1/2}^s, \\
            \Eu^\nl(R_{1/2}^s) + \dfrac12 \left(R - R_{1/2}^s\right) & \text{if } R>R_{1/2}^s.
        \end{cases}
\end{align*}
Here, $E_{\rc}^\nl(R)$ denotes the random-coding error exponent for noiseless channels and is given in Remark \ref{rmk-cuff}.
\end{proposition}

Proposition \ref{prop-nl-exact} is proved in Appendix \ref{app-nl-exact}. In summary, for $R \in [H(P_Y),R_{1/2}^s]$, all three curves -- $\Eu^\nl(R)$, $\El^\nl(R)$, and $E_\rc^\nl(R)$ -- overlap, implying that random coding is tight. For $R\in[R_{1/2}^s,R_1^s]$, random coding is not tight; nevertheless, $\Eu^\nl(R)$ and $\El^\nl(R)$ still coincide, and an exact exponent exists via the construction of a deterministic code that is provided in the proof of Theorem \ref{thm-nl-El}. For $R>R_1^s$, $\Eu^\nl(R)$ and $\El^\nl(R)$ no longer coincide. It is also noteworthy that $\Eu^\nl(R)$ diverges when $R>H_{-\infty}(P_Y)$; see Lemma \ref{lem-nl-Euinf}.

Assuming uniform messages is conventional in many problems in information theory. Nonetheless, in the soft covering problem, adhering to this uniform formulation highlights an interesting discrepancy between rational and irrational output distributions for noiseless channels. This is because $\tPy$ in \eqref{tPy} can only take rational values as it is a multiple of $1/M$. If $P_Y$ is irrational in some symbols, then approximating $P_Y^n$ by $\tPy$ amounts to a Diophantine approximation, which necessarily incurs nonzero errors for arbitrarily large $M$ according to studies in number theory. We summarize the rational-irrational discrepancy in Theorem \ref{thm-irr} and Theorem \ref{thm-ra}.

\begin{definition}\label{def-ra}
    Given a distribution $P_Y\in\mP(\mY)$, we say that $P_Y$ is rational, denoted $P_Y\in\mathbb{Q}^{\carY}$, if $P_Y(y)\in\mathbb{Q}$ for all $y\in\mY$. Otherwise, we say that $P_Y$ is irrational, denoted $P_Y\notin\mathbb{Q}^{\carY}$.
\end{definition}

\begin{theorem}[Linear converse for irrational $P_Y$]\label{thm-irr}         
    Suppose $W_{Y|X}$ is noiseless.
    Then for almost every $P_Y\notin\mathbb{Q}^{\carY}$ (in the sense of Lebesgue measure on $[0,1]^{\carY}$), we have $E_\su^\nl(R) \leq 2 R$.
\end{theorem}

\begin{theorem}[Infinite achievability for rational $P_Y$]\label{thm-ra}
    Suppose $W_{Y|X}$ is noiseless and $P_Y\in\mathbb{Q}^{\carY}$, say, $P_Y(y) = \frac{A_y}{B_y}$ with coprime $A_y,B_y\in\mathbb{Z}$ for $y\in\mY$. Then
    $E_\su^\nl(R) = \infty$ when $R \geq R_\infty(P_Y) := \log(\lcm(\{B_y\}_{y\in\mY}))$, where $\lcm(\{B_y\}_{y\in\mY})$ refers to the least common multiple of $B_y$'s among all $y\in\mY$.
\end{theorem}

Theorem \ref{thm-irr} and Theorem \ref{thm-ra} are proved in Section \ref{sec-nl-disc}. Evidently, the discrepancy concerns whether an infinite exponent, corresponding to perfect covering with no error, is achievable. For rational $P_Y$'s, this is achievable at high rates, whereas for most irrational $P_Y$'s, it is not. It is noteworthy that the linear converse in Theorem \ref{thm-irr} applies to most irrational probabilities, which are dense in $[0,1]^{\carY}$. However, it may not be possible to claim a universal converse for an arbitrary irrational $P_Y$. Let $\by\in\mY$ be an irrational symbol, i.e., $P_Y(\by)\notin\mathbb{Q}$. If $P_Y(\by)$ is an irrational algebraic number, then the linear converse of $2R$ holds, by virtue of the Thue-Siegel-Roth theorem \cite[Thm.~3.1.4]{queffelec2013diophantine}. On the other hand, $P_Y(\by)$ can also be a `good' irrational number, known as a Liouville number \cite[Def.~3.1.8]{queffelec2013diophantine}, which admits infinitely many integer pairs $K,M\in\mathbb{Z}$ such that $\left|K/M - P_Y(\by) \right|\leq M^{-(2+\epsilon)}$ for all $\epsilon>0$; consequently, it is not clear how to establish a linear converse.

To eliminate such a discrepancy, $\tPy$ must be allowed to also take irrational values; hence, messages cannot be uniformly distributed. Switching from the uniform to the non-uniform formulation indeed removes this discrepancy and yields an exact exponent. However, observing Figure \ref{fig-nl}, the error exponent $E_\sn^\nl(R)$ under the non-uniform formulation deviates from both $\Eu^\nl(R)$ and $\El^\nl(R)$ for all $R$, including in the low-rate regime near $H(P_Y)$. In fact, for noiseless channels, soft covering under the non-uniform formulation is equivalent to lossless source coding (see Lemma \ref{lem-equiv} in Section \ref{sec-nl-non}). In order to remove the rational-irrational discrepancy without substantially altering the feature of soft covering, we proposed the $H_{-\infty}$-constrained formulation in Definition \ref{def-ren}. Its exponent $E_\sr^\nl(R)$ is exact, and is identical to $\Eu^\nl(R)$ by Theorem \ref{thm-nl-ren}. Thus, in the neighborhood of $H(P_Y)$, this new formulation aligns with the uniform formulation.

\subsection{Results on error exponents for noisy channels}\label{sec-result-ee-ny}

More generally, in this subsection, we consider noisy channels. Note that $E_\su^\nl(R)$ in Theorem \ref{thm-nl-El} becomes a straight line of slope 1 for large $R$. Consequently, at high rates, this noiseless achievability can exceed the noisy random-coding exponent $E_\rc(R)$ whose slope is $1/2$, thereby demonstrating a high-rate improvement via input covering over random coding under the uniform formulation for noisy channels. This noisy achievability is summarized in Theorem \ref{thm-ee-achi}. Furthermore, we provide a general converse bound for noisy channels under the non-uniform formulation in Theorem \ref{thm-ee-con}.

\begin{theorem}[High-rate achievability]\label{thm-ee-achi} For noisy channels under the uniform, non-uniform, or $H_{-\infty}$-constrained formulation, we have the following lower bounds on the respective error exponents defined in Definition \ref{def-exp}.
\begin{align*}
    E_\su(R) &\geq \El(R)
        := \max_{P_X\in\mS} \ \max_{\alpha\geq1} \left\{ \frac{\alpha-1}{\alpha} \left[R - H_{\frac{1}{\alpha}}(P_X)\right] \right\}, \\
    E_\sn(R) 
        &\geq \max_{P_X\in\mS} \ \max_{\alpha\geq1} \left\{ (\alpha-1) \left[R - H_{\frac{1}{\alpha}}(P_X)\right] \right\}, \\
    E_\sr(R)
        &\geq \max_{P_X\in\mS} \ \max_{\alpha\in(-\infty,0)\cup[1,\infty)} \left\{ \frac{\alpha-1}{\alpha} \left[R - H_{\frac{1}{\alpha}}(P_X)\right] \right\}.
\end{align*}
\end{theorem}

\begin{theorem}[Converse]
\label{thm-ee-con} 
For noisy channels under the non-uniform formulation, we have the following upper bound on $E_\sn(R)$ defined in Definition \ref{def-exp}.
\begin{align*}
    E_\sn(R) \le \Eu(R) 
      := \ & \max_{P_X\in\mS} \ \min_{V\in\mP(\mY|\mX): \iota(P_X V_{Y|X}) \geq R} D(V\|W|P_X) \\
      = \ & \max_{P_X\in\mS} \ \max_{\alpha\ge1} \left\{ (\alpha-1) \left[ R - \mathbb{E}_{P_X} D_\alpha(W_{Y|X}\|P_Y) \right] \right\}.
\end{align*}
\end{theorem}

Theorem \ref{thm-ee-achi} is proved in Section \ref{sec-ee-ny-achi} and Theorem \ref{thm-ee-con} is proved in Section \ref{sec-ee-ny-con}. Figure \ref{fig-ny} exhibits examples of the proposed bounds for the uniform formulation, as well as the random coding error exponent $E_\rc(R)$ in Remark \ref{rmk-cuff}. Immediately following Theorem \ref{thm-ra} and Theorem \ref{thm-ee-achi}, we obtain Corollary \ref{cor-ra-ny} below, which demonstrates an achievable infinite error exponent for noisy channels, provided that there exists some rational input distributions $P_X\in\mS$ and the rate is sufficiently large. 

\begin{corollary}[Infinite achievability for noisy channels]\label{cor-ra-ny}
    Suppose some $P_X\in\mS$ is rational, say, $P_X(x) = \frac{A_x}{B_x}$ with coprime $A_x,B_x\in\mathbb{Z}$ for $x\in\mX$. Then
    $E_\su(R) = \infty$ when $R \geq R_\infty(P_X) := \log(\lcm(\{B_x\}_{x\in\mX}))$.
\end{corollary}

\begin{figure}[!htbp]
    \centering
    \subfloat[A fully noisy binary channel]{%
        \includegraphics[width=0.49\linewidth]{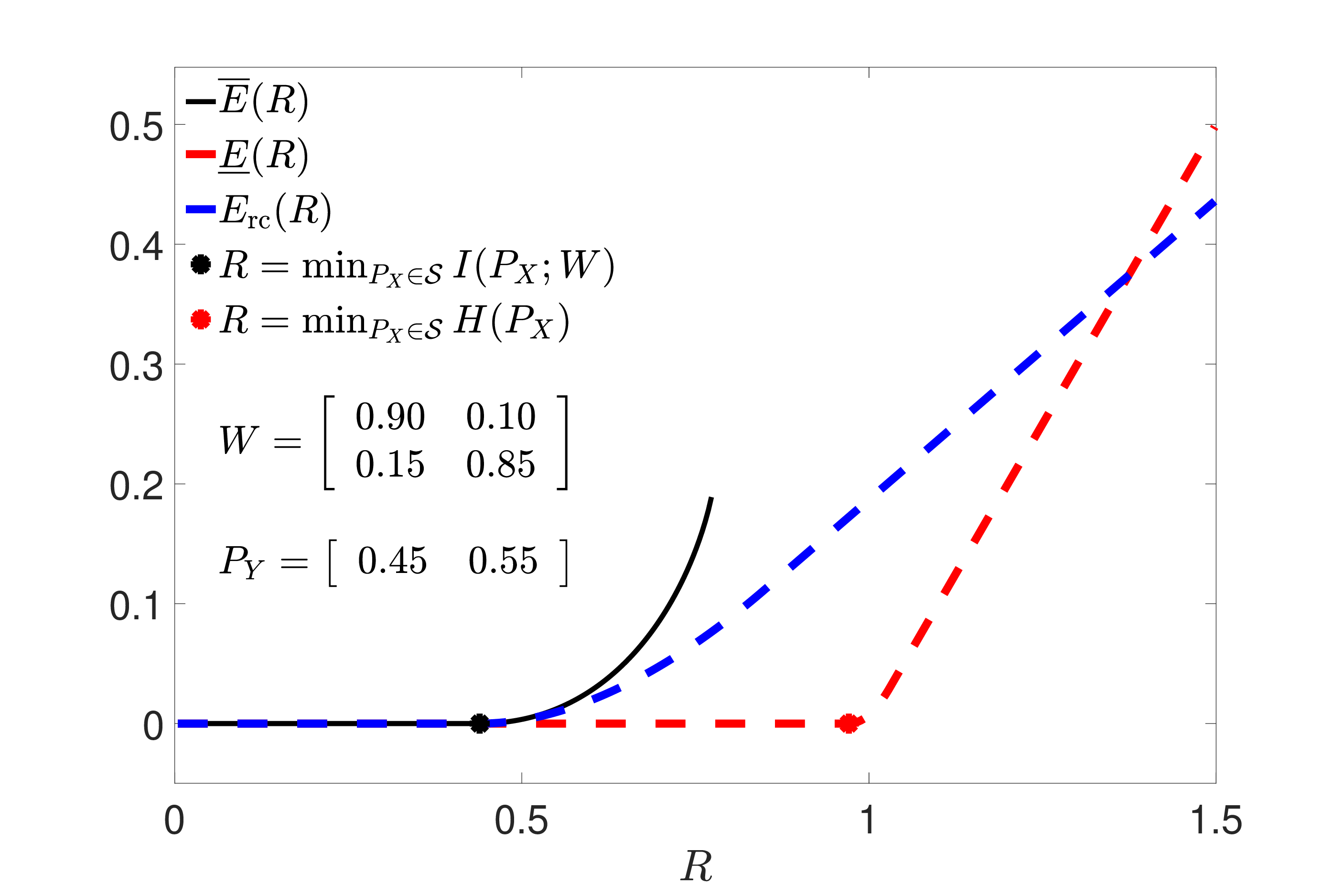}%
    }\hfill
    \subfloat[A fully noisy ternary channel]{%
        \includegraphics[width=0.49\linewidth]{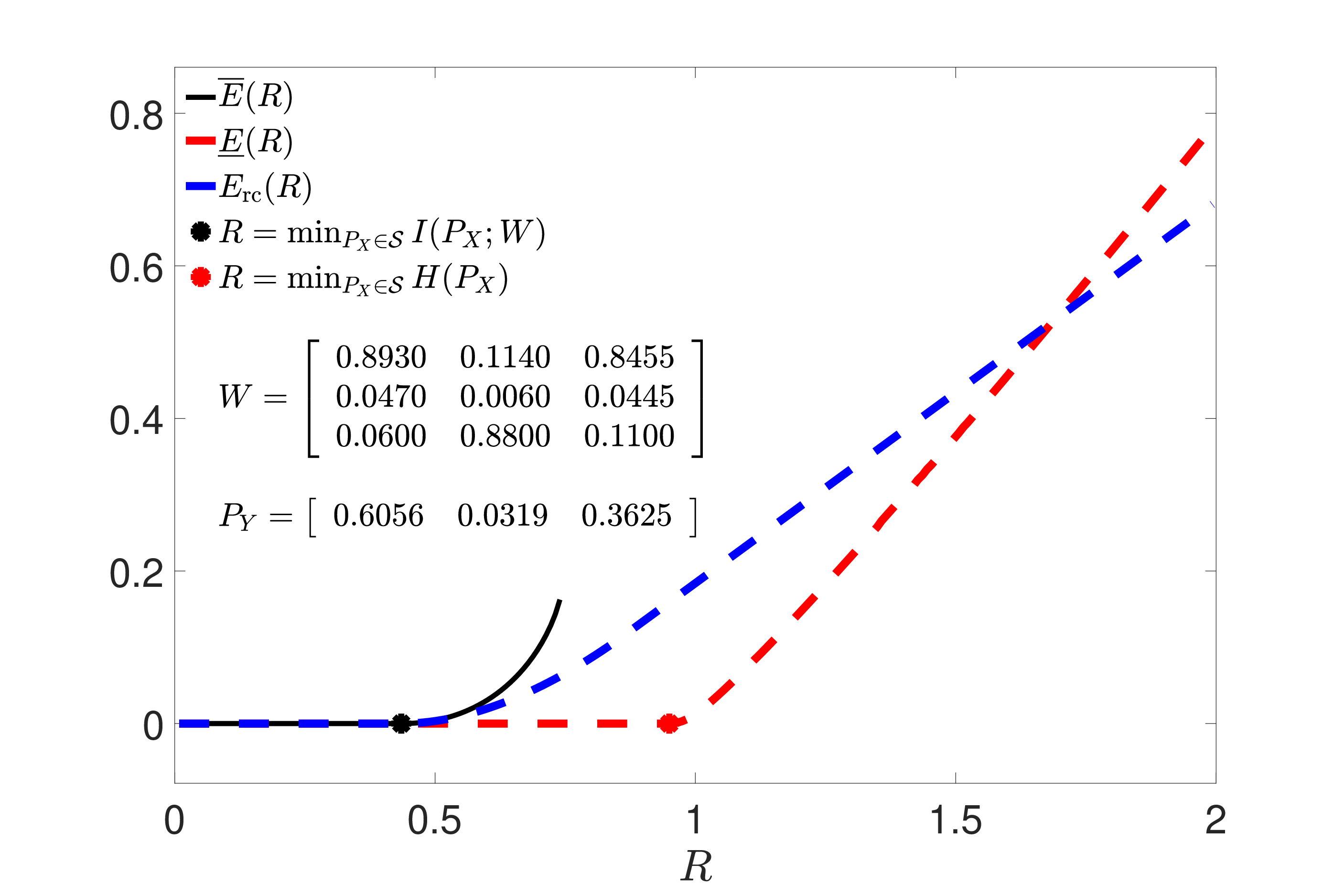}%
    }
    \caption{Examples of the converse bound $\Eu(R)$, the achievability bound $\El(R)$,
    and the random-coding error exponent $E_{\rc}(R)$ for noisy channels.}
    \label{fig-ny}
\end{figure}

\subsection{Summary of main results} \label{sec-result-summary}

In Table \ref{table-summary}, we summarize the aforementioned results for the uniform, non-uniform, and $H_{-\infty}$-constrained formulations, in conjunction with Remark \ref{rmk-relation}.

\begin{table}[!htbp]
\centering
\caption{Summary of Results}
\label{table-summary}

\newcolumntype{C}[1]{>{\centering\arraybackslash}m{#1}}

\renewcommand{\arraystretch}{1.25}
\begin{tabular}{|C{2.25cm}|C{6.25cm}|C{6.25cm}|}
\hline

Formulation & 
\multicolumn{2}{c|}{Strong Converse Exponent} \\ \hline \hline

Uniform, Non-uniform, $H_{-\infty}$-constrained & \multicolumn{2}{c|}{$ \displaystyle \max_{\alpha,\beta\in[0,1]}
    \left[ J_{\alpha,\beta}\left(W_{Y|X}\|P_Y\right) + \beta(\alpha-1) R \right] $} \\ 

\specialrule{3pt}{0pt}{0pt}

\multirow{2}{*}{Formulation} & \multicolumn{2}{c|}{Error Exponent for Noiseless Channels} \\ \cline{2-3}

& Achievability & Converse \\ \cline{1-3} 

\hline \hline

\multirow{2}{*}[-3pt]{\makecell{Uniform}} & 
\addstackgap[5pt]{ $\displaystyle \max_{\alpha\geq1} 
    \left\{ \frac{\alpha-1}{\alpha} \left[R - H_{\frac{1}{\alpha}}(P_Y)\right] \right\}$ }  & 
\addstackgap[5pt]{ $\displaystyle \max_{\alpha\in(-\infty,0)\cup[1,\infty)} 
    \left\{ \frac{\alpha-1}{\alpha} \left[R - H_{\frac{1}{\alpha}}(P_Y)\right] \right\}$ } \\ \cline{2-3}

& $\infty$ if $P_Y\in\mathbb{Q}^{\carY}$ and $R\geq R_\infty(P_Y)$ & 
$ 2 R$ for almost every $P_Y\notin\mathbb{Q}^{\carY}$\\ \cline{1-3}

\multirow{1}{*}[-3pt]{Non-uniform} & 
\multicolumn{2}{c|}{
    \addstackgap[5pt]{$\displaystyle \max_{\alpha\geq1} \left\{ (\alpha-1) \left[R - H_{\frac{1}{\alpha}}(P_Y)\right] \right\}$}} \\ \cline{1-3}

\multirow{1}{*}[-3pt]{$H_{-\infty}$-constrained} & 
\multicolumn{2}{c|}{
    \addstackgap[5pt]{$ \displaystyle \max_{\alpha\in(-\infty,0)\cup[1,\infty)} 
    \left\{ \frac{\alpha-1}{\alpha} \left[R - H_{\frac{1}{\alpha}}(P_Y)\right] \right\} $}} \\ \cline{1-3}

\specialrule{3pt}{0pt}{0pt}

\multirow{2}{*}{Formulation} & \multicolumn{2}{c|}{Error Exponent for Noisy Channels} \\ \cline{2-3}

& High-rate Achievability & Converse \\ \cline{1-3} 

\hline\hline
 
\multirow{2}{*}[-3pt]{\makecell{Uniform}} & 
\addstackgap[5pt]{$\displaystyle \max_{P_X\in\mS} \ \max_{\alpha\geq1} 
    \left\{ \frac{\alpha-1}{\alpha} \left[R - H_{\frac{1}{\alpha}}(P_X)\right] \right\} $} & 
\multirow[c]{4}{*}[-12pt]{%
    \makecell[c]{ 
        $\displaystyle \max_{P_X \in \mathcal{S}} \ \max_{\alpha \ge 1} \left\{ (\alpha-1) \left[ R - \mathbb{E}_{P_X} D_{\alpha}(W_{Y|X}\|P_Y) \right] \right\}$} } \\ \cline{2-2}

& $\infty$ if $\exists \ P_X\in\mS\cap\mathbb{Q}^{\carX}$ and $R\geq R_\infty(P_X)$ & \\ \cline{1-2}

\multirow{1}{*}[-3pt]{Non-uniform} & 
\addstackgap[5pt]{$ \displaystyle \max_{P_X\in\mS} \ \max_{\alpha\geq1} \left\{ (\alpha-1) \left[R - H_{\frac{1}{\alpha}}(P_X)\right] \right\} $} & \\ \cline{1-2}

\multirow{1}{*}[-3pt]{$H_{-\infty}$-constrained} & 
\addstackgap[5pt]{$ \displaystyle \max_{P_X\in\mS} \ \max_{\alpha\in(-\infty,0)\cup[1,\infty)} 
    \left\{ \frac{\alpha-1}{\alpha} \left[R - H_{\frac{1}{\alpha}}(P_X)\right] \right\} $} & \\ \cline{1-2}
\hline

\end{tabular}
\end{table}

\section{Strong Converse Exponent}
\label{sec-sc}
In this section, we prove Theorem \ref{thm-sc}. The sketch of the proof is as follows. In Section \ref{sec-sc-con}, we establish a lower bound for $\Ga_\sn(R) $, serving as a converse result; while Section \ref{sec-sc-achi} presents an achievability result that provides an upper bound for $\Ga_\su(R)$. Both bounds are expressed in variational forms in terms of the optimization over joint distributions. We convey them into dual representations in Section \ref{sec-sc-dual} and show that they coincide when $R < \minI$. Since $\Ga_\sn(R)\leq \Ga_\su(R)$, this coincidence characterizes an exact exponent for both uniform and non-uniform formulations. Additionally, in Appendix \ref{app-arimoto}, we provide an alternative proof of the converse part of Theorem \ref{thm-sc} only for the uniform formulation following Arimoto's techniques in \cite{arimoto1973converse}.

\subsection{A lower bound for the strong converse exponent}
\label{sec-sc-con}
In the following, we show a lower bound for $\Ga_\sn(R)$ for the non-uniform formulation, denoted by $\Gl(R)$. 

\begin{proposition}[Converse]
\label{prop-sc-con}
The strong converse exponent $\Ga_\sn(R)$ in the non-uniform formulation can be bounded from below as follows:
    $$ \Ga_\sn(R) \geq \Gl(R) 
    := \min_{Q_X\in\mP(\mX)} \ \min_{s\geq0} \ \max\big\{s, \Ga(s,Q_X,R)\big\}, $$
    where
    \begin{equation}\label{Ga-s-Qx-R}
        \Ga(s,Q_X,R) := \min_{V\in\mP(\mY|\mX): D(V\|W|Q_X)\leq s} \left[ D(Q_X V\|P_Y) + \big|I(Q_X;V) - R\big|^+ \right].
    \end{equation}
\end{proposition}

\begin{proof}
We use a hypothesis testing perspective,
and apply the following well-known property of the total variation (e.g.,\cite[Def.~2.24]{moser2019advanced}) in terms of decision regions:
\begin{equation}\label{setA}
    \frac12 \left\|\tPyQ - P_Y^n \right\|_1 
    \geq \left| \tPyQ(\mA) - P_Y^n(\mA) \right|, \quad \forall \mA \subseteq \mY^n.
\end{equation}
We take a collection of conditional types around each codeword to create a set $\mathcal{A}$ that could be used to discriminate $\tPy$ from $P_Y^n$, as follows:
\begin{equation}\label{setA-def}
    \mA = \bigcup_{Q_X\in\mP_n(\mX)} \
    \bigcup_{i:X^n(i)\in \mT_{Q_X}} \ 
    \bigsqcup_{V\in\mV_n(Q_X)} \mT_V(X^n(i)).
\end{equation}
That is, restrict to the output type classes whose types are in the form of $Q_Y = Q_X V$, where $Q_X$ is the type of some codeword and $V\in\mV_n(Q_X)$. Here $\mV_n(Q_X)$ can be any subset of $\mP_n(\mY|Q_X)$. Later, we will choose an appropriate subset to yield the optimal bound. When a codeword $X^n(i)$ is fixed, so is its type. Then the $V$-shells $\mT_V(X^n(i))$ for different $V$ must be disjoint. We thereby write disjoint union $\bigsqcup_{V\in\mV_n(Q_X)}$ in $\mA$.

Consider a codeword $X^n(i)$ whose type is $Q_X$, i.e., $X^n(i)\in \mT_{Q_X}$ for some $Q_X\in\mP_n(\mX)$. We have
\begin{align*}
    W_{Y|X}^n(\mA|X^n(i)) 
    &= W_{Y|X}^n\left(\bigcup_{Q_X'\in\mP_n(\mX)} \
    \bigcup_{j:X^n(j)\in \mT_{Q_X'}} \ 
    \bigsqcup_{V\in\mV_n(Q_X')} \mT_V(X^n(j)) \bigg|X^n(i)\right) \\
    &\overset{a}{\geq} W_{Y|X}^n \left( \bigsqcup_{V\in\mV_n(Q_X)} \mT_V(X^n(i)) \bigg| X^n(i) \right) \\
    &\overset{b}{=} 1 - W_{Y|X}^n \left( \bigsqcup_{V\notin\mV_n(Q_X)} \mT_V(X^n(i)) \bigg| X^n(i) \right) \\
    &\overset{c}{\geq} 1 - \sum_{V\notin\mV_n(Q_X)} 2^{-nD(V\|W|Q_X)} \\
    &\overset{d}{\geq} 1 - (n+1)^{|\mX||\mY|} \max_{V\notin\mV_n(Q_X)} 2^{-nD(V\|W|Q_X)} \\
    &= 1 - (n+1)^{|\mX||\mY|} \exp_2 \left\{ -n \min_{V\notin\mV_n(Q_X)} D(V\|W|Q_X) \right\},
\end{align*}
where in $(a)$ we restrict to the subset in which $j=i$; $(b)$ follows from the fact that $\bigsqcup_{V\in\mP_n(\mY|Q_X)} \mT_V(X^n(i)) = \mY^n$; $(c)$ and $(d)$ follow from properties of types. By averaging over all codewords, we obtain that 
\begin{align*}
    \tPyQ(\mA) 
    &=  \sum_{i=1}^M  q(i) \sum_{Q_X\in\mP_n(\mX)} \mathbbm{1}\{X^n(i)\in\mT_{Q_X}\} \ W_{Y|X}^n(\mA|X^n(i)) \\
    &\geq 1 - (n+1)^{|\mX||\mY|} \sum_{Q_X\in\mP_n(\mX)} \sum_{i=1}^M q(i) \mathbbm{1}\{X^n(i)\in\mT_{Q_X}\} \exp_2 \left\{ -n \min_{V\notin\mV_n(Q_X)} D(V\|W|Q_X) \right\} \\
    &\overset{a}{\geq} 1 - (n+1)^{|\mX|(|\mY|+1)} \max_{Q_X\in\mP_n(\mX)} \sum_{i=1}^M q(i) \mathbbm{1}\{X^n(i)\in\mT_{Q_X}\} \exp_2 \left\{ -n \min_{V\notin\mV_n(Q_X)} D(V\|W|Q_X) \right\} \\
    &\overset{b}{\geq} 1 - (n+1)^{|\mX|(|\mY|+1)} \max_{Q_X\in\mP_n(\mX)} \exp_2 \left\{ -n \min_{V\notin\mV_n(Q_X)} D(V\|W|Q_X) \right\} \\
    &= 1 - (n+1)^{|\mX|(|\mY|+1)} \exp_2 \left\{ -n \min_{Q_X\in\mP_n(\mX)} \ \min_{V\notin\mV_n(Q_X)} D(V\|W|Q_X) \right\},
\end{align*}
where $(a)$ follows from property of types and $(b)$ from $\sum_{i=1}^M q(i) \mathbbm{1}\{X^n(i)\in\mT_{Q_X}\} \leq 1$. On the other hand, 
\begin{align*}
    P_Y^n(\mA) 
    &\overset{a}{=} \sum_{Q_Y\in\mP_n(\mY)} \left|\mA\cap\mT_{Q_Y}\right| \ 2^{-n[D(Q_Y\|P_Y) + H(Q_Y)]} \\
    &= \sum_{Q_Y\in\mP_n(\mY)} \left| \left(\bigcup_{Q_X\in\mP_n(\mX)} \
        \bigcup_{i:X^n(i)\in \mT_{Q_X}} \ 
        \bigsqcup_{V\in\mV_n(Q_X)} \mT_V(X^n(i)) \right) \bigcap \mT_{Q_Y} \right| 2^{-n[D(Q_Y\|P_Y) + H(Q_Y)]} \\
    &= \sum_{Q_Y\in\mP_n(\mY)} \left| \bigcup_{Q_X\in\mP_n(\mX)} \
        \bigcup_{i:X^n(i)\in \mT_{Q_X}} \ 
        \bigsqcup_{V\in\mV_n(Q_X)} \mT_V(X^n(i)) \cap \mT_{Q_Y} \right| 2^{-n[D(Q_Y\|P_Y) + H(Q_Y)]} \\
    &\overset{b}{\leq} \sum_{Q_X\in\mP_n(\mX)} \ \sum_{V\in\mV_n(Q_X)} 
        \left|\bigcup_{i:X^n(i)\in \mT_{Q_X}} \mT_V(X^n(i)) \cap \mT_{Q_X V}\right| 2^{-n[D(Q_X V\|P_Y) + H(Q_X V)]} \\
    &\leq \sum_{Q_X\in\mP_n(\mX)} \ \sum_{V\in\mV_n(Q_X)} 
        \min\left\{\sum_{i=1}^M \mathbbm{1}\{X^n(i)\in \mT_{Q_X}\} \big|\mT_V(X^n(i))\big|, \big|\mT_{Q_X V}\big|\right\}
        2^{-n[D(Q_X V\|P_Y) + H(Q_X V)]} \\
    &\overset{c}{\leq} \sum_{Q_X\in\mP_n(\mX)} \ \sum_{V\in\mV_n(Q_X)} \min\left\{M\cdot 2^{nH(V|Q_X)}, 2^{nH(Q_X V)}\right\} 2^{-n[D(Q_X V\|P_Y) + H(Q_X V)]} \\
    &= \sum_{Q_X\in\mP_n(\mX)} \ \sum_{V\in\mV_n(Q_X)} 2^{-n\left[D(Q_X V\|P_Y) + |I(Q_X;V)-R|^+\right]} \\
    &\overset{d}{\leq} (n+1)^{|\mX|(|\mY|+1)} \max_{Q_X\in\mP_n(\mX)} \ \max_{V\in\mV_n(Q_X)} 2^{-n\left[D(Q_X V\|P_Y) + |I(Q_X;V)-R|^+\right]} \\
    &= (n+1)^{|\mX|(|\mY|+1)} \exp_2 \left\{ -n \min_{Q_X\in\mP_n(\mX)} \ \min_{V\in\mV_n(Q_X)}  \left[ D(Q_X V\|P_Y) + \big|I(Q_X;V) - R\big|^+ \right] \right\},
\end{align*}
where $(a)$ follows from property of types; in $(b)$, the summation $\sum_{Q_Y\in\mP_n(\mY)}$ is missing because the non-emptiness of $\mT_V(X^n(i)) \cap \mT_{Q_Y}$ implies that $Q_Y$ can only take the unique value $Q_Y = Q_X V$; $(c)$ and $(d)$ follow from properties of types.
Inserting the above expressions into \eqref{setA}, we obtain that 
$$ \frac12 \left\|\tPyQ - P_Y^n \right\|_1 \geq 1 - 2(n+1)^{|\mX|(|\mY|+1)} 2^{-n\Gl(R,n)}, $$
where the exponent is
\begin{align*}
    \Gl(R,n) &= \min\left\{ 
    \min_{Q_X\in\mP_n(\mX)} \min_{V\notin\mV_n(Q_X)} D(V\|W|Q_X), \ 
    \min_{Q_X\in\mP_n(\mX)} \min_{V\in\mV_n(Q_X)} \left[ D(Q_X V\|P_Y) + \big|I(Q_X;V) - R\big|^+ \right]\right\} \\
    &= \min_{Q_X\in\mP_n(\mX)} 
    \min\left\{ \min_{V\notin\mV_n(Q_X)} D(V\|W|Q_X), \ 
    \min_{V\in\mV_n(Q_X)}  \left[ D(Q_X V\|P_Y) + \big|I(Q_X;V) - R\big|^+ \right]\right\}.
\end{align*}

It remains to choose a suitable $\mV_n(Q_X)$. We choose it to be the relative-entropy ball centered at $W$:
$$ \mV_n(Q_X) = \left\{ V\in\mP_n(\mY|\mX): D(V\|W|Q_X) \leq s(Q_X) \right\}, $$
where the quantity $s(Q_X)$, serving as the ball radius, can be chosen as an arbitrary function of $Q_X$, since any set $\mA$ of the form \eqref{setA-def} can satisfy \eqref{setA}.
As $n\to\infty$, $\mP_n(\mX)$ becomes dense in $\mP(\mX)$ and $\mV_n(Q_X)$ boils down to
$$ \mV(Q_X) := \left\{ V\in\mP(\mY|\mX): D(V\|W|Q_X) \leq s(Q_X)\right\}. $$
Consequently, the exponent $\Gl(R,n)$ converges as follows:
\begin{equation}\label{Ednl-minmin}
    \liminf_{n\to\infty} \Gl(R,n) 
    = \min_{Q_X\in\mP(\mX)} 
    \min\left\{ \inf_{V\notin\mV(Q_X)} D(V\|W|Q_X), \ 
    \min_{V\in\mV(Q_X)} \left[ D(Q_X V\|P_Y) + \big|I(Q_X;V) - R\big|^+ \right]\right\}.
\end{equation}
Since the choice of $s(Q_X)$ is free and we seek a lower bound, we choose $s(Q_X)$ to maximize the expression in \eqref{Ednl-minmin} for each $Q_X$, thereby obtaining the tightest possible lower bound. Specifically, given $Q_X\in\mP(\mX)$, we write 
$$ \mV(s,Q_X) := \left\{ V\in\mP(\mY|\mX): D(V\|W|Q_X) \leq s\right\} $$ and set
\begin{align*}
    s(Q_X) &= \argmax_{s\geq0} \ 
    \min\left\{ \inf_{V\notin\mV(s,Q_X)} D(V\|W|Q_X), \ 
    \min_{V\in\mV(s,Q_X)} \left[ D(Q_X V\|P_Y) + \big|I(Q_X;V) - R\big|^+ \right]\right\} \\
    &= \argmax_{s\geq0} \
    \min\left\{ \inf_{V\notin\mV(s,Q_X)} D(V\|W|Q_X), \ \Ga(s,Q_X,R)\right\},
\end{align*}
where $\Ga(s,Q_X,R)$ is defined in \eqref{Ga-s-Qx-R}. Consequently, \eqref{Ednl-minmin} yields a lower bound
\begin{equation}
\label{ElR-1}
    \Ga_\sn(R) \geq 
    \Gl(R) := \min_{Q_X\in\mP(\mX)} \ \max_{s\geq0} \ 
    \min\left\{ \inf_{V\notin\mV(s,Q_X)} D(V\|W|Q_X), \ \Ga(s,Q_X,R)\right\}.
\end{equation}

In \eqref{ElR-1}, it might seem natural to write $\inf_{V\notin\mV(s,Q_X)} D(V\|W|Q_X) = s$ since $\mV(s,Q_X)$ describes a relative entropy ball. However, this is true only if $V$ can be those conditional distributions such that $V\ll W$ and $V\neq W$ under $Q_X$. To elaborate, our discussions and the derived bounds should be applicable to all channels. Without loss of generality, it is possible that the channel $W_{Y|X}$ contains both noiseless and noisy input symbols, meaning that $W_{Y|X}(y|x) = 1$ for some $x,y$ and $W_{Y|X}(y|x) < 1$ for others. If $Q_X$ is supported only on those noiseless input symbols, then $W$ is deterministic and hence for any finite $s\geq0$, 
\[
  D(V\|W|Q_X) 
   = \begin{cases}
        \infty & \text{if } D(V\|W|Q_X)>s, \\
        0      & \text{if } D(V\|W|Q_X)\leq s.
     \end{cases}
\]
The case $D(V\|W|Q_X) = 0$ implies that $V|_{\supp(Q_X)} = W$, i.e., $V$ is also noiseless when restricted to $x\in\supp(Q_X)$. As a result, $\Ga(s,Q_X,R) = D(Q_X W\|P_Y) + \big|H(Q_X W) - R\big|^+$, which is a constant independent of $s$, as illustrated in Figure \ref{fig-Gs}\subref{fig-Gs-a}. Then in \eqref{ElR-1} we have
\begin{align*}
     \max_{s\geq0} 
    \min\left\{ \inf_{V\notin\mV(s,Q_X)} D(V\|W|Q_X), \ \Ga(s,Q_X,R)\right\} 
    = \ & \max_{s\geq0} 
    \min\left\{ \infty, D(Q_X W\|P_Y) + \big|H(Q_X W) - R\big|^+\right\} \\
    = \ & D(Q_X W\|P_Y) + \big|H(Q_X W) - R\big|^+ \\
    \overset{a}{=} \ & \max_{s\geq0} \min\left\{s, D(Q_X W\|P_Y) + \big|H(Q_X W) - R\big|^+\right\} \\
    = \ & \max_{s\geq0} 
    \min\left\{s, \Ga(s,Q_X,R)\right\},
\end{align*}
where $(a)$ can be observed from Figure \ref{fig-Gs}\subref{fig-Gs-a}. Interestingly, even though $\inf_{V\notin\mV(s,Q_X)} D(V\|W|Q_X) = \infty \neq s$, in \eqref{ElR-1} it is still valid to write 
\begin{equation}\label{maxs}
    \max_{s\geq0} 
    \min\left\{ \inf_{V\notin\mV(s,Q_X)} D(V\|W|Q_X), \ \Ga(s,Q_X,R)\right\} 
    = \max_{s\geq0} 
    \min\left\{s, \Ga(s,Q_X,R)\right\}.
\end{equation}

If $Q_X$ has support on at least one noisy input symbol, $D(\cdot\|W|Q_X)$ can thereby take continuous values. In this case, writing \eqref{maxs} is explicit and correct. According to Lemma \ref{lem-fs} in Appendix \ref{app-fs}, $\Ga(s,Q_X,R)$ is non-increasing in $s$: it decreases over a certain interval and then remains constant. Therefore, $\max_{s\geq0} \min\left\{s, \Ga(s,Q_X,R)\right\}$ is attained at a unique fixed point where $s^* = \Ga(s^*,Q_X,R)$, i.e., the intersection point of the curve $\Ga(s,Q_X,R)$ and the straight line $s$. Such an intersection can occur either in the decreasing regime of $\Ga(s,Q_X,R)$ (Figure \ref{fig-Gs}\subref{fig-Gs-b}) or in the constant regime (Figure \ref{fig-Gs}\subref{fig-Gs-c}). 

To summarize, for any $Q_X$ supported on $\mX$, \eqref{maxs} is always true. Furthermore, since \eqref{maxs} is achieved at the intersection point $s^*$, we can equivalently write 
\begin{equation}\label{maxs=mins}
    \max_{s\geq0} \min\left\{s, \Ga(s,Q_X,R)\right\} 
    = \min_{s\geq0} \max\left\{s, \Ga(s,Q_X,R)\right\}
\end{equation}
from Figure \ref{fig-Gs}. The reason for doing so is that the right-hand side of \eqref{maxs=mins} facilitates a convenient derivation from variational forms to dual forms in Section \ref{sec-sc-dual}. Substituting \eqref{maxs} and \eqref{maxs=mins} into \eqref{ElR-1} completes the proof.
\end{proof}

\begin{figure}[!htbp]
    \centering
    \subfloat[$\supp(Q_X)\subseteq\mX_l$]{%
        \includegraphics[width=0.32\linewidth]{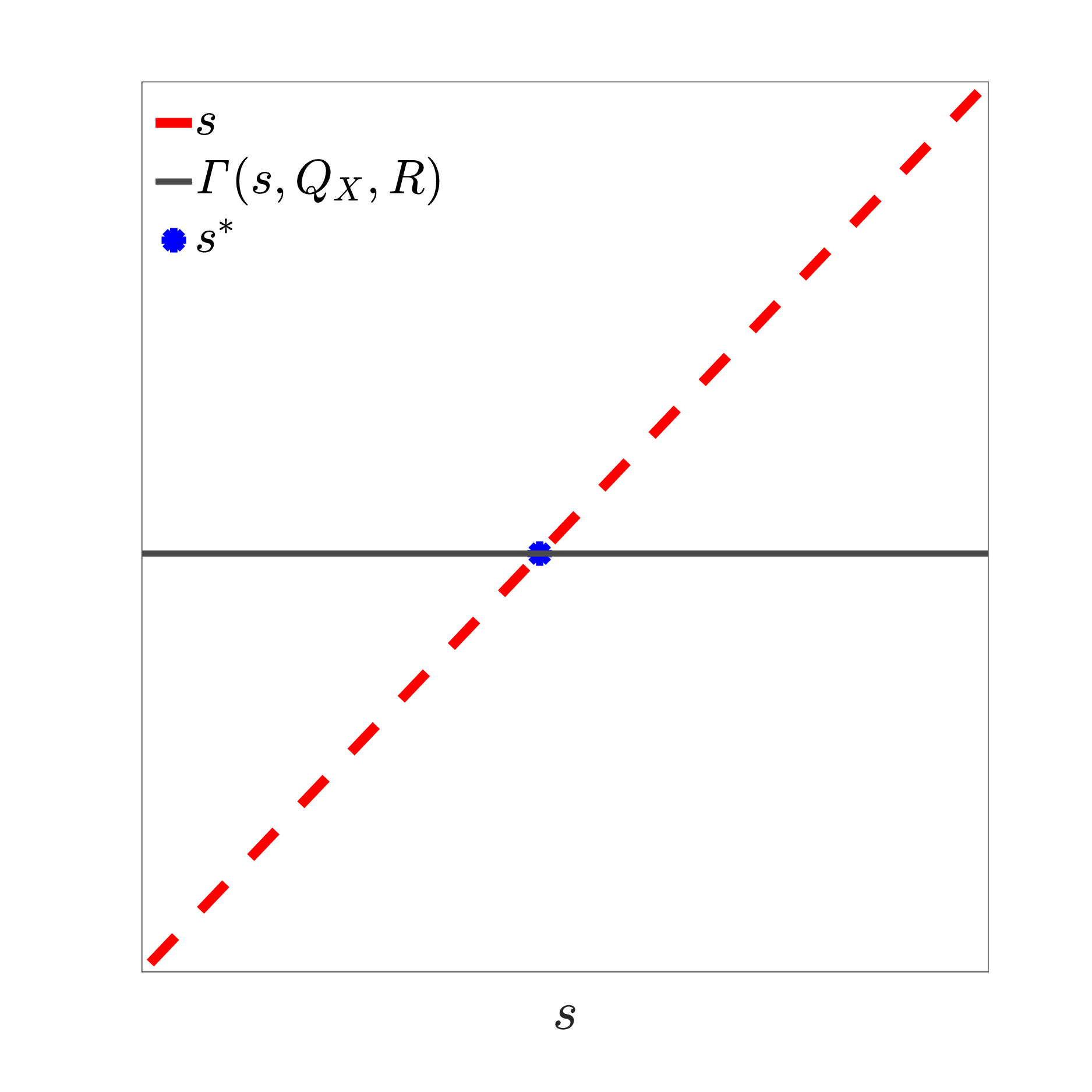}%
        \label{fig-Gs-a}%
    }\hfill
    \subfloat[$\supp(Q_X)\cap\mX_o\neq\varnothing$]{%
        \includegraphics[width=0.32\linewidth]{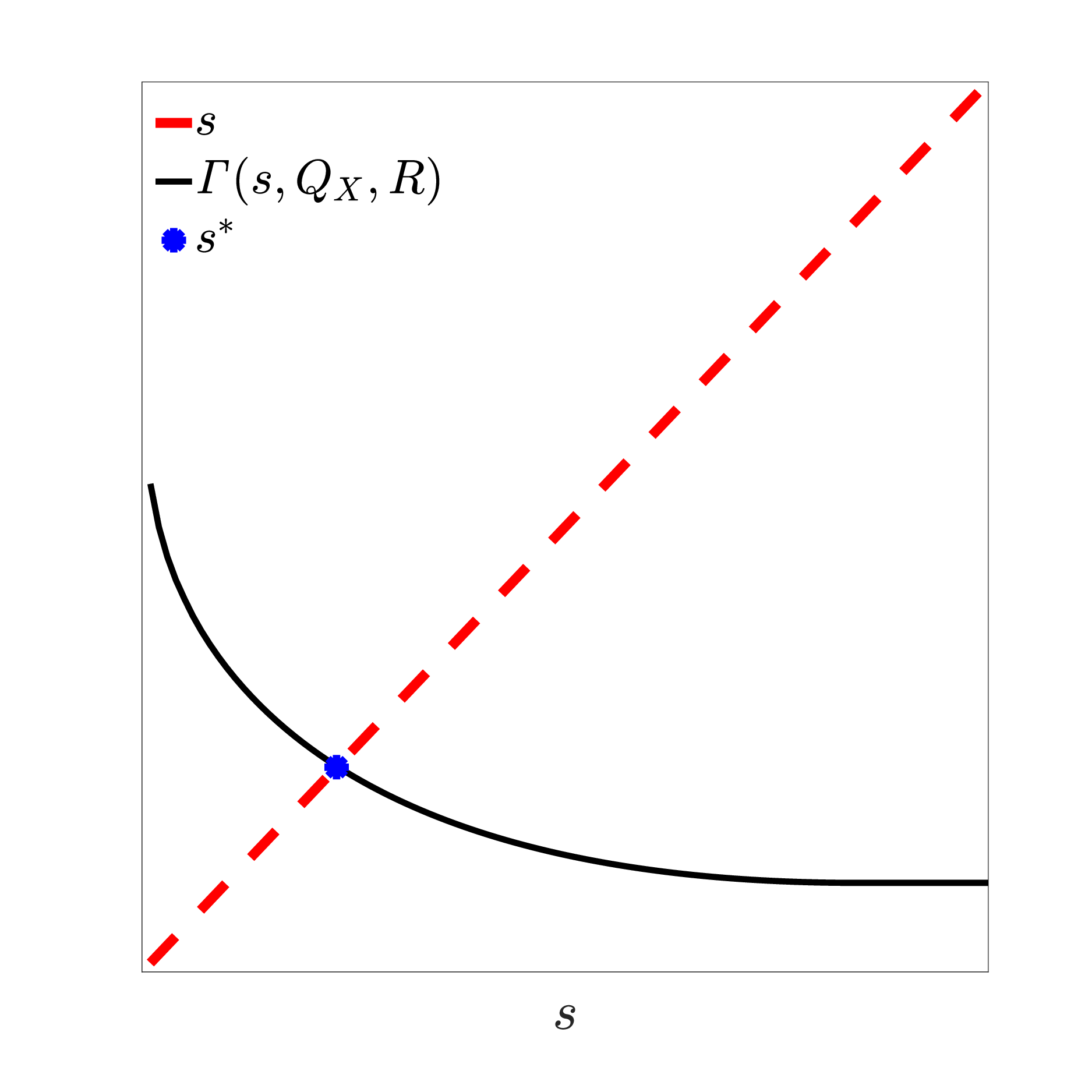}%
        \label{fig-Gs-b}%
    }\hfill
    \subfloat[$\supp(Q_X)\cap\mX_o\neq\varnothing$]{%
        \includegraphics[width=0.32\linewidth]{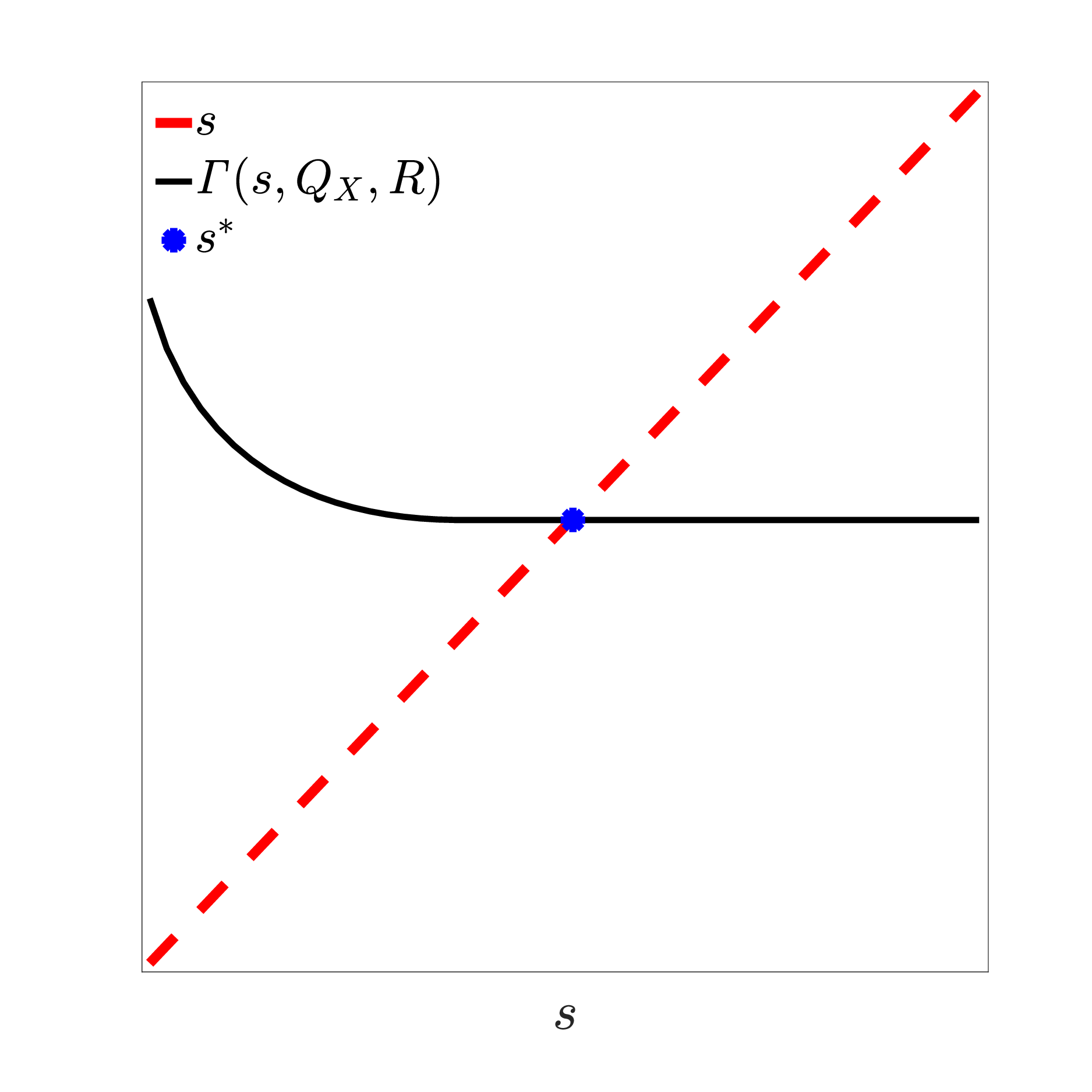}%
        \label{fig-Gs-c}%
    }
    \caption{Three possible ways in which $\Ga(s,Q_X,R)$ intersects $s$.
    Here, $\mX_l$ denotes the set of noiseless input symbols, and $\mX_o$ denotes the set of noisy input symbols.}
    \label{fig-Gs}
\end{figure}

The idea of proving Proposition \ref{prop-sc-con} is closely related to hypothesis testing. We are in fact looking for a decision region $\mA$ such that $\tPy(\mA^c)\leq 2^{-ns}$ and $P_Y^n(\mA) \leq 2^{-n\Ga(s,R)}$ with some exponent $\Ga(s,R)$. The average probability of testing error is $\frac12\big[\tPy(\mA^c) + P_Y^n(\mA)\big] = \frac12 - \frac14\big\|\tP_{Y^n|\mC} - P_Y^n \big\|_1$. Then we optimize over $s$ to find the optimal decision region.
Observing the expression in Proposition \ref{prop-sc-con}, the following properties of $\Gl(R)$ can be established, which are proved in Appendix \ref{app-lem-Gl0}.

\begin{lemma}\label{lem-Gl0}  We have the following properties of $\Gl(R)$: \ 

    (i) $\Gl(R) = 0$ when $R \geq \minI$. \quad 
    (ii) $\Gl(R) > 0$ when $R < \minI$.
\end{lemma}

\subsection{An achievability bound for the strong converse exponent}\label{sec-sc-achi}

In this subsection, we show an upper bound for $\Ga_\su(R)$ (the uniform formulation), denoted by $\Gu(R)$. We develop a novel deterministic code construction technique based on the type covering lemma. The key observation is that since we are operating below $\minI$, we cannot afford codeword repetitions, while random coding produces repeated codewords, albeit rarely. In contrast, we only cover joint types with mutual information less than $R$, and cover only once.

\begin{proposition}[Achievability]\label{prop-sc-achi} The strong converse exponent  $\Ga_\su(R)$ in the uniform formulation can be bounded from above as follows:
    $$ \Ga_\su(R) \leq \Gu(R) := \min_{Q_{XY}\in\mP(\mX\mY): I(Q_X;V) \leq R} \left[D(Q_Y\|P_Y) + \big| R - \iQ \big|^+ \right], $$ 
    where $\iQ$ is defined in Definition \ref{def-iQ}.
\end{proposition}

\begin{proof}
Consider any joint type $Q_{XY}\in\mP_n(\mX\mY)$. Following our notation convention, we introduce the backward conditional type $\bV_{X|Y} = Q_{XY}/Q_Y\in\mP_n(\mX|Q_Y)$. Take any $y^n\in \mT_{Q_Y}$. \eqref{tPy} yields that
\begin{align*}
    \tPy(y^n)
    &= \frac{1}{M} \sum_{\bV\in\mP_n(\mX|Q_Y)} \ \sum_{i=1}^M \mathbbm{1}\{X^n(i)\in \mT_{\bV}(y^n)\} \ W_{Y|X}^n(y^n|X^n(i)) \\
    &\overset{a}{=} \frac{1}{M} \sum_{\bV\in\mP_n(\mX|Q_Y)} k_{\bV}(y^n) \ 2^{-n\left[D(V\|W|Q_X) + H(V|Q_X) \right]},
\end{align*}
where $(a)$ follows from property of types. Furthermore, in $(a)$ we have defined 
\begin{equation}\label{kv}
    k_{\bV}(y^n) := \sum_{i=1}^M \mathbbm{1}\{X^n(i)\in \mT_{\bV}(y^n)\},
\end{equation}
which is the number of codewords that lie in the $\bV$-shell of $y^n$. In other words, $k_{\bV}(y^n)$ counts the number of codewords that have joint type $Q_{XY}$ with $y^n$. Since $P_Y^n(y^n) = 2^{-n[D(Q_Y\|P_Y) + H(Q_Y)]}$ according to property of types, we can evaluate the following ratio:
\begin{align}
    \frac{\tPy(y^n)}{P_Y^n(y^n)}
    &= \frac{1}{M} \sum_{\bV\in\mP_n(\mX|Q_Y)} k_{\bV}(y^n) \ 
        2^{n\left[D(Q_Y\|P_Y) + H(Q_Y) - D(V\|W|Q_X) - H(V|Q_X) \right]} \notag\\
    &= \frac{1}{M} \sum_{\bV\in\mP_n(\mX|Q_Y)} k_{\bV}(y^n) \ 
        2^{n\left[D(Q_Y\|P_Y) + I(Q_X;V) - D(V\|W|Q_X)\right]} \notag\\
    &= \sum_{\bV\in\mP_n(\mX|Q_Y)} k_{\bV}(y^n) \ 2^{n\left[\iQ - R\right]}, \label{ratio}
\end{align}
where the last equality follows from Lemma \ref{lem-Deq}.

Using the identity $|a-b| = a + b - 2\min\{a,b\}$, the total variation satisfies the following property:
\begin{equation}
\label{1-L1=min}
    1 - \frac12 \left\|\tPy - P_Y^n \right\|_1 
    = \sum_{y^n} \min\left\{\tPy(y^n), P_Y^n(y^n)\right\}.
\end{equation}
Combining this and \eqref{ratio}, we obtain that
\begin{align}\label{1-L1}
    & 1 - \frac12 \left\|\tPy - P_Y^n \right\|_1 
    = \sum_{Q_Y\in\mP_n(\mY)}  \ \sum_{y^n\in \mT_{Q_Y}} P_Y^n(y^n) 
        \min\left\{\frac{\tPy(y^n)}{P_Y^n(y^n)}, 1\right\} \notag\\
    &\overset{a}{=} \sum_{Q_Y\in\mP_n(\mY)}  2^{-n\left[D(Q_Y\|P_Y) + H(Q_Y)\right]} \sum_{y^n\in \mT_{Q_Y}} 
        \!\!\!\min\left\{\sum_{\bV\in\mP_n(\mX|Q_Y)} k_{\bV}(y^n) 2^{n\left[\iQ - R\right]}, 1\right\} \notag\\
    &= \sum_{Q_Y\in\mP_n(\mY)} 2^{-n\left[D(Q_Y\|P_Y) + H(Q_Y)\right]} \sum_{y^n\in \mT_{Q_Y}} 
        \!\!\!\min\left\{\sum_{\bV\in\mP_n(\mX|Q_Y)} k_{\bV}(y^n) 2^{n\left[\iQ - R\right]}, \sum_{\bV\in\mP_n(\mX|Q_Y)} \frac{1}{|\mP_n(\mX|Q_Y)|}\right\} \notag\\
    &\overset{b}{\geq} \sum_{Q_Y\in\mP_n(\mY)}  2^{-n\left[D(Q_Y\|P_Y) + H(Q_Y)\right]} \sum_{y^n\in \mT_{Q_Y}} \sum_{\bV\in\mP_n(\mX|Q_Y)} 
        \!\!\!\min\left\{k_{\bV}(y^n) 2^{n\left[\iQ - R\right]}, \frac{1}{|\mP_n(\mX|Q_Y)|}\right\} \notag\\
    &\overset{c}{\geq} \sum_{Q_{XY}\in\mP_n(\mX\mY)}  2^{-n\left[D(Q_Y\|P_Y) + H(Q_Y)\right]} \sum_{y^n\in \mT_{Q_Y}}
        \min\left\{k_{\bV}(y^n) \ 2^{n\left[\iQ - R\right]}, (n+1)^{-|\mX||\mY|}\right\}\notag\\
    &\overset{d}{\geq} (n+1)^{-|\mX||\mY|} \sum_{Q_{XY}\in\mP_n(\mX\mY)}  2^{-n\left[D(Q_Y\|P_Y) + H(Q_Y)\right]} \sum_{y^n\in \mT_{Q_Y}}
        \min\left\{k_{\bV}(y^n) \ 2^{n\left[\iQ - R\right]}, 1\right\},
\end{align}
where $(a)$ follows from property of types and \eqref{ratio}; $(b)$ follows from the fact that the minimum of a sum is greater than or equal to the sum of the minima (in this case, it is equivalent to the triangle inequality); and $(c)$ follows from property of types. In $(d)$, we artifically add an extra term $(n+1)^{-|\mX||\mY|}$ to $k_{\bV}(y^n) \ 2^{n\left[\iQ - R\right]}$ and then extract it outside the minimum. The reason for doing so is that $k_{\bV}(y^n) \ 2^{n\left[\iQ - R\right]}$ is exponential in $n$, so comparing it with 1 or any polynomial factor does not affect the exponential order.

For the purpose of designing a good code, we need to maximize \eqref{1-L1}, and it suffices to make $k_{\bV}(y^n)$ large (or at least nonzero) for all joint types. However, due to the low rate and hence the limited number of codewords, this can only be achieved for certain joint types, and even for those, a large $k_{\bV}(y^n)$ would seem too greedy. A more reasonable approach is to set $k_{\bV}(y^n) = 1$, meaning that there exists a single codeword that covers $y^n$. The type covering lemma guarantees the existence of such a covering for all $y^n$'s, provided that the mutual information is less than the code rate.

Based on the above analysis, we construct our code $\mC$ in the following way. Take any $\epsilon>0$ and examine all joint types $Q_{XY}\in\mP_n(\mX\mY)$: if its mutual information satisfies $I(Q_X;V) \leq R - 2\epsilon$, then according to the type covering lemma (e.g., \cite[Lem.~3.34]{moser2019advanced}), there exists a covering code with $2^{n[I(Q_X;V)+\epsilon]}$ codewords in $\mT_{Q_X}$ that ensures $k_{\bV}(y^n)\geq1$ for all the corresponding output sequences $y^n\in \mT_{Q_Y}$. We collect all codewords in this covering code into our code $\mC$. It is noteworthy that such a collection may contain repeated codewords. Under this strategy, the number of codewords we have assigned so far is
\begin{align*}
    |\mC_{\text{eff}}|
    &=\sum_{Q_{XY}\in\mP_n(\mX\mY): I(Q_X;V) \leq R-2\epsilon} 2^{n[I(Q_X;V)+\epsilon]} \\
    &\overset{a}{\leq} (n+1)^{\carX\carY} \max_{Q_{XY}\in\mP_n(\mX\mY): I(Q_X;V) \leq R-2\epsilon} 2^{n[I(Q_X;V)+\epsilon]} \\
    &\leq (n+1)^{\carX\carY} 2^{n(R-\epsilon)} \leq 2^{nR}
\end{align*}
for sufficienly large $n$, where $(a)$ follows from property of types. The inequality $|\mC_{\text{eff}}| \leq 2^{nR}$ confirms that $\mC$ is a valid code. Specifically, the codewords are drawn from a uniform distribution of $1/M = 2^{-nR}$. $|\mC_{\text{eff}}|$ among them are used for the sake of covering all joint types $Q_{XY}\in\mP_n(\mX\mY)$ satisfying $I(Q_X;V) \leq R - 2\epsilon$; as a result, $k_{\bV}(y^n)\geq1$ for all $y^n$'s in those joint type classes. There are $M - |\mC_{\text{eff}}|$ codewords left, which can be used to cover the other joint types. However, this fraction is small, and we apply a trivial lower bound zero on it. To sum up, in \eqref{1-L1} we have
$$ k_{\bV}(y^n) \geq
\begin{cases}
    1 & \text{if } I(Q_X;V) \leq R - 2\epsilon \\
    0 & \text{if } I(Q_X;V) > R - 2\epsilon
\end{cases} $$
for all $y^n\in \mT_{Q_Y}$. Consequently, \eqref{1-L1} simplifies to
\begin{align*}
  1-\frac12&\left\|\tPy - P_Y^n \right\|_1 \\
    &\geq (n+1)^{-|\mX||\mY|} \sum_{Q_{XY}\in\mP_n(\mX\mY):I(Q_Y;\bV) \leq R-2\epsilon}  2^{-n\left[D(Q_Y\|P_Y) + H(Q_Y)\right]} \ |\mT_{Q_Y}| \
        \min\left\{2^{n\left[\iQ - R\right]}, 1\right\} \\
    &\overset{a}{\geq} (n+1)^{-(|\mX|+1)|\mY|} \sum_{Q_{XY}\in\mP_n(\mX\mY):I(Q_Y;\bV) \leq R-2\epsilon} 2^{-n\left[D(Q_Y\|P_Y) + \left|R - \iQ\right|^+ \right]} \\
    &\geq (n+1)^{-(|\mX|+1)|\mY|} \max_{Q_{XY}\in\mP_n(\mX\mY):I(Q_Y;\bV) \leq R-2\epsilon} 2^{-n\left[D(Q_Y\|P_Y) + \left|R - \iQ\right|^+ \right]} \\
    &= (n+1)^{-(|\mX|+1)|\mY|} \exp_2\left\{ -n \min_{Q_{XY}\in\mP_n(\mX\mY): I(Q_X,V) \leq R-2\epsilon} \left[D(Q_Y\|P_Y) + \big| R - \iQ \big|^+\right] \right\},
\end{align*}
where $(a)$ follows from property of types. Taking $\epsilon\to0$ and $n\to\infty$ completes the proof.
\end{proof}

Based on Proposition \ref{prop-sc-achi}, the following property of $\Gu(R)$ holds, which is proved in Appendix \ref{app-lem-Gu0}.

\begin{lemma}
\label{lem-Gu0} 
If $\minI \leq R \leq \maxI$, then $\Gu(R) = 0$.
\end{lemma}

\subsection{Dual forms of $\Gl(R)$ and $\Gu(R)$ and their equality}
\label{sec-sc-dual}
So far, we have obtained a lower bound $\Gl(R)$ for $\Ga_\sn(R)$ and an upper bound $\Gu(R)$ for $\Ga_\su(R)$, stated in Proposition \ref{prop-sc-con} and \ref{prop-sc-achi}, respectively. Figure \ref{fig-ElEu} presents examples of these two bounds, using the same channel models as in Figure \ref{fig-rc}. Observing Figure \ref{fig-ElEu}, a natural question arises: do the proposed two bounds coincide when $R<\minI$? If they do, then an exact exponent can be tightly squeezed out since $\Ga_\sn(R) \leq \Ga_\su(R)$. In this subsection, we show that the answer is yes by reformulating $\Gl(R)$ and $\Gu(R)$ into their dual forms in terms of $J_{\alpha,\beta}$ in Definition \ref{def-Jab}. These dual forms are presented in the following Proposition \ref{prop-Gl-dual} and \ref{prop-Gu-dual}, respectively.

\begin{figure}[!htbp]
    \centering
    \subfloat[A fully noisy binary channel]{%
        \includegraphics[width=0.49\linewidth]{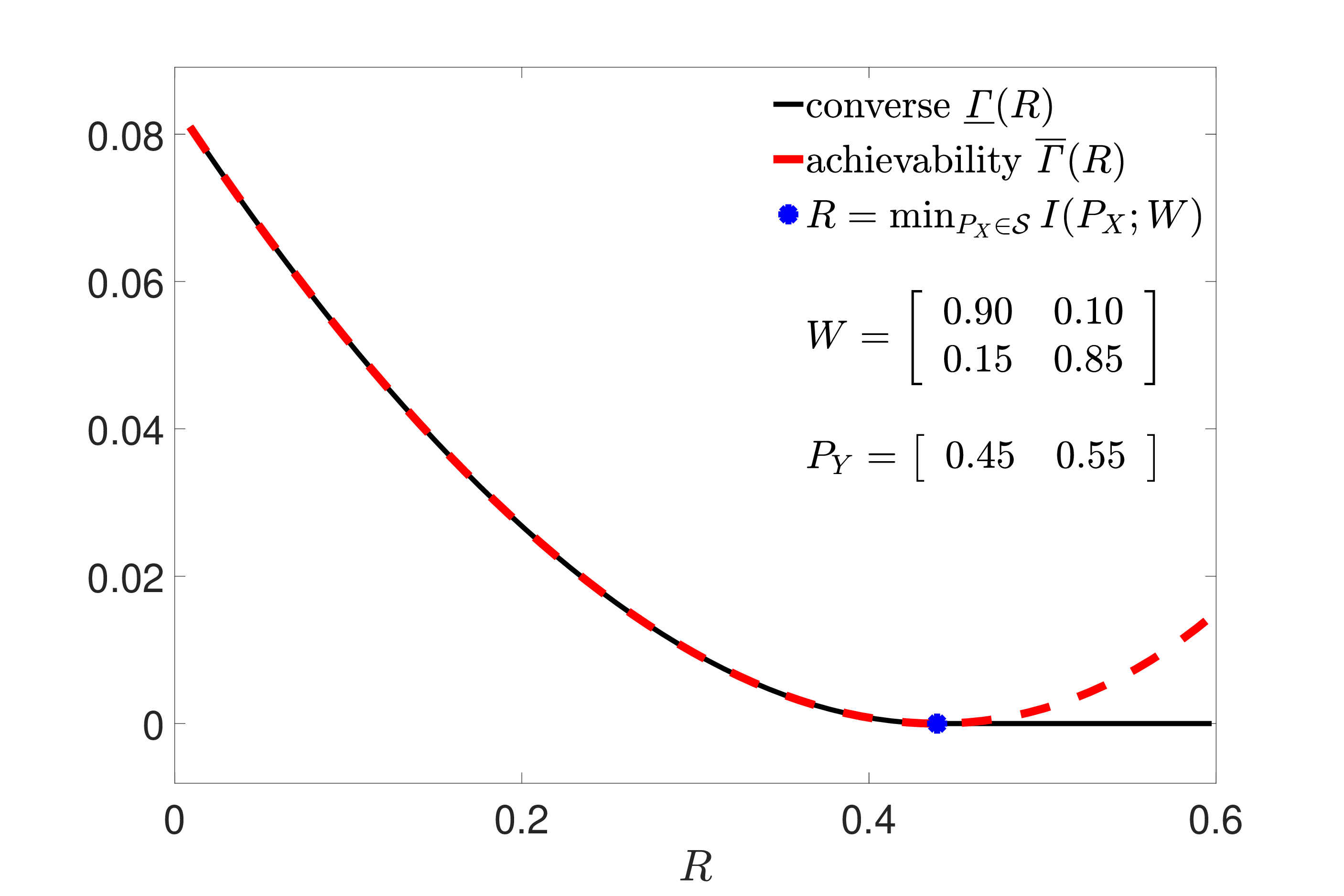}%
    }\hfill
    \subfloat[A fully noisy ternary channel]{%
        \includegraphics[width=0.49\linewidth]{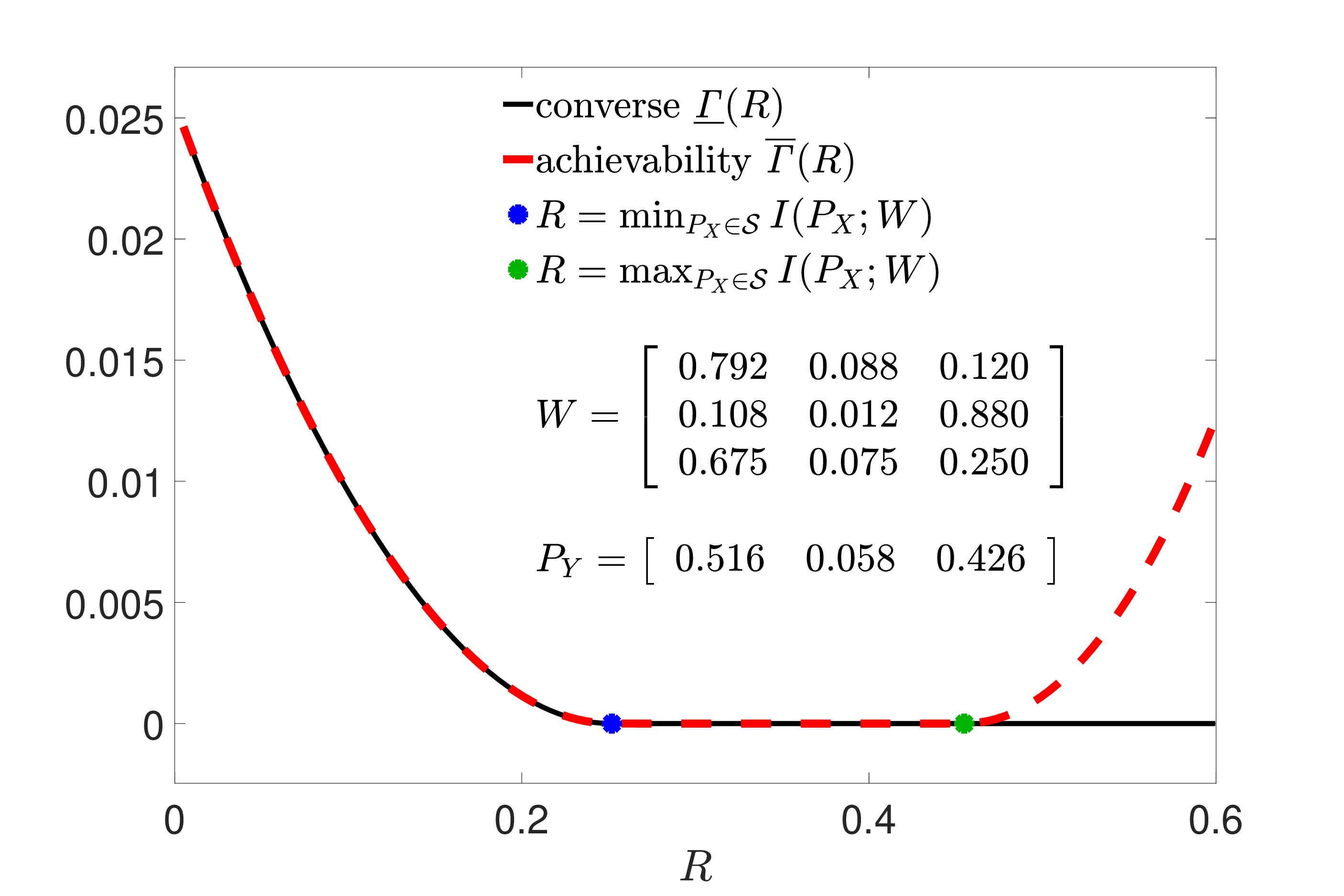}%
    }\\[1ex]
    \subfloat[A fully noiseless ternary channel]{%
        \includegraphics[width=0.49\linewidth]{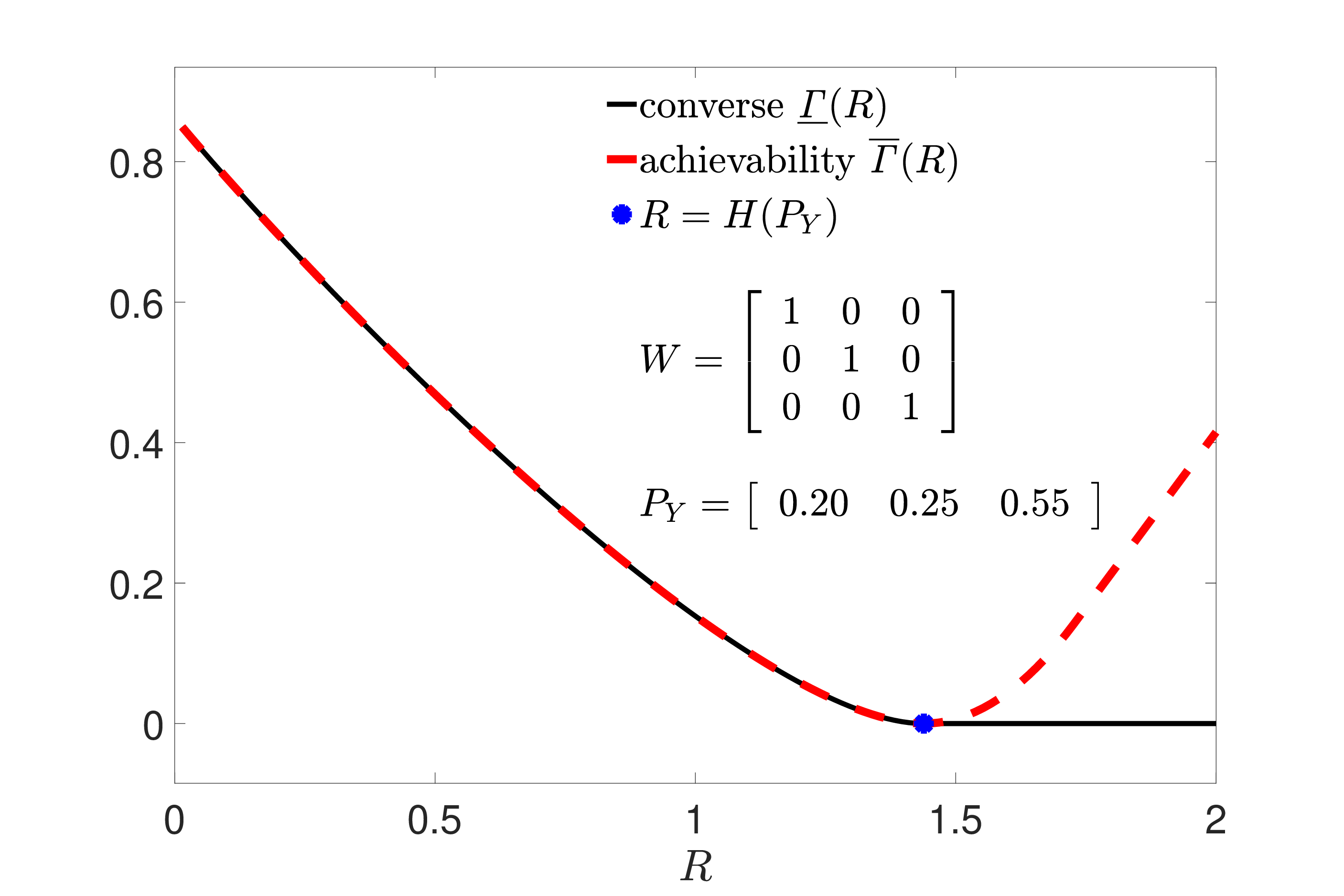}%
    }\hfill
    \subfloat[A hybrid ternary channel]{%
        \includegraphics[width=0.49\linewidth]{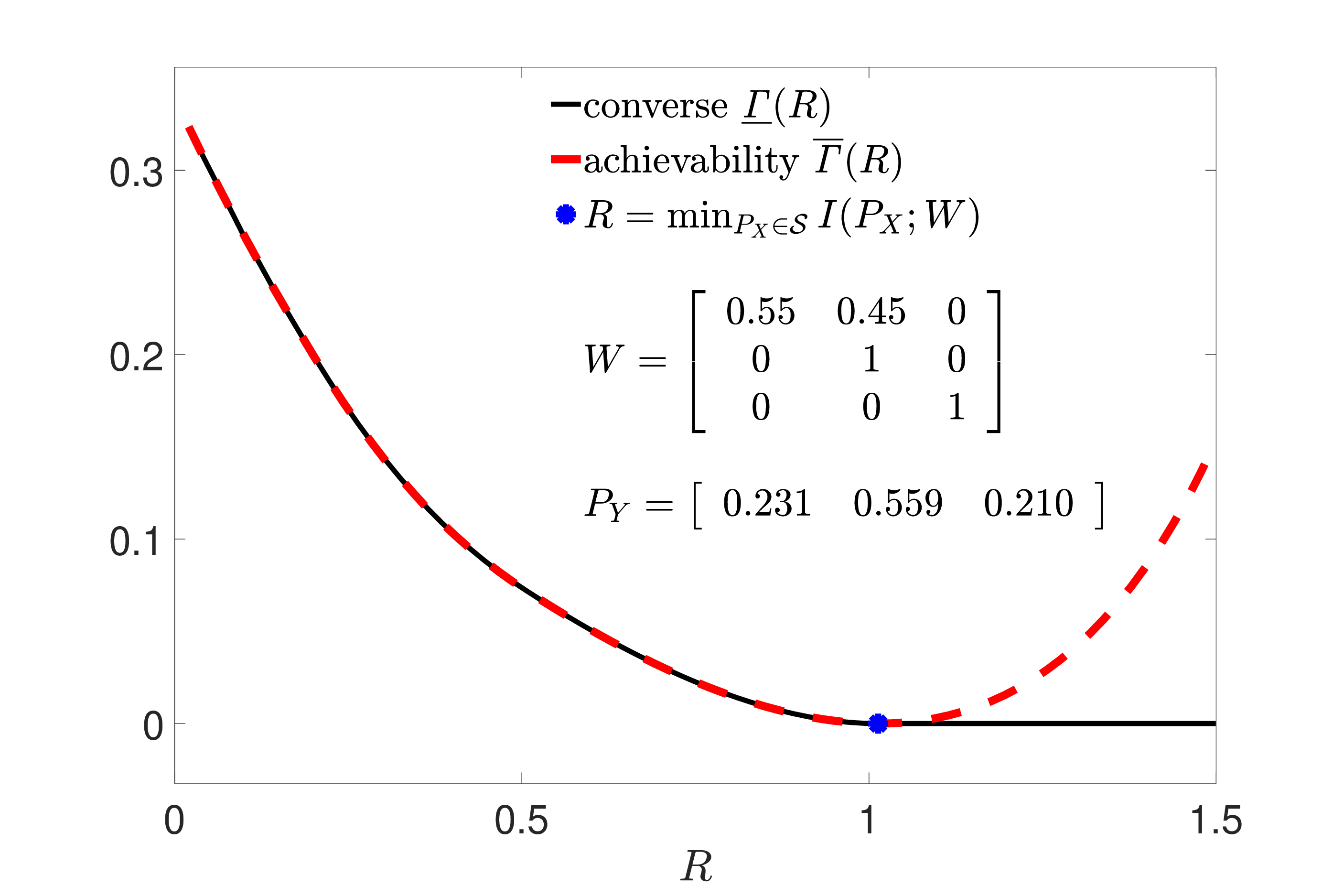}%
    }
    \caption{Examples of the converse bound $\Gl(R)$ and the achievability bound $\Gu(R)$
    for the strong converse exponent.}
    \label{fig-ElEu}
\end{figure}

\begin{proposition}\label{prop-Gl-dual}  $\Gl(R)$ is lower bounded by the following dual form:
    \begin{align}
        \Gl(R) \geq \Gll(R) 
        := \ & \min_{Q_{XY}\in\mP(\mX\mY)} \ \max_{\beta,\gamma\in[0,1],\beta\geq\gamma} 
            \big\{D(Q_Y\|P_Y) + \beta \left[I(Q_X;V)-R\right] + \gamma\left[R-\iQ\right] \big\} \label{prop-Gl-dual-a} \\
        = \ & \max_{\alpha,\beta\in[0,1]} \left[ J_{\alpha,\beta}\left(W_{Y|X}\|P_Y\right) + \beta(\alpha-1) R \right]. \label{prop-Gl-dual-b} 
    \end{align}
\end{proposition}

\begin{proposition}\label{prop-Gu-dual} $\Gu(R)$ has the following dual form:
    $$ \Gu(R) = \max_{\alpha\geq0,\beta\in[0,1]} \left[ J_{\alpha,\beta}\left(W_{Y|X}\|P_Y\right) + \beta(\alpha-1) R \right]. $$
\end{proposition}

Proposition \ref{prop-Gl-dual} is proved in Appendix \ref{app-sec-Gl} and Proposition \ref{prop-Gu-dual} is proved in Appendix \ref{app-sec-Gu}. Comparing the dual forms of $\Gu(R)$ and $\Gll(R)$, we establish their equality in the following Proposition \ref{prop-Gu=Gll}.

\begin{proposition}
\label{prop-Gu=Gll}
$\Gu(R) = \Gll(R)$ when $R<\minI$.
\end{proposition}

Proposition \ref{prop-Gu=Gll} is proved in Appendix \ref{app-sec-Gu=Gll}. Equipped with propositions in this section, the proof of the main result, Theorem \ref{thm-sc}, follows from straightforward logical reasoning.

\begin{proof}[Proof of Theorem \ref{thm-sc}]
First, $\Gll(R) \leq \Gl(R) \leq \Ga_\sn(R) \leq \Ga_\sr(R) \leq \Ga_\su(R) \leq \Gu(R)$ holds for all $R$ by combining Remark \ref{rmk-relation}, Proposition \ref{prop-sc-con}, Proposition \ref{prop-sc-achi}, and Proposition \ref{prop-Gl-dual}.
    
When $R<\minI$, by Proposition \ref{prop-Gu=Gll}, $\Gll(R) = \Gl(R) = \Ga_\sn(R) = \Ga_\sr(R) = \Ga_\su(R) = \Gu(R)$.

When $R\geq\minI$, by Proposition \ref{prop-sc-con}, Proposition \ref{prop-Gl-dual}, and Lemma \ref{lem-Gll0}, $\Ga_\sn(R)\geq\Gl(R) = \Gll(R)=0$. On the other hand, according to the soft covering lemma, e.g., \cite[Ch.~19]{moser2019advanced}, under the uniform formulation there exists a code that can make $\TVinText\to0$ for sufficiently large $n$, yielding that $\left( 1-\TVinText \right) \to 1 = 2^{-n\cdot0}$: a zero strong converse exponent is achievable and thus $\Ga_\su(R)\leq 0$. Ergo, $\Ga_\sn(R) = \Ga_\sr(R) = \Ga_\su(R) = 0 = \Gll(R)$.

To sum up, we can conclude that $\Ga_\sn(R) = \Ga_\sr(R) = \Ga_\su(R) = \Gll(R)$ for all $R$. Combining it with Proposition \ref{prop-Gl-dual} completes the proof. Since $\Gll(R)$ is now exact, rather than merely a lower bound, we drop the underline and denote it by $\Ga(R)$ to simplify notation.
\end{proof}

\section{Error Exponent for Noiseless Channels}
\label{sec-ee-nl}
In this section, we discuss the soft-covering error exponent when $W_{Y|X}$ is noiseless: for each $x\in\mX$, there exists a symbol $y\in\mY$ such that $W_{Y|X}(y|x) = 1$. 
Then we have $W_{Y|X}(y|x) = \mathbbm{1}\{y = w(x)\}$ for all $x \in \mathcal{X}$, $y \in \mathcal{Y}$, and for some function $w \colon \mX\to\mY$. 

\begin{remark}[Choice of alphabet]\label{rmk-chooseX} 
 If $w(\cdot)$ is a bijection, then, without loss of generality, we assume that  $\mX=\mY$.
     If $w(\cdot)$ is a surjection, for each $y\in\mY$, in the achievability proof we may select a single representative $x\in w^{-1}(y)$ and construct an alphabet $\mX_\mC \subset\mX$ consisting of these representatives and hence satisfying $|\mX_\mC| = |\mY|$. In the converse proof, given any code $\mC$, we first proceed with a restriction of symbols: if $y_1 = w(x_1) = w(x_2)$, we replace every occurrence of $x_2$ in the codewords with $x_1$. Since both $x_1$ and $x_2$ map to the same output symbol $y_1$, this modification leaves the output sequences unchanged and, consequently, does not affect the induced distribution $\tPyQ$ or $\tPy$. Under this restriction, the code alphabet set reduces to some $\mX_\mC \subset\mX$ such that $|\mX_\mC|=|\mY|$ and $w(\cdot)$ can be regarded bijective. 
Thus, in both cases, one can always identify an input alphabet set $\mX_\mC$ such that $|\mX_\mC|= |\mY|$ for encoding. 
\end{remark}

\subsection{Uniform formulation}\label{sec-nl-uni}

Under a noiseless channel, the output distribution \eqref{tPy} reduces to
\begin{equation}\label{ky}
    \tPy(y^n) 
    = \frac1M \sum_{x^n} \mathbbm{1}\{y^n = w(x^n)\} \sum_{i=1}^M  \mathbbm{1}\{X^n(i) = x^n\} 
    =: \frac{k(y^n)}{M},
\end{equation}
where $k(y^n)\in\mathbb{Z}^+$ counts the number of codewords mapping to $y^n$. We begin by showing the converse result stated by Theorem \ref{thm-nl-Eu}, which is valid for both rational and irrational $P_Y$.

\begin{proof}[Proof of Theorem \ref{thm-nl-Eu}]\label{pf-nl-Eu}
When $R\geq\minI$, it is relatively easier to approximate large probabilities, while harder to approximate small probabilities, in particular those smaller than the step size $1/M$ in \eqref{ky}. This leads us to the consideration of the following `bad' set:
\begin{equation}\label{badset}
    \mB:= \left\{y^n\in\mY^n: \frac{1}{M} \geq 2 P_Y^n(y^n) \right\} 
    = \bigsqcup_{Q_Y\in\mP_n(\mY): D(Q_Y\|P_Y) + H(Q_Y) \geq R + \frac1n} \mT_{Q_Y}.
\end{equation}
For any $y^n\in\mB$, we have
$$ \left|\frac{k(y^n)}{M} - P_Y^n(y^n)\right| 
\begin{cases}
    \geq \dfrac1M - P_Y^n(y^n) \geq P_Y^n(y^n), & \text{if } k(y^n)\geq1 \\
    = P_Y^n(y^n), & \text{if } k(y^n)=0
\end{cases}. $$
An error of value $P_Y^n(y^n)$ is unavoidable for all $y^n\in\mB$. In this sense, $\mB$ is regarded as `bad', producing errors no matter what code is applied. As a result,
\begin{align*}
    \frac12\left\|\tPy - P_Y^n\right\|_1
    &\geq \frac12 \sum_{y^n\in\mB} P_Y^n(y^n) \\
    &\overset{a}{\geq} \frac12(n+1)^{-\carY} \sum_{Q_Y\in\mP_n(\mY): D(Q_Y\|P_Y) + H(Q_Y) > R+\frac1n} 2^{-n D(Q_Y\|P_Y) } \\
    &\geq \frac12(n+1)^{-\carY} \exp_2\left\{ -n \min_{Q_Y\in\mP_n(\mY): D(Q_Y\|P_Y) + H(Q_Y) > R+\frac1n}  D(Q_Y\|P_Y) \right\},
\end{align*}
where $(a)$ follows from property of types. Taking $n\to\infty$ completes the proof of the variational form.

To obtain the dual form, consider the following.
\begin{align*}
     \min_{Q_Y\in\mP(\mY): D(Q_Y\|P_Y) + H(Q_Y) \geq R} D(Q_Y\|P_Y) 
    \overset{a}{=} \ & \max_{\delta\geq0} \ \min_{Q_Y\in\mP(\mY)} 
        \big\{ D(Q_Y\|P_Y) + \delta\left[ R - D(Q_Y\|P_Y) - H(Q_Y) \right] \big\} \\
    = \ & \max_{\delta\geq0} \left\{ \delta R + \min_{Q_Y\in\mP(\mY)}
        \left[ D(Q_Y\|P_Y) + \delta \sum_y Q_Y(y) \log P_Y(y) \right] \right\} \\
    \overset{b}{=} \ & \max_{\delta\geq0}
        \left\{ \delta R - \log\sum_y P_Y^{1-\delta}(y) \right\} \\
    \overset{c}{=} \ & \max_{\alpha\in(-\infty,0)\cup[1,\infty)} \left\{ \frac{\alpha-1}{\alpha} \left[R - H_{\frac{1}{\alpha}}(P_Y)\right] \right\},
\end{align*}
where $(a)$ follows from the convexity of the optimization problem; $(b)$ follows directly from \eqref{cuff-minQy}; and $(c)$ follows by setting $\alpha:= \frac{1}{1-\delta}$.
\end{proof}

    


It is noteworthy that $\Eu^\nl(R)$ is finite only in a certain region, as specified precisely in the following Lemma \ref{lem-nl-Euinf}, which is proved in Appendix \ref{app-lem-nl-Euinf}.

\begin{lemma}\label{lem-nl-Euinf}
    $\Eu^\nl(R)<\infty$ when $R \leq H_{-\infty}(P_Y)$, and 
    $\Eu^\nl(R)=\infty$ when $R> H_{-\infty}(P_Y)$.
\end{lemma}

Next, we prove the achievability result stated by Theorem \ref{thm-nl-El}, which is also valid for both rational and irrational $P_Y$. Towards proving it, we need the following lemma.

\begin{lemma}\label{lem-N1N2}
    When $R\geq H(P_Y)$, we have $N_1(R) \geq N_2(R)$,  where
    \begin{align*}
        N_1(R) &= \max_{Q_Y\in\mP(\mY): D(Q_Y\|P_Y) + H(Q_Y)\leq R} H(Q_Y), \\
        N_2(R) &= \max_{Q_Y\in\mP(\mY): D(Q_Y\|P_Y) + H(Q_Y)\geq R} \left[R-D(Q_Y\|P_Y)\right].
    \end{align*}
\end{lemma}

Lemma \ref{lem-N1N2} is proved in Appendix \ref{app-lem-N1N2}. Recalling \eqref{ky}, a strategy for constructing a good code is to make $k(y^n)\approx MP_Y(y^n)$ so that $k(y^n)/M$ can approximate $P_Y(y^n)$ as closely as possible. More precisely, to establish the achievability result in Theorem \ref{thm-nl-El}, we show that there exists a code in which $k(y^n)$ differs from either $\lfloor MP_Y(y^n)\rfloor$ or $\lceil MP_Y(y^n)\rceil$ by at most a polynomial quantity in $n$.

\begin{proof}[Proof of Theorem \ref{thm-nl-El}]
Taking inspiration from the proof of Proposition \ref{prop-sc-achi} regarding the achievability of the strong converse exponent, we provide a deterministic code construction for the problem at hand. 
Let us first assume that $R\geq H(P_Y)$. Define
\begin{equation}\label{goodset}
    \mG := \left\{ y^n\in\mY^n: MP_Y^n(y^n)\geq1 \right\}
    = \bigsqcup_{Q_Y\in\mP_n(\mY): D(Q_Y\|P_Y) + H(Q_Y)\leq R} \mT_{Q_Y}.
\end{equation}
$\mG$ is actually a `good' set in the sense that all sequences in $\mG$ satisfy that $\lfloor MP_Y^n(y^n)\rfloor\geq1$. Thus, it is always preferable to cover these sequences, since doing so will yield nonzero $k(y^n)$ values that can align with $MP_Y^n(y^n)$. In fact, ignoring the factor of two, $\mG$ is approximately the complement of the `bad' set $\mB$ defined in \eqref{badset}. Take any $\epsilon>0$. The size of $\mG$ is bounded by
\begin{align*}
    |\mG| &= \sum_{Q_Y\in\mP_n(\mY): D(Q_Y\|P_Y) + H(Q_Y)\leq R} \mT_{Q_Y} 
    \geq \max_{Q_Y\in\mP_n(\mY): D(Q_Y\|P_Y) + H(Q_Y)\leq R} \mT_{Q_Y} \\
    &\overset{a}{\geq} (n+1)^{-\carY} \max_{Q_Y\in\mP_n(\mY): D(Q_Y\|P_Y) + H(Q_Y)\leq R} 2^{nH(Q_Y)} \\
    &\overset{b}{\geq} (n+1)^{-\carY} \ 2^{n \left[N_1(R)-\epsilon\right]},
\end{align*}
where $(a)$ follows from property of types. In $(b)$, the maximization is over all types, so it differs from the maximization over distributions (namely $N_1(R)$ as defined in Lemma \ref{lem-N1N2}) by at most $\epsilon$ for sufficiently large $n$.

To design a good code, we adopt the alphabet choice described in Remark \ref{rmk-chooseX}. With this choice, $k(y^n)$ represents the number of codewords mapped one-to-one to $y^n$, and hence can be directly constructed via a repetition of codewords. Our goal is to judiciously design $k(y^n)$  so that it can satisfy $\lfloor MP_Y^n(y^n) \rfloor \leq k(y^n) \leq \lceil MP_Y^n(y^n) \rceil + \mathrm{poly}(n)$ under the constraint that $\sum_{y^n} k(y^n) = M$. Namely, let us first consider the following two codes:
\begin{align*}
    \mC_1: \quad 
    & k(y^n) = \lfloor MP_Y^n(y^n) \rfloor \quad \forall y^n\in\mY^n, \quad
    M_1 := |\mC_1| = \sum_{y^n} \ \lfloor MP_Y^n(y^n) \rfloor. \\
    \mC_2: \quad 
    & k(y^n) = 
    \begin{cases}
        \lceil MP_Y^n(y^n) \rceil & \text{if } y^n\in\mG \\
        0 & \text{if } y^n\notin\mG
    \end{cases}, \quad
    M_2 := |\mC_2| = \sum_{y^n\in\mG} \lceil MP_Y^n(y^n) \rceil.
\end{align*}
Clearly $M_1\leq M$, meaning that we can use $\mC_1$ to cover all sequences $y^n\in\mY^n$ and there will be some codewords left. Note that $\mG^c := \mY^n\backslash\mG$ is not covered since $\lfloor MP_Y^n(y^n) \rfloor = 0$ for all $y^n\in\mG^c$.
If $M_2\geq M$, then after applying $\mC_1$, each sequence in $\mG$ can be covered at most once more. Hence, there exists a code $\mC_3$ as follows, with $|\mC_3| = M$:
\begin{align*}
    \mC_3: \quad 
    & k(y^n) = 
    \begin{cases}
        \lfloor MP_Y^n(y^n) \rfloor \text{ or } \lceil MP_Y^n(y^n) \rceil  & \text{if } y^n\in\mG \\
        0 & \text{if } y^n\notin\mG
    \end{cases}.
\end{align*}
If $M_2<M$, then after applying $\mC_1$, each sequence in $\mG$ can still be covered more than once. Equivalently, we use $\mC_2$ to cover $\mY^n$, and there will be some codewords left. The number of those codewords that are left is
\begin{align*}
    \Delta M &:= M - M_2 
    = M - \sum_{y^n\in\mG} \lceil MP_Y^n(y^n) \rceil 
    \leq M - \sum_{y^n\in\mG} MP_Y^n(y^n) 
    = M P_Y^n(\mG^c) \\
    &= \sum_{Q_Y\in\mP_n(\mX): D(Q_Y\|P_Y) + H(Q_Y)>R} MP_Y^n(\mT_{Q_Y}) \\
    &\overset{a}{\leq} \sum_{Q_Y\in\mP_n(\mY): D(Q_Y\|P_Y) + H(Q_Y)>R} 2^{n\left[R-D(Q_Y\|P_Y)\right]} \\
    &\overset{b}{\leq} (n+1)^{\carY} \max_{Q_Y\in\mP_n(\mY): D(Q_Y\|P_Y) + H(Q_Y)>R} 2^{n\left[R-D(Q_Y\|P_Y)\right]} \\
    &= (n+1)^{\carY} \exp_2 \left\{ n \max_{Q_Y\in\mP_n(\mY): D(Q_Y\|P_Y) + H(Q_Y)>R} \left[R-D(Q_Y\|P_Y)\right] \right\} \\
    &\leq (n+1)^{\carY} \ 2^{n N_2(R)},
\end{align*}
where $(a)$ and $(b)$ follow from properties of types, and $N_2(R)$ is defined in Lemma \ref{lem-N1N2}. 
Now let us consider using these $\Delta M$ codewords to continue covering $\mG$ after $\mC_2$ is applied. We can cover each sequence in $\mG$ by $\frac{\Delta M}{|\mG|}$ more times, which is at most
$$ \frac{\Delta M}{|\mG|} \leq (n+1)^{2\carY} 2^{n \left[N_2(R)-N_1(R) + \epsilon\right]} \leq (n+1)^{2\carY} 2^{n\epsilon} \leq 2^{2n\epsilon} $$
for all sufficiently large $n$. Here we have used $N_1(R)\geq N_2(R)$ when $R\geq H(P_Y)$ from Lemma \ref{lem-N1N2}. Hence, there exists a code $\mC_4$ as follows, with $|\mC_4| = M$:
\begin{align*}
    \mC_4: \quad 
    &
    \begin{cases}
        \lceil MP_Y^n(y^n) \rceil \leq k(y^n) \leq \lceil MP_Y^n(y^n) \rceil + 2^{2n\epsilon} & \text{if } y^n\in\mG \\
        k(y^n) = 0 & \text{if } y^n\notin\mG
    \end{cases}.
\end{align*}

According to the above discussions, either $\mC_3$ or $\mC_4$ is possible and has exactly $M$ codewords. Therefore, we can summarize that for all $\epsilon>0$ and sufficiently large $n$, there exists a code $\mC$ as follows, with $|\mC| = \sum_{y^n} k(y^n) = M$:
\begin{align*}
    \mC: \quad 
    &
    \begin{cases}
        \lfloor MP_Y^n(y^n) \rfloor \leq k(y^n) \leq \lceil MP_Y^n(y^n) \rceil + 2^{2n\epsilon} & \text{if } y^n\in\mG \\
        k(y^n) = 0 & \text{if } y^n\notin\mG
    \end{cases}.
\end{align*}
Choose $\mC$ to be our covering code, which yields that
\[
  \left| \frac{k(y^n)}{M} - P_Y^n(y^n) \right| 
   \leq \left\{\begin{array}{@{}l l@{}}
          \dfrac{1+2^{2n\epsilon}}{M} & \text{if } y^n \in \mG \\
          P_Y^n(y^n) & \text{if } y^n \notin \mG
        \end{array}\right\}
   \leq 2^{3n\epsilon} \min\left\{\frac{1}{M}, P_Y^n(y^n)\right\}. 
\]
Consequently, we have 
\begin{align*}
    \frac12 \left\|\tPy - P_Y^n\right\|_1
    &\leq 2^{3n\epsilon-1}  \sum_{y^n} \min\left\{\frac{1}{M}, P_Y^n(y^n)\right\} \\
    &= 2^{3n\epsilon-1} \sum_{Q_Y\in\mP_n(\mY)} |\mT_{Q_Y}| 
        \min\left\{ 2^{-nR}, 2^{-n\left[D(Q_Y\|P_Y) + H(Q_Y)\right]} \right\} \\
    &\overset{a}{\leq} 2^{3n\epsilon-1} \sum_{Q_Y\in\mP_n(\mY)} 2^{nH(Q_Y)} 
        \min\left\{ 2^{-nR}, 2^{-n\left[D(Q_Y\|P_Y) + H(Q_Y)\right]} \right\} \\
    &\overset{b}{\leq} \frac12(n+1)^{\carY} \ 2^{3n\epsilon} \max_{Q_Y\in\mP_n(\mY)} \min\left\{ 2^{-n[R - H(Q_Y)]}, 2^{-nD(Q_Y\|P_Y)} \right\} \\
    &= \frac12(n+1)^{\carY} \exp_2 \left\{-n \min_{Q_Y\in\mP_n(\mY)} \left[ D(Q_Y\|P_Y) + \big| R - D(Q_Y\|P_Y) - H(Q_Y) \big|^+ - 3\epsilon \right] \right\},
\end{align*}
where $(a)$ and $(b)$ follow from properties of types. The exponent is exactly $\El^\nl(R)$ as $\epsilon\to0$ and $n\to\infty$. 

Note that all the above discussions are based on the assumption that $R\geq H(P_Y)$. If $R<H(P_Y)$, we can just apply a trivial bound $\TVinText\leq1$ and the error exponent is $0$. Since $\El^\nl(R)$ also vanishes for $R<H(P_Y)$, it is thus an achievable error exponent for $R$ values both below and above $H(P_Y)$. Hence, the variational form in the theorem is proved.

To obtain the dual form, consider the following.
\begin{align*}
    \min_{Q_Y\in\mP(\mY)} &
        \left\{ D(Q_Y\|P_Y) + \big| R - D(Q_Y\|P_Y) - H(Q_Y) \big|^+ \right\} \\
    = \ & \max_{\lambda\in[0,1]} \ \min_{Q_Y\in\mP(\mY)} 
        \big\{ D(Q_Y\|P_Y) + \lambda\left[ R - D(Q_Y\|P_Y) - H(Q_Y) \right] \big\} \\
    = \ & \max_{\lambda\in[0,1]} \left\{\lambda R + \min_{Q_Y\in\mP(\mY)}
        \left[ D(Q_Y\|P_Y) + \lambda \sum_y Q_Y(y) \log P_Y(y) \right] \right\} \\
    \overset{a}{=} \ & \max_{\lambda\in[0,1]}
        \left\{ \lambda R - \log \sum_y P_Y^{1-\lambda}(y) \right\} \\
    \overset{b}{=} \ & \max_{\alpha\geq1} \left\{ \frac{\alpha-1}{\alpha} \left[R - H_{\frac{1}{\alpha}}(P_Y)\right] \right\},
\end{align*}
where $(a)$ follows from \eqref{cuff-minQy} and $(b)$ follows by setting $\alpha:= \frac{1}{1-\lambda}$.
\end{proof}

The idea behind this proof is actually a quantization of $P_Y^n$ with a minimal gap of $1/M$. For all sequences in the good set $\mG$, the resulting error is $1/M$. For sequences outside $\mG$, each probability is less than $1/M$ so we choose not to cover them, and the error equals their probability. The key step is to show that such a construction can indeed be realized under the constraint that the number of codewords is exactly $M$.

\subsection{Rational-irrational discrepancy under the uniform formulation}\label{sec-nl-disc}

Recalling \eqref{ky}, under the uniform formulation, $\tPy$ always takes the form $\tPy = \mathbb{Z}^+/M \in\mathbb{Q}$. This gives rise to the rational-irrational discrepancy. In the following, we state Khintchine's theorem and use it to prove the linear converse in Theorem \ref{thm-irr}.

\begin{lemma}[{Khintchine's theorem~\cite[Thm.~1.10] {bugeaud2004approximation}, 
\cite[Thm.~3.3.2]{queffelec2013diophantine}}] \label{lem-khint}
    Let $\Psi: \mathbb{R}_{\geq1}\to\mathbb{R}_{>0}$ be a continuous function such that $M \mapsto M^2 \Psi(M)$ is non-increasing and that $\sum_{M=1}^\infty M \Psi(M) < \infty$.
    Then for almost every $p\notin\mathbb{Q}$ (in the sense of Lebesgue measure on $\mathbb{R}$), we have
    $ \left|K/M - p\right| < \Psi(M) $ for only finitely many $K,M\in\mathbb{Z}$. 
\end{lemma}

\begin{proof}[Proof of Theorem \ref{thm-irr}]
Let $\by\in\mY$ be an irrational symbol, i.e., $P_Y(\by)\notin\mathbb{Q}$. Define
$ \mA := \left\{ y^n\in\mY^n: y_1 = \by \right\} $
to be the set of output sequences in which $\by$ appears in the first position. According to \eqref{setA},
\begin{equation}\label{dio}
    \TVinEq \geq \left| \tPy(\mA) - P_Y^n(\mA) \right|,
\end{equation}
where 
\begin{align*}
    \tPy(\mA) &= \sum_{y^n\in\mA} \tPy(y^n) 
        = \frac1M \sum_{y^n\in\mA} k(y^n) =: \frac{K}{M}, \\
    P_Y^n(\mA) &= \sum_{y^n\in\mA} P_Y^n(y^n)
        = P_Y(\by) \sum_{y^n\in\mY^{n-1}} P_Y^{n-1}(y^n) = P_Y(\by).
\end{align*}
The quantity $k(y^n)$ is defined in \eqref{ky} and $K:= \sum_{y^n\in\mA} k(y^n)$ is an integer. Hence, $\tPy(\mA)\in\mathbb{Q}$ and $P_Y^n(\mA)\notin\mathbb{Q}$, and \eqref{dio} describes a problem of Diophantine approximation. Taking any $\epsilon>0$ and $\Psi(M) = M^{-(2+\epsilon)}$, it follows immediately from Lemma \ref{lem-khint} that
$$ \left| \tPy(\mA) - P_Y^n(\mA) \right| 
= \left|\frac{K}{M} - P_Y(\by) \right|
\geq \frac{1}{M^{2+\epsilon}} = -2^{-n(2+\epsilon)R} $$
for all sufficiently large $n$, since $n$ representing the number of channel uses can take infinitely many integers. Taking $\epsilon\to0$ completes the proof.
\end{proof}

Alternatively, if $P_Y$ is rational, for high rates a perfect covering can be achieved, as stated in Theorem \ref{thm-ra}. This is because $M P_Y(y^n)$ can always be an integer for sufficiently large $n$, making $k(y^n) = M P_Y(y^n)$ a valid code construction. The detailed proof is as follows.

\begin{proof}[Proof of Theorem \ref{thm-ra}]
Given $P_Y(y) = \frac{A_y}{B_y}$, we have
$$ M P_Y^n(y^n) = M \prod_{i=1}^n \frac{A_{y_i}}{B_{y_i}} = \prod_{i=1}^n \left( 2^{\frac{\log M}{n}}\frac{A_{y_i}}{B_{y_i}} \right). $$
When $R \geq \log(\lcm(\{B_y\}_{y\in\mY}))$, there always exist sufficiently large $M$ and $n$ such that $2^\frac{\log M}{n}$ is a multiple of $B_y$ for each $y\in\mY$ and $\frac{\log M}{n}\to R$ as $n\to\infty$. Therefore, we can employ the alphabet choice in Remark \ref{rmk-chooseX} and set a repetition number $k(y^n) = M P_Y^n(y^n)$.
\end{proof}

\begin{example}
    Consider a ternary output $P_Y = [\frac12, \frac13, \frac16]$. Then $E_\su^\nl(R) = \infty$ when $R\geq\log6$.
\end{example}

\subsection{Non-uniform formulation}\label{sec-nl-non}

If the messages are non-uniformly distributed, then in the achievability proof, one must address not only how to construct a good covering code, but also how to select the optimal message distribution. It turns out that an effective choice is to treat $P_Y^n$ as a source and apply lossless source coding, with the decoding index serving as the message index in the covering setting. More generally, after formulating the lossless source coding problem in Definition \ref{def-source}, we show the equivalence between noiseless soft covering and lossless source coding in Lemma \ref{lem-equiv}.

\begin{definition}[Lossless source coding]\label{def-source}
    Let $\mM = \{1,\dots,M\}$ with $M=2^{nR}$. Consider a discrete memoryless source $\mY^n$ subject to i.i.d. $P_Y$. An $(R,n,P_Y,\Pe)$ lossless source coding scheme consists of an encoder $\mE:\mY^n\to \mM$ and a decoder $\mD:\mM\to\hmY^n$, with $\Pe := \Pr\{Y^n\neq\hY^n\}$ denoting the probability of decoding error.
\end{definition}

\begin{lemma}[Equivalence between noiseless soft covering and lossless source coding]\label{lem-equiv}
    Let $w:\mX\to\mY$ be a bijective function. Consider a noiseless channel $W_{Y|X}(y|x) = \mathbbm{1}\{y = w(x)\}$. 
    \begin{itemize}
        \item [(a)] \hypertarget{lem-equiv-a}{}
        For any $(R,n,P_Y,\Pe)$ lossless source coding scheme, there exists an $(R,n,P_Y,W_{Y|X},q)$ non-uniform soft-covering scheme such that $\TVQinText\leq\Pe$.
        \item [(b)] \hypertarget{lem-equiv-b}{} 
        For any $(R,n,P_Y,W_{Y|X},q)$ non-uniform soft-covering scheme, there exists an $(R,n,P_Y,\Pe)$ lossless source coding scheme such that that $\Pe\leq\big\|\tPyQ-P_Y^n\big\|_1$.
    \end{itemize}
\end{lemma}

Lemma \ref{lem-equiv} is proved in Appendix \ref{app-lem-equiv}. Now we employ this lemma to show Theorem \ref{thm-nl-non}.

\begin{proof}[Proof of Theorem \ref{thm-nl-non}]
We first show achievability. Follow the alphabet choice in Remark \ref{rmk-chooseX} so that $w(\cdot)$ is a bijection. To apply Lemma \ref{lem-equiv}, we need to specify a lossless source coding scheme for the i.i.d. source $P_Y^n$. Take any $\epsilon>0$ and define
$$ \mA := \bigsqcup_{Q_Y\in\mP_n(\mY):H(Q_Y)\leq R-\epsilon} \mT_{Q_Y}. $$
Clearly $|\mA|\leq M = 2^{nR}$, meaning that we can establish a one-to-one label for every sequence in $\mA$ using at most $M$ codewords. Therefore, there exists a $(R,n,P_Y,\Pe)$ lossless source coding scheme, where all  sequences in $\mA$ can be correctly decoded and hence $\Pe\leq P_Y^n(\mA^c)$ with $\mA^c=\mY^n\backslash\mA$. According to Lemma \equiva, there exists an $(R,n,P_Y,W_{Y|X},q)$ non-uniform soft-covering scheme such that
\begin{align*}
    \frac12 \left\| \tPyQ - P_Y^n \right\|_1 
    &\leq \Pe \leq P_Y^n(\mA^c) 
    = \sum_{Q_Y\in\mP_n(\mY):H(Q_Y)>R-\epsilon} P_Y^n(\mT_{Q_Y}) \\
    &\overset{a}{\leq} (n+1)^{\carY} \exp_2\left\{ -n \min_{Q_Y\in\mP_n(\mY):H(Q_Y)>R-\epsilon} D(Q_Y\|P_Y) \right\},
\end{align*}
where $(a)$ follows from property of types. Taking $\epsilon\to0$ and $n\to\infty$ completes the proof of achievability.

For the converse, given any covering code, by alphabet restriction in Remark \ref{rmk-chooseX}, $w(\cdot)$ can be regarded bijective. Hence, according to Lemma \equivb, for any $(R,n,P_Y,W_{Y|X},q)$ non-uniform soft-covering scheme, there exists a corresponding $(R,n,P_Y,\Pe)$ lossless source coding scheme such that
\begin{align*}
    \frac12 \left\| \tPyQ - P_Y^n \right\|_1 
    \geq \frac12 \Pe
    \overset{a}{\geq} \frac12 (n+1)^{-\carY} \exp_2\left\{ -n \min_{Q_Y\in\mP_n(\mY):H(Q_Y)>R+\epsilon} D(Q_Y\|P_Y) \right\},
\end{align*}
where $(a)$ follows from the converse exponent of the lossless source coding, e.g., \cite[Thm.~2.15]{csiszar2011information}. Taking $\epsilon\to0$ and $n\to\infty$ completes the proof of the converse.

So far, we have proved the variational form, which is exactly the error exponent of the lossless source coding problem. The dual form thus follows immediately from \cite[Problem~2.15]{csiszar2011information}.
\end{proof}

\subsection{$H_{-\infty}$-constrained formulation}\label{sec-nl-ren}

\begin{proof}[Proof of Theorem \ref{thm-nl-ren}]
We begin with the converse. In fact, the proof of the converse is identical to that of Theorem \ref{thm-nl-Eu} in Section \ref{sec-nl-uni}. The `bad' set $\mB$ defined in \eqref{badset} applies here as well, because in the $H_{-\infty}$-constrained formulation, the minimal probability in the message set is also at least $1/M$. All output sequences with probability $P_Y^n(y^n) \leq 1/2M$ contribute to the covering error. Ergo, $\Eu^\nl(R)$ in Theorem \ref{thm-nl-Eu}, which serves as a converse for the uniform formulation, also constitutes a converse for the $H_{-\infty}$-constrained formulation, even though the latter is a more general formulation that encompasses the former.

For the achievability, we follow the alphabet choice in Remark \ref{rmk-chooseX} so that $w(\cdot)$ is a bijection. Similar to the non-uniform case, we consider a lossless source coding scheme for $P_Y^n$, where only sequences with probability greater than $1/M$ are encoded. This can be interpreted as a lossless source coding scheme where the rate is measured by $H_{-\infty}$ of the encoded messages, instead of $H_0$, which is equal to the logarithm of the size of the message set. Those sequences are in fact from the good set $\mG$ defined in \eqref{goodset}. The corresponding source coding scheme is given by
\begin{equation}\label{encoder-ren}
    \mE(y^n) = 
        \begin{cases}
            1,2,\dots,|\mG| & \text{if } y^n\in\mG, \\
            0 & \text{if } y^n\notin\mG,
        \end{cases} \qquad 
    \mD(i) = 
        \begin{cases}
            \mE^{-1}(i) & \text{if } i = 1,2,\dots,|\mG|, \\
            \text{declare error} & \text{if } i = 0.
        \end{cases}
\end{equation}
Clearly, all sequences in $\mG^c=\mY^n\backslash\mG$ contribute to the decoding error. According to Lemma \equiva, there exists an $(R,n,P_Y,W_{Y|X},q)$ non-uniform soft-covering scheme such that
\begin{align*}
    \frac12 \left\| \tPyQ - P_Y^n \right\|_1 
    &\leq \Pe = P_Y^n(\mG^c) 
    = \sum_{Q_Y\in\mP_n(\mY): D(Q_Y\|P_Y) + H(Q_Y) > R} P_Y^n(\mT_{Q_Y}) \\
    &\overset{a}{\leq} (n+1)^{\carY} \exp_2\left\{ -n \min_{Q_Y\in\mP_n(\mY):D(Q_Y\|P_Y) + H(Q_Y) > R} D(Q_Y\|P_Y) \right\},
\end{align*}
where $(a)$ follows from property of types. The exponent here is exactly $\Eu^\nl(R)$ as $n\to\infty$. However, directly from Lemma \equiva, the soft-covering coding scheme here is under the non-uniform formulation. We need to show that construction in \eqref{encoder-ren} can indeed result in a soft-covering coding scheme under the $H_{-\infty}$-constrained formulation. Specifically, the corresponding message distribution in \eqref{qi-source} (see Appendix \ref{app-lem-equiv}) needs to satisfy $q(i)\geq1/M$ for all $i=0,1,\dots,|\mG|$. First, the definition of $\mG$ in \eqref{goodset} ensures that $q(i)\geq 1/M$ for $i=1,\dots,|\mG|$. It remains to verify that $q(0)\geq 1/M$. Since
$q(0) = \Pe$, it is equivalent to check $\Eu^\nl(R) \leq R$. By \eqref{Eu-nl-boundary} in Appendix \ref{app-nl-exact}, when $\Eu^\nl(R)$ is finite, we have $\Eu^\nl(R) = D(Q_Y^*\|P_Y)$ for some optimizer $Q_Y^*\in\mP(\mY)$ such that $D(Q_Y^*\|P_Y) + H(Q_Y^*) = R$. Therefore, $\Eu^\nl(R) \leq R$. This completes the proof.
\end{proof}

\section{Error Exponent for Noisy Channels}\label{sec-ee-ny}

This section addresses soft covering error exponents for noisy channels. We prove Theorem \ref{thm-ee-achi} and Theorem \ref{thm-ee-con}.

\subsection{Achievability}\label{sec-ee-ny-achi}

\begin{proof}[Proof of Theorem \ref{thm-ee-achi}]
We can use the soft covering with noiseless channels to the case of noisy channels by covering the input distribution instead of the output. This leads to a high rate improvement in the error exponent achievability as compared to random coding. As an example, consider the uniform formulation.
The output distribution in \eqref{tPy} can be written as 
$$ \tPy(y^n) = \sum_{x^n} \tPx(x^n) \ W(y^n|x^n), $$
where $\tPx$ is the code-induced input distribution. Pick any $P_X\in\mS$. Then $P_Y = P_X W$. We have
$$ \TVinEq \overset{a}{\leq} \frac12 \left\|\tPx - P_X^n\right\|_1 
\overset{b}{\leq} 2^{-n \El^\nl(R)}, $$
where $(a)$ follows from the data processing inequality of the total variation, and $(b)$ follows from Theorem \ref{thm-nl-El}. Here $\El^\nl(R)$ is the noiseless bound of the $(R,n,P_X)$ covering problem. Hence, we can follow the code construction in the proof of Theorem \ref{thm-nl-El} in Section \ref{sec-nl-uni} and then cover the input distribution $P_X^n$. An optimal bound is further generated by maximizing over $P_X\in\mS$. Similar arguments extend to the non-uniform and $H_{-\infty}$-constrained formulations.
\end{proof}

\subsection{Converse}\label{sec-ee-ny-con}

\begin{proof}[Proof of Theorem \ref{thm-ee-con}]
Consider any code $\mC$. We claim that exists a type $\bQ_X\in\mP_n(\mX)$ such that
\begin{equation}\label{const-comp}
    \sum_{i=1}^M q(i) \ \mathbbm{1}\{X^n(i)\in\mT_{\bQ_X}\} \geq 2^{-n\epsilon}. 
\end{equation}
To see this, suppose $ \sum_{i=1}^M q(i) \ \mathbbm{1}\{X^n(i)\in\mT_{Q_X}\} < 2^{-n\epsilon} $ for all types $Q_X\in\mP_n(\mX)$. Then we have
\begin{align*}
    \sum_{i=1}^M q(i) 
    &= \sum_{i=1}^M q(i) \sum_{Q_X\in\mP_n(\mX)} \mathbbm{1}\{X^n(i)\in\mT_{Q_X}\} 
    = \sum_{Q_X\in\mP_n(\mX)} \ \sum_{i=1}^M q(i) \ \mathbbm{1}\{X^n(i)\in\mT_{Q_X}\} \\ 
    &\leq |\mP_n(\mX)| 2^{-n\epsilon}
    \leq (n+1)^{\carX} 2^{-n\epsilon} < 1
\end{align*}
for sufficiently large $n$, thereby leading to a contradiction. Hence, \eqref{const-comp} is proved, which implies that even though $\mC$ is not necessarily a constant composition code, it has a constant composition subset that includes most probable codewords.

We follow \eqref{setA} and set
$$ \mA = \bigcup_{i:X^n(i)\in \mT_{\bQ_X}}
    \bigsqcup_{V\in\mV_n(\bQ_X)} \mT_V(X^n(i)). $$
That is, only restrict to codewords with type $\bQ_X$ and $V$-shell in $\mV_n(\bQ_X)$, where $\mV_n(\bQ_X)$ is some subset of $\mP_n(\mY|\bQ_X)$ that will be properly chosen later. Following the proof of Proposition \ref{prop-sc-con}, we take a codeword $X^n(i)\in \mT_{\bQ_X}$ and can obtain that
\begin{align*}
    W_{Y|X}^n(\mA|X^n(i))
    &\geq W_{Y|X}^n \left( \bigsqcup_{V\in\mV_n(\bQ_X)} \mT_V(X^n(i)) \bigg| X^n(i) \right) 
    \geq \max_{V\in\mV_n(\bQ_X)} W_{Y|X}^n \left(\mT_V(X^n(i)) \big| X^n(i) \right) \\
    &\overset{a}{\geq} (n+1)^{-|\mX||\mY|} \exp_2 \left\{ -n \min_{V\in\mV_n(\bQ_X)} D(V\|W|\bQ_X) \right\},
\end{align*}
where $(a)$ follows from property of types. Averaging over all codewords further yields that
\begin{align*}
    \tPyQ(\mA) 
    &= \sum_{i=1}^M q(i) \ W_{Y|X}^n(\mA|X^n(i)) \\
    &\geq \sum_{i=1}^M q(i) \ \mathbbm{1}\{X^n(i)\in\mT_{\bQ_X}\} \ W_{Y|X}^n(\mA|X^n(i)) \\
    &\geq (n+1)^{-|\mX||\mY|} \exp_2 \left\{ -n \min_{V\in\mV_n(\bQ_X)} D(V\|W|\bQ_X) \right\} \sum_{i=1}^M  q(i) \sum_{Q_X\in\mP_n(\mX)} \mathbbm{1}\{X^n(i)\in\mT_{Q_X}\} \\
    &\overset{a}{\geq} (n+1)^{-|\mX||\mY|} \ 2^{-n\epsilon} \exp_2 \left\{ -n\min_{V\in\mV_n(\bQ_X)} D(V\|W|\bQ_X) \right\},
\end{align*}
where $(a)$ follows from \eqref{const-comp}. On the other hand, similarly to the calculation of $P_Y^n(\mA)$ in the proof of Proposition \ref{prop-sc-con}, here we have
\begin{align*}
    P_Y^n(\mA) 
    &\leq (n+1)^{\carX\carY} \exp_2 \left\{ -n \min_{V\in\mV_n(\bQ_X)}  \left[ D(\bQ_X V\|P_Y) + \big|I(\bQ_X;V) - R\big|^+ \right] \right\} \\
    &\leq (n+1)^{-\carX\carY} \ 2^{n\epsilon} \exp_2 \left\{ -n \min_{V\in\mV_n(\bQ_X)}  \left[ D(\bQ_X V\|P_Y) + \big|I(\bQ_X;V) - R\big|^+ \right] \right\}.
\end{align*}

Inserting the above expressions into \eqref{setA}, it is then reasonable to choose 
$$ \mV_n(Q_X) 
= \left\{ V\in\mP_n(\mY|Q_X): D(V\|W|Q_X) + 3\epsilon < D(Q_X V\|P_Y) + \big|I(Q_X;V) - R\big|^+ \right\} $$
and hence
\begin{align*}
    \frac12 \left\|\tPyQ - P_Y^n \right\|_1
    &\geq (n+1)^{-\carX\carY} \ 2^{-n\epsilon} (1-2^{-n\epsilon}) \exp_2 \left\{ -n\min_{V\in\mV_n(\bQ_X)} D(V\|W|\bQ_X) \right\} \\
    &\geq (n+1)^{-\carX\carY} \ 2^{-n\epsilon} (1-2^{-n\epsilon}) \exp_2 \left\{ -n \max_{Q_X\in\mP_n(\mX)} \ \min_{V\in\mV_n(Q_X)} D(V\|W|Q_X) \right\}.
\end{align*}
Taking $n\to\infty$ and $\epsilon\to0$, $\mV_n(\bQ_X)$ then boils down to
\begin{align}
    \mV(Q_X) 
    &= \left\{ V\in\mP(\mY|\mX): D(V\|W|Q_X) < D(Q_X V\|P_Y) + \big|I(Q_X;V) - R\big|^+ \right\} \label{WinVQx} \\
    &= \left\{ V\in\mP(\mY|\mX): \iQ > \min[R,I(Q_X; V)] \right\}. \label{i>minRI}
\end{align}
where Lemma \ref{lem-Deq} is used and thus we have acquired a bound 
\begin{equation}\label{El-ee-con}
    \Eu(R) := \max_{Q_X\in\mP(\mX)} \ \inf_{V\in\mV(Q_X)} D(V\|W|Q_X),
\end{equation}
which contains an unconstrained maximization over all input distributions. Clearly $\Eu(R)\geq0$. Furthermore, take any $Q_X\notin\mS$; then $D(Q_X W\|P_Y) > 0$ and hence \eqref{WinVQx} implies that $W\in\mV(Q_X)$. Ergo, for any $Q_X\notin\mS$, we have $\min_{V\in\mV(Q_X)} D(V\|W|Q_X) = 0$. As a result, the maximization $\max_{Q_X\in\mP(\mX)}$ in \eqref{El-ee-con} only occurs at those $Q_X\in\mS$. To be consistent with our notations ($Q$ associated with $V$ and $P$ associated with $W$), we rewrite \eqref{El-ee-con} as 
$$ \Eu(R) := \max_{P_X\in\mS} \ \inf_{V\in\mV(P_X)} D(V\|W|P_X). $$
Moreover, Lemma \ref{lem-Deq} implies that
$$ I(P_X;V) - \iota(P_X V_{Y|X}) = D(V\|W|P_X) - D(P_X V\|P_Y) \geq 0, $$
where the inequality is simply the data processing inequality of the relative entropy. Consequently, when $P_X\in\mS$ is taken, \eqref{i>minRI} reduces to
$$ \mV(P_X) = \left\{ V\in\mP(\mY|\mX): \iota(P_X V_{Y|X}) > R \right\}, $$
which completes the proof of the variational form in the theorem.

To reach the dual form, fix $P_X\in\mS$ and consider the following arguments. 
\begin{align*}
    \min_{V\in\mP(\mY|\mX): \iota(P_X V_{Y|X}) \geq R} D(V\|W|P_X) 
    &= \max_{\delta\geq0} \ \min_{V\in\mP(\mY|\mX)} \left\{ D(V\|W|P_X) + \delta\left[R - \iota(P_X V_{Y|X})\right] \right\} \\
    &\overset{a}{=} \max_{\delta\geq0} 
        \left\{ \delta R - \sum_x P_X(x) \log\left(\sum_y \frac{W_{Y|X}^{1+\delta}(y|x)}{P_Y^\delta(y)}\right) \right\} \\
    &\overset{b}{=} \max_{\alpha\geq1} 
        \left\{ (\alpha-1) \left[ R - \mathbb{E}_{P_X} D_\alpha(W_{Y|X}\|P_Y) \right] \right\}
\end{align*}
where $(a)$ follows from \eqref{cuff-minV} in Appendix \ref{app-ran}; $(b)$ follows from Definition \ref{def-Da} and setting $\alpha := \delta+1$.
\end{proof}

From Theorem \ref{thm-ee-con}, we can summarize the following properties of the proposed converse bound $\Eu(R)$, which are proved in Appendix \ref{app-lem-Eu0}.

\begin{lemma}\label{lem-Eu0} We have the following properties of $\Eu(R)$: \  

    (i) $\Eu(R) = 0$ when $ R\leq\minI$. \quad
    (ii) $\Eu(R) > 0$ when $R > \minI$.
    
    (iii) $\Eu(R) = \infty$ when 
        $R > \min_{P_X\in\mS} \max_{V\in\mP(\mY|\mX)} \iota(P_X V_{Y|X})$.
\end{lemma}

\section{Conclusion}
In the present work, we have characterized the exact strong converse exponent of the classical soft covering for rates below the mutual information. This exponent is expressed using a two-parameter information quantity that, to the best of our knowledge, has not been studied in the literature on error exponents with respect to a given channel. A promising direction for future work is a deeper investigation into the implications, properties, and potential applications of this new quantity.

Moreover, this work reveals that the conventional random coding bound is generally not tight in the achievability regime of rates above the mutual information, and that the traditional formulation assuming uniformly distributed words in the code can inherently lead to a rational–irrational discrepancy: even in the noiseless channel case, the reliability exponent diverges at sufficiently high rate for target distributions with all rational values, whereas for typical irrational probability values the exponent remains finite for all rates. Future work includes designing a well-behaved deterministic code for noisy channels that outperforms random coding in both high-rate and low-rate regimes, or even an optimal code that achieves the exact error exponent. Our observations also raise an important open question: 
can a similar rational–irrational discrepancy arise in other information-theoretic settings?

\appendices 

\renewcommand{\thesubsection}{\thesection-\alph{subsection}}

\renewcommand{\thesubsectiondis}{\alph{subsection}.}

\section{Properties of $J_{\alpha,\beta}\left(W_{Y|X}\|P_Y\right)$ (Proof of Lemma \ref{lem-Jab})} \label{app-Jab}

\begin{proof}
Statement $(a)$ is easy to verify. $(b)$ can be understood as a corollary of $(a)$: when $\alpha=1$, the optimizer is $Q_X^*\in\mS$ and hence $Q_X^* W = P_Y$. 

Now, we prove $(c)$. We follow the arguments in \cite[Thm.~4]{gallager1965simple} and \cite[Lem.~1]{arimoto1973converse}.
According to Definition \ref{def-Jab}, $J_{\alpha,\beta}\left(W_{Y|X}\|P_Y\right)$ is defined via solving the following optimization problem:
\begin{align}\label{convex-opt}
     \max_{Q_X} \sum_y P_Y^{1-\alpha\beta}(y) \left( \sum_x Q_X(x)  W_{Y|X}^\alpha(y|x) \right)^\beta \ 
    \text{ s.t. } \ Q_X(x) \geq 0, \forall x \in \mX, 
    \text{ and } 
     \sum_x Q_X(x) = 1.
\end{align}
The Lagrangian of this problem is
$$ \mathcal{L} = - \sum_y P_Y^{1-\alpha\beta}(y) \left( \sum_x Q_X(x)  W_{Y|X}^\alpha(y|x) \right)^\beta - \sum_x \lambda_x Q_X(x) + \mu \left( \sum_x Q_X(x) - 1\right). $$
Since $\beta\in[0,1]$, problem \eqref{convex-opt} is convex. The KKT conditions are necessary and sufficient for the optimizer $Q_X$ and multipliers $\lambda_x,\mu$ \cite[Ch.~5.5.3]{boyd2004convex}. Specifically, letting $Q_X^*, \lambda_x^*, \mu^*$ be the optimizers, we have
\begin{align}
    \sum_x Q_X^*(x) &= 1; \label{kkt1} \\
    \lambda_x^* &\geq 0, \ \ \lambda_x^* Q_X^*(x)=0, & \forall x \in\mX; \label{kkt2} \\
    0 &= \frac{\partial\mathcal{L}}{\partial Q_X(x)} \bigg|_{Q_X^*(x)}
    = - \beta \sum_y P_Y^{1-\alpha\beta}(y) \ W_{Y|X}^\alpha(y|x) \ \eta_y^{\beta-1} - \lambda_x^* + \mu^*, & \forall x \in\mX,  \label{kkt4}
\end{align}
where 
\begin{equation}\label{etay}
    \eta_y = \sum_x Q_X^*(x) W_{Y|X}^\alpha(y|x).
\end{equation}
Multiplying \eqref{kkt4} by $Q_X^*(x)$ and sum over $x$ gives that
\begin{equation}\label{kkt5}
    \mu^* = \beta \sum_y P_Y^{1-\alpha\beta}(y) \ \eta_y^\beta
\end{equation}
where we have plugged in \eqref{kkt1} and \eqref{kkt2}. Comparing \eqref{kkt2}, \eqref{kkt4}, and \eqref{kkt5}, $Q_X^*$ must satisfy 
\begin{equation}
\label{Q_X-opt}
    \sum_y P_Y^{1-\alpha\beta}(y) \ W_{Y|X}^\alpha(y|x) \ \eta_y^{\beta-1}
    \geq 
    \sum_y P_Y^{1-\alpha\beta}(y) \ \eta_y^\beta, \quad \forall x\in\mX,
\end{equation}
with equality for every $x\in\supp(Q_X^*)$ due to complementary slackness \eqref{kkt2}.

If $Q_X^*$ satisfies \eqref{Q_X-opt}, then $Q_X^{*n}$ also satisfies the $n$-shot version of \eqref{Q_X-opt} (replacing $W_{Y|X}$ by $W_{Y|X}^n$ and $P_Y$ by $P_Y^n$). To see this, inserting $Q_X^{*n}$ in \eqref{etay} gives
$$ \eta_{y^n} 
    = \sum_{x_1,\dots,x_n} \prod_{i=1}^n Q_X^*(x_i) 
    W_{Y|X}^\alpha(y_i|x_i) 
    = \prod_{i=1}^n \left[ \sum_{x_i} Q_X^*(x_i) W_{Y|X}^\alpha(y_i|x_i) \right]
    = \prod_{i=1}^n \eta_{y_i} $$
and hence \eqref{Q_X-opt} becomes
$$ \prod_{i=1}^n 
    \left[ \sum_{y_i} P_Y^{1-\alpha\beta}(y_i) \ W_{Y|X}^\alpha(y_i|x_i) \ \eta_{y_i}^{\beta-1} \right]
    \geq 
    \prod_{i=1}^n \left[ \sum_{y_i} P_Y^{1-\alpha\beta}(y_i) \ \eta_{y_i}^\beta \right], \quad \forall x^n\in\mX^n, $$
which is true as long as the inequality holds for each $y_i$ in the product. Considering that \eqref{Q_X-opt} is a necessary and sufficient condition for the optimizer, our conclusion is thus proved.
\end{proof}

\section{Dual forms of the Strong Converse Bounds $\Gl(R)$ and $\Gu(R)$}
\label{app-dual}
In this appendix, we establish the dual forms of the strong converse bounds $\Gl(R)$ (in Proposition \ref{prop-sc-con}) and $\Gu(R)$ (in Proposition \ref{prop-sc-achi}). That is, we prove Proposition \ref{prop-Gl-dual} and Proposition \ref{prop-Gu-dual} that are stated in Section \ref{sec-sc-dual}.

\subsection{Proof of Proposition \ref{prop-Gl-dual}}
\label{app-sec-Gl}

\begin{proof}
We first show \eqref{prop-Gl-dual-a} in Proposition \ref{prop-Gl-dual}.
Fix $Q_X\in\mP(\mX)$ and $R\geq0$. The function $\Ga(s,Q_X,R)$ defined in \eqref{Ga-s-Qx-R} describes a convex optimization of $V$: the objective function is $D(Q_X V\|P_Y) + \big|I(Q_X;V)-R\big|^+ = \max\{ D(Q_X V\|P_Y), D(V\|P_Y|Q_X) - R\}$, which is convex in $V$; the inequality constraint $D(V\|W|Q_X)\leq s$ is also convex in $V$; and the equality constraint is $\sum_y V_{Y|X}(y|x) = 1$, which is linear in $V$. Therefore, $\Ga(s,Q_X,R)$, as the optimal value over $V$, equals its Lagrangian dual due to the strong duality \cite[Ch.~5]{boyd2004convex}. Explicitly,
\begin{align*}
    \Ga(s,Q_X,R) &= \max_{\delta\geq0} \ \min_{V\in\mP(\mY|\mX)} 
        \left\{ D(Q_X V\|P_Y) + \big|I(Q_X;V)-R\big|^+
        + \delta\left[ D(V\|W|Q_X) - s\right] \right\} \\
    &= \max_{\delta\geq0} \ \min_{V\in\mP(\mY|\mX)} \ \max_{\mu\in[0,1]}
        \left\{ D(Q_X V\|P_Y) + \mu \big[I(Q_X;V)-R\big]
        + \delta \big[D(V\|W|Q_X) - s\big] \right\} \\
    &\overset{a}{=} \max_{\mu\in[0,1],\delta\geq0} \ \min_{V\in\mP(\mY|\mX)} 
        \left\{ D(Q_X V\|P_Y) + \mu \big[I(Q_X;V)-R\big]
        + \delta \big[D(V\|W|Q_X) - s\big] \right\} \\
    &\overset{b}{=} \max_{\mu\in[0,1],\delta\geq0} \ \min_{V\in\mP(\mY|\mX)}
        \big\{ (1+\delta)D(Q_X V\|P_Y) + (\mu+\delta) \left[I(Q_X;V)-R\right] + \delta\left[R-\iQ\right] -  \delta s \big\},
\end{align*}
where $(a)$ follows from Lemma \ref{lem-sion} in Appendix \ref{app-opt}: the objective function is continuous, and is convex in $V$ and linear in $\mu$, so we can swap $\min_V$ and $\max_\mu$. $(b)$ follows from Lemma \ref{lem-Deq}. Plugging this expression of $\Ga(s,Q_X,R)$ into Proposition \ref{prop-sc-con} yields that 
\begin{align*}
    \Gl(R) 
    = \ & \min_{Q_X\in\mP(\mX)} \ \min_{s\geq0} \ \max\big\{s, \Ga(s,Q_X,R)\big\} 
    = \  \min_{Q_X\in\mP(\mX)} \ \min_{s\geq0} \left\{ s + \big|\Ga(s,Q_X,R)\big\} - s\big|^+ \right\} \\
    = \ & \min_{Q_X\in\mP(\mX)} \ \min_{s\geq0} \ \max_{\lambda\in[0,1]} \left\{ s + \lambda \big[\Ga(s,Q_X,R) - s\big] \right\} \\
    = \ & \min_{Q_X\in\mP(\mX)} \ \min_{s\geq0} \ \max_{\lambda,\mu\in[0,1],\delta\geq0} \ \min_{V\in\mP(\mY|\mX)} \\
        & \big\{ \lambda(1+\delta)D(Q_X V\|P_Y) +\lambda (\mu+\delta) \left[I(Q_X;V)-R\right] + \lambda\delta\left[R-\iQ\right] + (1-\lambda-\lambda\delta) s \big\} \\
    \overset{a}{\geq} \ & 
        \min_{Q_X\in\mP(\mX)} \ \max_{\mu\in[0,1],\delta\geq0} \ \min_{V\in\mP(\mY|\mX)}
        \left\{D(Q_Y\|P_Y) + \frac{\mu+\delta}{1+\delta} \left[I(Q_X;V)-R\right] + \frac{\delta}{1+\delta}\left[R-\iQ\right] \right\} \\
    \overset{b}{=} \ & 
        \min_{Q_X\in\mP(\mX)} \ \max_{\beta,\gamma\in[0,1],\beta\geq\gamma} \ \min_{V\in\mP(\mY|\mX)}
        \big\{D(Q_Y\|P_Y) + \beta\left[I(Q_X;V)-R\right] + \gamma\left[R-\iQ\right] \big\} \\
    \overset{c}{=} \ & \min_{Q_{XY}\in\mP(\mX\mY)} \ \max_{\beta,\gamma\in[0,1],\beta\geq\gamma} 
        \big\{D(Q_Y\|P_Y) + \beta \left[I(Q_X;V)-R\right] + \gamma \left[R-\iQ\right] \big\}.
\end{align*}
Here $(a)$ follows from choosing $\lambda = \frac{1}{1+\delta}$; this is a valid choice because $\delta\geq0$ gives $\frac{1}{1+\delta}\in[0,1]$. $(b)$ follows from defining $\beta:= \frac{\mu+\delta}{1+\delta}$ and $\gamma:=\frac{\delta}{1+\delta}$ (and clearly $\beta\geq\gamma$). $(c)$ follows again from Lemma \ref{lem-sion}: the objective function is convex in $V$ and linear in $(\beta,\gamma)$, so we can swap $\min_V$ and $\max_{\beta,\gamma}$. Finally, we merge $\min_{Q_X}$ and $\min_V$ together as $\min_{Q_{XY}}$. This completes the proof of \eqref{prop-Gl-dual-a}.

Next, we prove \eqref{prop-Gl-dual-b} in Proposition \ref{prop-Gl-dual}. Observe that $D(Q_Y\|P_Y) + I(Q_X;V) = D(Q_{XY}\|Q_X P_Y)$ is convex in $Q_{XY}$ due to the joint convexity of $D(\cdot\|\cdot)$. Then for $\beta\leq1$, the expression $D(Q_Y\|P_Y) + \beta I(Q_X;V) = (1-\beta) D(Q_Y\|P_Y) + \beta D(Q_{XY}\|Q_X P_Y)$ is also convex in $Q_{XY}$. Moreover, $\iota(Q_{XY})$ is linear in $Q_{XY}$. Hence, the objective function of $\Gll(R)$ is convex in $Q_{XY}$ and linear in $(\beta,\gamma)$, so by Lemma \ref{lem-sion}, we can swap $\min_{Q_{XY}}$ and $\max_{\beta,\gamma}$ and obtain from \eqref{prop-Gl-dual-a} that
\begin{align}
    \Gll(R) 
    &= \max_{\beta,\gamma\in[0,1],\beta\geq\gamma} \ \min_{Q_{XY}\in\mP(\mX\mY)} 
        \big\{D(Q_Y\|P_Y) + \beta \left[I(Q_Y;\bV)-R\right] + \gamma \left[R-\iQ\right] \big\} \label{Gl-abQxy} \\
    &= \max_{\beta,\gamma\in[0,1],\beta\geq\gamma} \ \min_{Q_Y\in\mP(\mY)} 
        \left\{D(Q_Y\|P_Y) + (\gamma-\beta)R + K_{\beta,\gamma}(Q_Y) \right\}, \label{Gl-Kab}
\end{align}
where we have rewritten $I(Q_X;V)$ by $I(Q_Y;\bV)$ and defined
\begin{equation}\label{Kab}
    K_{\beta,\gamma}(Q_Y) = \min_{\bV\in\mP(\mX|\mY)} \big[ \beta I(Q_Y;\bV) - \gamma\iQ \big],
\end{equation}
which is a convex optimization over $\bV$ and can be solved using the method of the Lagrange multiplier. However, directly plugging in the derivative $\frac{\partial I(Q_Y;\bV)}{\partial \bV_{X|Y}(x|y)}$ leads to an equation too complicated to solve. Instead, we employ the following variational form of the mutual information $I(Q_Y;\bV)$:
\begin{equation}\label{I=minSx}
    I(Q_Y;\bV) = \min_{S_X\in\mP(\mX)} D(\bV_{X|Y} Q_Y \| S_X Q_Y).
\end{equation}
Therefore, \eqref{Kab} can be rewritten as
\begin{equation}\label{Kab-S}
    K_{\beta,\gamma}(Q_Y) = \min_{S_X\in\mP(\mX)} \ \min_{\bV\in\mP(\mX|\mY)} \big[ \beta D(\bV_{X|Y} Q_Y \| S_X Q_Y) - \gamma \iQ \big].
\end{equation}

Now, to optimize over $\bV$, we can introduce Lagrange multipliers $\lambda_y$ for each $y\in\mY$ and write the Lagrangian as
$$ \mathcal{L} = \beta D(\bV_{X|Y} Q_Y \| S_X Q_Y) - \gamma\iQ 
    + \sum_y \lambda_y \left( \sum_x \bV_{X|Y}(x|y) - 1 \right). $$
Substituting 
$$ \frac{\partial D(\bV_{X|Y} Q_Y \| S_X Q_Y)}{\partial \bV_{X|Y}(x|y)}
        = Q_Y(y) \left( 1 + \log\frac{\bV_{X|Y}(x|y)}{S_X(x)}\right), \quad
    \frac{\partial \iQ}{\partial \bV_{X|Y}(x|y)}
        = Q_Y(y) \log\frac{W_{Y|X}(y|x)}{P_Y(y)} $$
into $\frac{\partial \mathcal{L}}{\partial \bV_{X|Y}(x|y)} = 0$ gives that 
$$ Q_Y(y) \left( \beta + \beta \log\frac{\bV_{X|Y}(x|y)}{S_X(x)} - \gamma \log\frac{W_{Y|X}(y|x)}{P_Y(y)} \right) + \lambda_y = 0, $$
further generating the following solution of the optimizer $\bV^*$:
\begin{equation}\label{bV*}
    \bV^*_{X|Y}(x|y) = C(y) S_X(x) \left(\frac{W_{Y|X}(y|x)}{P_Y(y)}\right)^\frac{\gamma}{\beta}
\end{equation}
with some $y$-dependent normalization factor
$$ C(y) = \left[ \sum_x S_X(x) \left(\frac{W_{Y|X}(y|x)}{P_Y(y)}\right)^\frac{\gamma}{\beta} \right]^{-1}. $$
Hence, \eqref{Kab-S} reduces to
\begin{align*}
    K_{\beta,\gamma}(Q_Y) 
    &= \min_{S_X\in\mP(\mX)} \left[ \beta \sum_{x,y} Q_Y(y) \bV^*_{X|Y}(x|y) \log\frac{\bV^*_{X|Y}(x|y)}{S_X(x)} - \gamma \iota(\bV^*_{X|Y} Q_Y) \right] \\
    &= \min_{S_X\in\mP(\mX)} \left[ \beta \sum_{x,y} Q_Y(y) \bV^*_{X|Y}(x|y) \left(\log C(y) + \frac{\gamma}{\beta} \log \frac{W_{Y|X}(y|x)}{P_Y(y)} \right) - \gamma \iota(\bV^*_{X|Y} Q_Y) \right] \\
    &= \min_{S_X\in\mP(\mX)} \left[ \beta \sum_{x,y} Q_Y(y) \bV^*_{X|Y}(x|y) \log C(y) \right] 
    = \min_{S_X\in\mP(\mX)} \left[ \beta \sum_{y} Q_Y(y) \log C(y) \right] \\
    &= \min_{S_X\in\mP(\mX)} \left[ - \beta \sum_{y} Q_Y(y) \log \left( \sum_x S_X(x) \left(\frac{W_{Y|X}(y|x)}{P_Y(y)}\right)^\frac{\gamma}{\beta} \right) \right].
\end{align*}

Inserting this result into \eqref{Gl-Kab} yields that
\begin{align}
    \Gll(R)
    = \ & \max_{\beta,\gamma\in[0,1],\beta\geq\gamma}  \ \min_{S_X\in\mP(\mX)} \ \min_{Q_Y\in\mP(\mY)} \notag\\
        & \left\{D(Q_Y\|P_Y) + (\gamma-\beta)R - \beta \sum_{y} Q_Y(y) \log \left( \sum_x S_X(x) \left(\frac{W_{Y|X}(y|x)}{P_Y(y)}\right)^\frac{\gamma}{\beta} \right) \right\} \notag\\
    \overset{a}{=} \ & \max_{\beta,\gamma\in[0,1],\beta\geq\gamma} \ \min_{S_X\in\mP(\mX)} 
        \left\{ - \log\left( \sum_y P_Y^{1-\gamma}(y) \left[\sum_x S_X(x) \ W_{Y|X}^\frac{\gamma}{\beta}(y|x) \right]^\beta \right) + (\gamma-\beta)R \right\} \label{minSx} \\
    = \ &  \max_{\beta,\gamma\in[0,1],\beta\geq\gamma} 
        \left[ J_{\gamma/\beta,\beta}\left(W_{Y|X}\|P_Y\right) + (\gamma-\beta)R \right], \label{Jbc}
\end{align}
where $(a)$ is due to the following identity \cite[Lem.~20]{yagli2019exact} \cite[Thm.~1]{verdu2021error}:
\begin{equation}\label{cuff-minQy}
    \min_{Q_Y\in\mP(\mY)} \left[ D(Q_Y\|P_Y) - \sum_y Q_Y(y) F(y) \right]
    = -\log\left( \sum_y P_Y(y) \ 2^{F(y)} \right)
\end{equation}
for any function $F(\cdot)$ with a unique minimizer $Q_Y^*(y) \ \propto \ 2^{F(y)} P_Y(y)$. Setting $\alpha := \frac{\gamma}{\beta} \in [0,1]$ completes the proof of \eqref{prop-Gl-dual-b}. 
\end{proof}

For later convenience, we identify the optimizer here. The minimizer $Q_Y^*$ that yields \eqref{minSx} is
\begin{equation}\label{Qy*}
    Q_Y^*(y) = \frac{ P_Y^{1-\gamma}(y) \left(\displaystyle{\sum_x S_X(x) \ } W_{Y|X}^\frac{\gamma}{\beta}(y|x) \right)^\beta} {\displaystyle{\sum_{y'} P_Y^{1-\gamma}(y')} \left( \displaystyle{\sum_{x'} \ } S_X(x') \ W_{Y|X}^\frac{\gamma}{\beta}(y'|x') \right)^\beta}.
\end{equation}
Recall that we introduced $S_X$ merely to apply the variational form of mutual information \eqref{I=minSx}, so the minimizer $S_X^*$ in \eqref{minSx} is exactly the marginal $Q_X^*$ of the minimizer $Q_{XY}^*$ of \eqref{Gl-abQxy}. Hence, fixing $\beta$ and $\gamma$ (then $\alpha = \frac{\gamma}{\beta}$ is also fixed), if $J_{\alpha,\beta}\left(W_{Y|X}\|P_Y\right)$ in Definition \ref{def-Jab} yields an optimizer $Q_X^*$ (which must satisfy \eqref{Q_X-opt} in Appendix \ref{app-Jab}), then we can plug $S_X = Q_X^*$ in \eqref{bV*} and \eqref{Qy*} and obtain the corresponding optimizer $Q_{XY}^*$ of \eqref{Gl-abQxy} as
\begin{equation}\label{Qxy*}
    Q_{XY}^*(x,y) = \frac{ Q_X^*(x) \ W_{Y|X}^\frac{\gamma}{\beta}(y|x) \ P_Y^{1-\gamma}(y) \left(\displaystyle{\sum_{x'} Q_X^*(x') \ } W_{Y|X}^\frac{\gamma}{\beta}(y|x') \right)^{\beta-1} } {\displaystyle{\sum_{y'} P_Y^{1-\gamma}(y')} \left( \displaystyle{\sum_{x'} Q_X^*(x') \ }  W_{Y|X}^\frac{\gamma}{\beta}(y'|x') \right)^\beta}.
\end{equation}

Analogous to Lemma \ref{lem-Gl0}, one can observe the following properties of $\Gll(R)$ given its definition in \eqref{prop-Gl-dual-a}.

\begin{lemma}\label{lem-Gll0} We have the following properties of $\Gll(R)$: \  

    (i) $\Gll(R) = 0$ when $R \geq \minI$. \quad
    (ii) $\Gll(R) > 0$ when $R < \minI$.
\end{lemma}

\begin{proof}
For simplicity, write $\Gll(R) = \min_{Q_{XY}\in\mP(\mX\mY)} \Upsilon(R,Q_{XY})$ with
$$ \Upsilon(R,Q_{XY}) := \max_{\beta,\gamma\in[0,1],\beta\geq\gamma} 
    \big\{D(Q_Y\|P_Y) + \beta \left[I(Q_X;V)-R\right] + \gamma\left[R-\iQ\right] \big\}. $$
Given $R$, let $Q_{XY}^*\in\mP(\mY|\mX)$ optimize $\Upsilon(R,Q_{XY})$, i.e., $\Gll(R) = \Upsilon(R,Q_{XY}^*)$. Setting $\beta=\gamma=0$ gives $\Gll(R) = \Upsilon(R,Q_{XY}^*)\geq D(Q_Y^*\|P_Y)\geq0$. Hence, $\Gll(R)$ is non-negative for all $R$. \textit{(i)} follows immediately from Lemma \ref{lem-Gl0} and that $\Gll(R)\leq\Gl(R)$. 

For \textit{(ii)}, suppose not. Then $\Upsilon(R,Q_{XY}^*) = 0$ and hence plugging any combination of $(\beta,\gamma)\in[0,1]^2$ such that $\beta\geq\gamma$ into the objective function of $\Upsilon(R,Q_{XY}^*)$ will yield a non-positive value. First, plugging $\beta=\gamma=0$ gives that $D(Q_Y^*\|P_Y)=0$, so $Q_Y^* = P_Y$. Second, plugging $\beta=\gamma=1$ and combining it with Lemma \ref{lem-Deq} gives that $D(V^*\|W|Q_X^*) = 0$, so $V^* = W$ and hence $Q_X^*\in\mS$. Finally, plugging $\beta=1,\gamma=0$ gives that $I(Q_X^*;V^*)\leq R$. To sum up, we obtain that $R = I(Q_X^*;W)$ with $Q_X^*\in\mS$, contradicting the condition that $R < \minI$.
\end{proof}

\subsection{Proof of Proposition \ref{prop-Gu-dual}}\label{app-sec-Gu}

\begin{proof}
Given any $R>0$, define 
\begin{align}
    \Ga_1(R) 
    &:= \min_{Q_{XY}\in\mP(\mX\mY)} \left[ D(Q_Y\|P_Y) + \big|I(Q_X;V)-R\big|^+ + \big|R - \iQ \big|^+ \right] \label{Ga1}\\
    &=: \min_{Q_{XY}\in\mP(\mX\mY)} \Delta(Q_{XY}) = \Delta(Q_{XY}^\star). \label{Ga1-Qxy*}
\end{align}
That is, let $\Delta(\cdot)$ denote the objective function in $\Ga_1(R)$ and let $Q_{XY}^\star$ be the optimizer. We will claim that $\Gu(R) = \Ga_1(R)$. Note that the optimizer $Q_{XY}^\star$ may not be unique. Our strategy is therefore to claim the existence of some optimizer $Q_{XY}^\star$ such that $I(Q_X^\star;V^\star)\leq R$. 

We first show this is true when $\Ga_1(R)$ vanishes. Suppose $\Ga_1(R)=0$. Then $Q_Y^\star=P_Y$ and $I(Q_X^\star;V^\star)\leq R \leq\iota(Q_{XY}^\star)$. However, Lemma \ref{lem-Deq} gives that $I(Q_X^\star;V^\star) - \iota(Q_{XY}^\star) = D(V^\star\|W|Q_X^\star)\geq0$, so $V^\star = W$, and hence $Q_X^\star\in\mS$ and $R = I(Q_X^\star;W)$. Due to the continuity of $I(\cdot \ ;W)$, we have $\minI\leq R\leq\maxI$. By Lemma \ref{lem-Gu0}, we further have $\Gu(R) = \Ga_1(R) = 0$.

Next, we show that $\Gu(R) = \Ga_1(R)$ also holds in their positive regimes. In the rest of our discussions, we may assume that $\Ga_1(R)>0$. We further write $\Ga_1(R)$ as
\begin{align}\label{F-Qxy-ab}
    \Ga_1(R) &= \ \min_{Q_{XY}\in\mP(\mX\mY)} \ \max_{\beta,\gamma\in[0,1]} \big\{D(Q_Y\|P_Y) + \beta \left[I(Q_X;V)-R\right] + \gamma \left[R-\iQ\right] \big\} \notag \\
    &=: \min_{Q_{XY}\in\mP(\mX\mY)} \ \max_{\beta,\gamma\in[0,1]} F(Q_{XY},\beta,\gamma)
\end{align}
where we have used $F(Q_{XY},\beta,\gamma)$ to denote the objective function. It is shown in the proof of Proposition \ref{prop-Gl-dual} that $Q_{XY}\mapsto F(Q_{XY},\beta,\gamma)$ is convex and $(\beta,\gamma)\mapsto F(Q_{XY},\beta,\gamma)$ is linear (a trivial case of being concave). Hence, according to Lemma \ref{lem-sion} in Appendix \ref{app-opt}, the optimization in \eqref{F-Qxy-ab} can be achieved at its saddle points (there might exist some other non-saddle points that also produce the optimal value, but at least one saddle point exists). Now, suppose we have a saddle point, denoted by $(Q_{XY}',\beta',\gamma')$. It must satisfy \cite[Lem.~36.2]{rockafellar1970convex}
\begin{equation}\label{saddle}
    F(Q_{XY}',\beta,\gamma) \leq 
    F(Q_{XY}',\beta',\gamma') \leq 
    F(Q_{XY},\beta',\gamma'), \quad 
    \forall \ Q_{XY}\in\mP(\mX\mY), \ \beta,\gamma\in[0,1].
\end{equation}
It should be noted that the saddle points belong to both minimax ($\min_{Q_{XY}} \max_{\beta,\gamma} F(Q_{XY},\beta,\gamma)$) and maximin ($\max_{\beta,\gamma} \min_{Q_{XY}} F(Q_{XY},\beta,\gamma)$) solutions. Now we make the following claim:

\noindent\textbf{Claim 1.}
If $\beta'<1$, then $\Gu(R) = \Ga_1(R)$.

\begin{subproof}
If $\beta'<1$, then $Q_{XY}'$ must give $I(Q_X';V')\leq R$, otherwise the first inequality in \eqref{saddle} is violated as $F(Q_{XY}',\beta=1,\gamma=\gamma') > F(Q_{XY}',\beta',\gamma')$. Since this saddle point $(Q_{XY}',\beta',\gamma')$ is also the minimax solution, then in \eqref{Ga1-Qxy*} we can take $Q_{XY}^\star = Q_{XY}'$ and hence $\Ga_1(R) = \Delta(Q_{XY}')$: the minimization $\min_{Q_{XY}}$ in \eqref{Ga1} can always occur at a point where $I(Q_X';V')\leq R$, so $\min_{Q_{XY}}$ can reduce to a constrained optimization $\min_{Q_{XY}:I(Q_X;V)\leq R}$ and thus $\Gu(R) = \Ga_1(R)$. 
\end{subproof}

It remains to discuss the case in which $\beta'=1$. Comparing \eqref{F-Qxy-ab} and \eqref{prop-Gl-dual-a}, we can immediately write from \eqref{prop-Gl-dual-b} that
\begin{equation}\label{Ga1-dual}
    \Ga_1(R) = \max_{\beta,\gamma\in[0,1]} \left[ J_{\gamma/
    \beta,\beta}\left(W_{Y|X}\|P_Y\right) + (\gamma-\beta) R \right].
\end{equation}
Here we keep it in the form of \eqref{Jbc} with $\gamma = \alpha\beta$ for notational convenience later, and note that the condition $\beta\geq\gamma$ is not used in the proof of \eqref{Jbc}.
Since the saddle point $(Q_{XY}',\beta'=1,\gamma')$ belongs to the maximin solutions, $(\beta'=1,\gamma')$ must be the maximizers of \eqref{Ga1-dual}, and correspondingly the marginal $Q_X'$ of $Q_{XY}'$ belongs to the set of optimizers $Q_X^*$ associated with the evaluation of $J_{\gamma',1}$ by Definition \ref{def-Jab}. Specifically, $Q_X'$ is in the set
\begin{equation}\label{QX*=argmax}
    \mP(\mX_{\gamma'}) = \argmax_{Q_X\in\mP(\mX)} \ \left[ \sum_{x,y} Q_X(x) \  W_{Y|X}^{\gamma'}(y|x) \ P_Y^{1-\gamma'}(y) \right],
\end{equation}
which is a linear optimization over $Q_X$. Here $\mX_{\gamma'}:= \argmax_x \left[\sum_y W_{Y|X}^{\gamma'}(y|x) \ P_Y^{1-\gamma'}(y) \right]$ is set of optimal symbols and $\mP(\mX_{\gamma'})$ is the set of input distributions supported on $\mX_{\gamma'}$. 

Since $\beta'=1$ and we are working with $J_{\gamma',1}$, we can first exclude two trivial cases: $\gamma'=0$ or 1. If $\gamma'=0$, then \eqref{Ga1-dual} gives $\Ga_1(R) = -R$, while \eqref{Ga1} implies that $\Ga_1(R)\geq0$. Thus, this case can never happen. Similarly, if $\gamma'=1$, then Lemma \postvt \ gives $\Ga_1(R) = 0$, while we have assumed the strict positivity of $\Ga_1(R)$. 
In conclusion, as long as $\Ga_1(R)>0$, we have $\gamma'\in(0,1)$ and can further make the following claim.

\noindent\textbf{Claim 2.}
A saddle point $(Q_{XY}',\beta'=1,\gamma'\in(0,1))$ must satisfy that $I(Q_X';V')\geq R = \iota(Q_{XY}')$.

\begin{subproof}
    This follows directly from the first inequality in \eqref{saddle}.
\end{subproof}

In the rest of our discussions, we may assume that $\gamma'\in(0,1)$ (while still assuming $\beta'=1$). In most circumstances, $\mX_{\gamma'}$ may contain only one symbol, and the corresponding maximizer of \eqref{QX*=argmax} is a deterministic distribution. More generally, however, $\mX_{\gamma'}$ may contain multiple symbols, and all distributions supported on these symbols are maximizers of \eqref{QX*=argmax}; the marginal $Q_X'$ of the saddle point $Q_{XY}'$ is only one of those distributions. Let us consider this general situation: if $\mX_{\gamma'}$ contains multiple symbols, we take any two symbols, say,  $x_1,x_2\in\mX_{\gamma}$. Denote
$$ G_{x}(\gamma) 
    := 2^{(\gamma-1) D_\gamma(W_{Y|x}\|P_Y)} 
    = \sum_y W_{Y|X}^{\gamma}(y|x) \ P_Y^{1-\gamma}(y), $$
where $D_\gamma$ is the $\gamma$-R\'enyi divergence in Definition \ref{def-Da}. Then we have $G_{x_1}(\gamma') = G_{x_2}(\gamma')$, meaning that these two functions of $\gamma$ intersect at $\gamma'$. Since they are both analytic, in the neighborhood of $\gamma'$, they either overlap or intersect only at $\gamma'$. The former indicates a certain symmetry in $W_{Y|X}$ and $P_Y$ (for example, $W_{Y|X}$ is a binary symmetric channel and $P_Y$ is uniform). We define that $x_1$ and $x_2$ belong to the same \textit{class}, if $G_{x_1}(\gamma) = G_{x_2}(\gamma)$ for all $\gamma$ in the neighborhood of $\gamma'$. We now conduct our discussions based on how many classes $\mX_{\gamma'}$ contains, which can be classified into the following three cases. 

\noindent\textbf{Case 1.}
$\mX_{\gamma'}$ contains only one class, and there is only one symbol in that class.

In this case, $\mX_{\gamma'}$ basically contains a single symbol, say, $\mX_{\gamma'} = \{x_0\}$ for some $x_0\in\mX$. Then $\mP(\mX_{\gamma'}) = \{\delta_{x,x_0}\}$, so \eqref{QX*=argmax} has a unique solution which is a deterministic distribution. Since the saddle marginal $Q_X'\in\mP(\mX_{\gamma'})$ is in this solution set, we must have $Q_X'=\delta_{x,x_0}$; however, this gives $I(Q_X';V') = 0<R$, contradicting $I(Q_X';V')\geq R$ in Claim 2. Ergo, this case is impossible. In essence, this suggests that there exists no saddle point in the form of $(Q_{XY}',\beta'=1,\gamma'\in(0,1))$ when $\mP(\mX_{\gamma'})$ contains only a delta distribution. The saddle point must have $\beta'<1$. We are back to Claim 1 and hence $\Gu(R) = \Ga_1(R)$.

\noindent\textbf{Case 2.} $\mX_{\gamma'}$ contains only one class, and there is more than one symbol in that class.

In this case, take any two symbols, say, $x_1,x_2\in\mX_{\gamma'}$. Then $G_{x_1}(\gamma)$ and $ G_{x_2}(\gamma)$ are locally the same function. We can further take their derivatives at $\gamma'$ and write $G_{x_1}'(\gamma') = G_{x_2}'(\gamma')$. As a result, the expression
\begin{equation}\label{Gb-deri}
    \sum_y W_{Y|X}^{\gamma'}(y|x) \ P_Y^{1-\gamma'}(y) \log\frac{W_{Y|X}(y|x)}{P_Y(y)} 
\end{equation}
will be independent of symbol choice in that class, i.e., \eqref{Gb-deri} is identical for all $x\in\mX_{\gamma'}$. Then, take any $Q_X^*\in\mP(\mX_{\gamma'})$ (so $Q_X^*$ is an optimizer in $J_{\gamma',1}$) and its corresponding joint optimizer $Q_{XY}^*$ for $\min_{Q_{XY}} F(Q_{XY},1,\gamma')$ is given by \eqref{Qxy*} (under $\beta=1$ and $\gamma=\gamma'$). The expression
$$ \iota(Q_{XY}^*) 
    = \frac{\displaystyle{ \sum_x Q_X^*(x) \sum_y W_{Y|X}^{\gamma'}(y|x) \ P_Y^{1-\gamma'}(y) \log\frac{W_{Y|X}(y|x)}{P_Y(y)}}} {\displaystyle{ \sum_x Q_X^*(x) \sum_y W_{Y|X}^{\gamma'}(y|x) \ P_Y^{1-\gamma'}(y)} } $$
will be independent of $Q_X^*$. Since the saddle point $Q_{XY}'$ also has marginal $Q_X'\in\mP(\mX_{\gamma'})$ and satisfies $\iota(Q_{XY}') = R$ according to Claim 2, we must have $\iota(Q_{XY}^*) = R$ for all $Q_X^*\in\mP(\mX_{\gamma'})$. Note that $\mP(\mX_{\gamma'})$ contains the deterministic distribution (which yields a zero mutual information) and the saddle point (which yields $I(Q_X';V')\geq R$ by Claim 2). By continuity of distributions supported on $\mX_{\gamma'}$, there must exist some $Q_X^\star\in\mP(\mX_{\gamma'})$ such that $I(Q_X^\star;V^\star) = R = \iota(Q_{XY}^\star)$. Plugging this $Q_{XY}^\star$ in \eqref{Ga1-Qxy*} gives
\begin{align*}
    \Delta(Q_{XY}^\star)
    &= D(Q_Y^\star\|P_Y) + \big|I(Q_X^\star;V^\star)-R\big|^+ + \big|R - \iota(Q_{XY}^\star) \big|^+ \\
    &= D(Q_Y^\star\|P_Y) + \left[ I(Q_X^\star;V^\star)-R\right] + \gamma'\left[R - \iota(Q_{XY}^\star) \right] \\
    &= J_{\gamma',1}\left(W_{Y|X}\|P_Y\right) + (\gamma'-1) R 
    = \Ga_1(R).
\end{align*}
This means the minimization $\min_{Q_{XY}}$ in \eqref{Ga1} can occur at a point where $I(Q_X^\star;V^\star) = R$, so $\Gu(R) = \Ga_1(R)$.

\noindent\textbf{Case 3.} $\mX_{\gamma'}$ contains more than one class. 

In this case, we cannot deduce much about $\gamma'$, and hence about the corresponding $R$. However, let us perturb $\gamma'$ in its neighborhood, say, to some $\gamma''\in(0,1)$. Due to the existence of only finitely many symbols, in the neighborhood of $\gamma'$, there are only finite intersections of $G_x(\gamma)$'s for $x$'s from different classes, so $\mX_{\gamma''}$ must contain only one class. Take any symbol, say, $x_1$ from this class. Then we can write $\gamma''$ as
\begin{align*}
    \gamma'' 
    &= \argmax_{\gamma\in[0,1]} \left\{ J_{\gamma,1}\left(W_{Y|X}\|P_Y\right) + (\gamma-1) (R+\Delta R) \right\} \\
    &= \argmax_{\gamma\in[0,1]} \left\{ (1-\gamma) \left[ D_\gamma\left(W_{Y|x_1}\|P_Y\right) - (R+\Delta R) \right] \right\}
\end{align*}
for some $R+\Delta R$ in the neighborhood of $R$. Note that this is a strictly concave maximization due to the strict concavity of $\gamma\mapsto (1-\gamma)D_\gamma$ (if $(1-\gamma)D_\gamma$ is linear in $\gamma$, then the maximizer $\gamma''$ must be 0 or 1). Thus, this new $R+\Delta R$ is uniquely determined by $\gamma''$. If the corresponding saddle point at $R+\Delta R$ has $\beta''<1$, then we are back to Claim 1. If it has $\beta''=1$, then that saddle point must be $(Q_{XY}'',\beta''=1,\gamma'')$, and we are back to Case 1 or Case 2 and can obtain that $\Gu(R+\Delta R) = \Ga_1(R+\Delta R)$. In summary, for all $R+\Delta R$ (where $\Delta R\neq0$) in the neighborhood of $R$, we have $\Gu(R+\Delta R) = \Ga_1(R+\Delta R)$. Ergo, $\Gu(\cdot) = \Ga_1(\cdot)$ holds almost everywhere in the neighborhood of $R$. On the other hand, $\Gu(\cdot)$ and $\Ga_1(\cdot)$ are both continuous due to Lemma \ref{lem-berge} in Appendix \ref{app-opt} (note that in $\Gu(\cdot)$, $R\twoheadrightarrow\{Q_{XY}\in\mP(\mX\mY): I(Q_X;V)\leq R\}$ is a continuous and compact correspondence), so $\Gu(R) = \Ga_1(R)$ also holds at $R$.

To sum up, $\Gu(R) = \Ga_1(R)$ holds everywhere, in both vanishing and positive regimes. The dual form of $\Gu(R)$ follows immediately from \eqref{Ga1-dual}. Finally, we set $\alpha := \frac{\gamma}{\beta}$. Note that here we no longer impose the condition $\beta\geq\gamma$, so the range of $\alpha$ is $\alpha\geq0$. This completes the proof.
\end{proof}

\subsection{Proof of Proposition \ref{prop-Gu=Gll}}\label{app-sec-Gu=Gll}
\begin{proof}
We first show that both $\Gu(R)$ and $\Gll(R)$ are convex in $R$. Collecting Proposition \ref{prop-Gl-dual} and \ref{prop-Gu-dual}, both $\Gu(R)$ and $\Gll(R)$ take the dual form
\begin{equation}\label{dual}
    \max_{\alpha,\beta} \left[ J_{\alpha,\beta}\left(W_{Y|X}\|P_Y\right) + \beta(\alpha-1) R \right].
\end{equation}
Their only difference is the optimization range in $(\alpha,\beta)$. For a fixed combination of $(\alpha,\beta)$, the objective function in \eqref{dual} is linear in $R$ (a trivial case of being convex). Then after maximizing over $(\alpha,\beta)$, \eqref{dual} is convex in $R$.

Next, we show that both $\Gu(R)$ and $\Gll(R)$ are monotone decreasing when $R<\minI$. Combining convexity and Lemma \ref{lem-Gu0}, $\Gu(R)$ is convex in $R$ and positive when $R<\minI$. As $R$ grows, $\Gu(R)$ remains convex and must vanish when $R$ reaches $\minI$. Ergo, $\Gu(R)$ must be monotone decreasing in $R$ when $R<\minI$. Similar arguments can be established for $\Gll(R)$ by combining its convexity and Lemma \ref{lem-Gll0}.

Finally, to prove our conclusion, it suffices to show that when $R<\minI$, the optimization over $\alpha\geq0$ in $\Gu(R)$ occurs at some $\alpha\in[0,1]$. Let $R<\minI$ and the optimizers in $\Gu(R)$ be $\alpha',\beta'$, i.e., $\Gu(R) = J_{\alpha',\beta'}\left(W_{Y|X}\|P_Y\right) + \beta'(\alpha'-1') R$.  
Suppose $\alpha'>1$. Take $\Delta R\to 0^+$ and we have
\begin{align*}
    \Gu(R+\Delta R) 
    &= \max_{\alpha\geq0,\beta\in[0,1]} \left[ J_{\alpha,\beta}\left(W_{Y|X}\|P_Y\right) + \beta(\alpha-1) (R+\Delta R) \right] \\
    &\geq J_{\alpha',\beta'}\left(W_{Y|X}\|P_Y\right) + \beta'(\alpha'-1) (R+\Delta R) \\
    &= \Gu(R) + \beta'(\alpha'-1) \Delta R > \Gu(R),
\end{align*}
meaning that $\Gu(R)$ is locally increasing in $R$, which contradicts the decreasing property that we just showed. Hence, when $R<\minI$, the optimizer $\alpha'$ for $\Gu(R)$ must lie in $[0,1]$. 
\end{proof}

\section{A Different Approach to Proving the Converse Part of Theorem \ref{thm-sc}}\label{app-arimoto} 

Based on Arimoto's techniques in \cite{arimoto1973converse}, we provide a different proof of the converse part of Theorem \ref{thm-sc} for the uniform formulation, i.e., we show $\Ga_\su(R) \geq \Ga(R)$.

\begin{proof}
For simplicity, we first narrow down our discussions to the 1-shot case ($n=1$). Let $\mC^* = \{X^*(1),\ldots,X^*(M)\}$ be the optimal code that achieves the minimal total variation using $M$ codewords, i.e.,
$$ \frac12 \left\| \tP_{Y^n|\mC^*} - P_Y^n \right\|_1 
= \min_{\mC:|\mC|= M} \frac12 \left\| \tPy - P_Y^n \right\|_1 . $$ 
Then there exists a probabilistic distribution on the set of all codes, say, $\tQ_\mC := \tQ_\mC(X(1),\ldots,X(M))$ that satisfies the following two properties.
\begin{itemize}
    \item [\textit{(i)}] The expectation of total variation under $\tQ_\mC$ equals to the minimal total variation:
        $$ \mathbb{E}_{\tQ_\mC} \left[ \frac12 \left\|\tP_{Y|\mC} - P_Y \right\|_1 \right] 
        = \frac12 \left\|\tP_{Y|\mC^*} - P_Y^n \right\|_1. $$
    \item [\textit{(ii)}] The marginal distribution of each codeword is identical, i.e.,
        $$ \tQ_X(X(i)) 
        := \sum_{X(1)} \cdots \sum_{X(i-1)} \sum_{X(i+1)} \cdots \sum_{X(M)} \tQ_\mC(X(1),\ldots,X(M)) $$
        is the same for each $i=1,\ldots,M$. 
\end{itemize}
One example of such a $\tQ_\mC$ is \cite{arimoto1973converse}:
\[
  \tQ_\mC = \begin{cases}
              \dfrac{1}{M!} & \text{if } \mC = \mC^* \text{ up to permutations of codewords}, \\
              0 & \text{otherwise}.
            \end{cases}
\]

Taking any $\alpha,\beta\in[0,1]$, we have the following:
\begin{align*}
    1 - \frac12 \left\|\tP_{Y|\mC} - P_Y \right\|_1 
    &\leq 1 - \frac12 \left\|\tP_{Y|\mC^*} - P_Y \right\|_1 
    = \mathbb{E}_{\tQ_\mC} \left[ 1 - \frac12 \left\|\tP_{Y|\mC^*} - P_Y \right\|_1 \right] \\
    &\overset{a}{=} \mathbb{E}_{\tQ_\mC} \sum_y P_Y(y) \min\left\{ \frac{1}{M} \sum_{i=1}^M \frac{W_{Y|X}(y|X(i))}{P_Y(y)},1 \right\} \\
    &\overset{b}{\leq} \mathbb{E}_{\tQ_\mC} \sum_y P_Y(y) \left( \frac{1}{M} \sum_{i=1}^M \frac{W_{Y|X}(y|X(i))}{P_Y(y)} \right)^{\alpha\beta} \\
    &= M^{-\alpha\beta} \sum_y P_Y^{1-\alpha\beta}(y) \ \mathbb{E}_{\tQ_\mC} \left( \sum_{i=1}^M W_{Y|X}(y|X(i)) \right)^{\alpha\beta} \\
    &\overset{c}{\leq} M^{-\alpha\beta} \sum_y P_Y^{1-\alpha\beta}(y) \left( \mathbb{E}_{\tQ_\mC} \left[  \sum_{i=1}^M W_{Y|X}(y|X(i)) \right]^\alpha\right)^\beta \\
    &\overset{d}{\leq} M^{-\alpha\beta} \sum_y P_Y^{1-\alpha\beta}(y) \left( \mathbb{E}_{\tQ_\mC} \left[ \sum_{i=1}^M W_{Y|X}^\alpha(y|X(i)) \right] \right)^\beta  \\
    &\overset{e}{=} M^{-\alpha\beta} \sum_y P_Y^{1-\alpha\beta}(y) \left( M \ \mathbb{E}_{\tQ_X(X(1))} \left[W_{Y|X}^\alpha(y|X(1))\right] \right)^\beta \\
    &= M^{\beta-\alpha\beta} \sum_y P_Y^{1-\alpha\beta}(y) \left( \sum_x \tQ_X(x) W_{Y|X}^\alpha(y|x) \right)^\beta \\
    &\leq M^{-\beta(\alpha-1)} \max_{Q_X\in\mP(\mX)} \sum_y P_Y^{1-\alpha\beta}(y) \left( \sum_x Q_X(x) W_{Y|X}^\alpha(y|x) \right)^\beta \\
    &= 2^{-\left[ J_{\alpha,\beta}\left(W_{Y|X}\|P_Y\right) + \beta(\alpha-1) R \right]}.
\end{align*}
Here $(a)$ follows from \eqref{1-L1=min}. $(b)$ is due to that $\min\{x,1\}\leq x^\gamma$ when $x\geq0$ and $\gamma=\alpha\beta\in[0,1]$. $(c)$ applies Jensen's inequality to the concave function $x^\beta$. $(d)$ follows from the relation $(x+y)^\alpha \leq x^\alpha + y^\alpha$ for $x,y\geq0$ and $\alpha\in[0,1]$. $(e)$ follows from the second property of $\tQ_\mC$. 
Extending this expression to the $n$-shot case, our conclusion immediately follows from the additivity of $J_{\alpha,\beta}$ in Lemma \addtvt.
\end{proof}

\section{A Random Coding Achievability for the Strong Converse Exponent} \label{app-ran}
In this appendix, we prove Theorem \ref{thm-rc}. We start with the following lemma.

\begin{lemma}\label{lemma-binomial}
    Let $K\sim \text{Binomial}(M,p)$ be a binomial random variable and $M\geq2$. We have
    $$ \frac12 \mathbb{E} \left|\frac{K}{M} - p\right| \leq p - \frac12 p \min\{Mp, 1\}. $$
\end{lemma}

\begin{proof}
Write 
\begin{equation}\label{bi}
    \frac12 \mathbb{E} \left|\frac{K}{M} - p\right|
    = \frac12 \sum_{k=0}^{\lfloor Mp \rfloor} \Pr\{K=k\} \left(p - \frac{k}{M}\right) + \frac12 \sum_{k = \lceil Mp \rceil}^M \Pr\{K=k\} \left(\frac{k}{M} - p\right).
\end{equation}
and consider the following two cases.

\begin{itemize}[leftmargin=12.5pt]
    \item [1)] $Mp < 1$:
        In this case $\lfloor Mp \rfloor = 0$ and $\lceil Mp \rceil = 1$. \eqref{bi} reduces to
    \begin{align*}
        \frac12 \mathbb{E} \left|\frac{K}{M} - p\right| 
        &= \frac12 \Pr\{K=0\} p + \frac12 \sum_{k=1}^M \Pr\{K=k\} \left(\frac{k}{M} - p\right) \\
        &= \frac12 (1-p)^M p + \frac{\mathbb{E}K}{2M} - \frac12 [1-(1-p)^M]p 
        = (1-p)^M p \\
        &\overset{a}{\leq} p - \frac12 Mp^2,
    \end{align*}
    where $(a)$ follows from, when $M\geq 2$,
    $$ (1-p)^M \leq 1 - Mp + \frac12 M(M-1)p^2 \leq 1 - Mp + \frac12 M^2 p^2 \leq 1 - Mp + \frac12 Mp = 1 - \frac12Mp. $$
    
\item [2)] $Mp > 1$:
    In this case $\lfloor Mp \rfloor\geq 1$, analyzing \eqref{bi} becomes challenging. Instead, we can apply Jensen's inequality for the square root, which is a commonly used technique in the soft covering problem:
    \[
      \frac12 \mathbb{E} \left|\frac{K}{M} - p\right| 
        = \frac{\mathbb{E} \left|K - \mathbb{E}K \right|}{2M} 
        \leq \frac{\sqrt{\text{Var}(K)}}{2M} 
        = \frac12 \sqrt{\frac{p(1-p)}{M}}
        \leq \frac12 \sqrt{\frac{p}{M}}
        \leq \frac12 p.
    \]
\end{itemize}
Combining these two cases gives the result.
\end{proof}

Now we prove Theorem \ref{thm-rc} using the random coding strategy. 

\begin{proof}[Proof of Theorem \ref{thm-rc}]
We first prove the variational form \eqref{thm-rc-a}. Take any $P_X\in\mathcal{S}$ and generate a random code $\mC = \{X^n(1),\dots,X^n(M)\}$ with $M=2^{nR}$, where each codeword $X^n(i)$ for any $i=1,\dots,M$ is drawn from i.i.d. $P_X$.
Consider any joint type $Q_{XY}\in\mP_n(\mX\mY)$. Similar to the proof of Proposition \ref{prop-sc-achi}, introduce the backward conditional type $\bV_{X|Y} = Q_{XY}/Q_Y\in\mP_n(\mX|Q_Y)$. We define
\begin{align}
    \omega(Q_{XY}) &:= W_{Y|X}^n(y^n|x^n), \quad \qquad \forall (x^n,y^n)\in\mT_{Q_{XY}}. \label{wQXY} \\
    p_{\bV}(Q_Y) &:= \Pr\{X^n\in \mT_{\bV}(y^n)\}, \ \quad \forall y^n\in\mT_{Q_Y}. \label{pV(y)}
\end{align}
The values of $\omega(Q_{XY})$ and $p_{\bV}(Q_Y)$ are uniquely determined by the joint type $Q_{XY}$, independently of the particular sequence $x^n,y^n$. Take any $y^n\in \mT_{Q_Y}$. Then we can write
\begin{align}
    \tP_{Y^n|\mC}(y^n) 
    &= \frac{1}{M} \sum_{i=1}^M W_{Y|X}^n(y^n|X^n(i)) 
    = \sum_{\bV\in\mP_n(\mX|Q_Y)} \frac{k_{\bV}(y^n)}{M} \omega(Q_{XY}), \notag\\
    P_Y^n(y^n)
    &= \sum_{x^n} P_X^n(x^n) W_{Y|X}^n(y^n|x^n) 
    = \sum_{\bV\in\mP_n(\mX|Q_Y)} p_{\bV}(Q_Y) \omega(Q_{XY}),   \label{Py-type}
\end{align}
where $k_{\bV}(y^n)$ is defined in \eqref{kv}. Under random coding,
\begin{align}\label{EcTV}
    \EcTVinEq &= \frac12 \sum_{y^n} \mathbb{E}_{\mC} \left|\tP_{Y^n|\mC}(y^n) - P_Y^n(y^n)\right| \notag\\
    &= \frac12 \sum_{Q_Y\in\mP_n(\mY)} \ \sum_{y^n\in\mT_{Q_Y}} \mathbb{E}_{\mC} \left|\sum_{\bV\in\mP_n(\mX|Q_Y)} \omega(Q_{XY}) \left( \frac{k_{\bV}(y^n)}{M} - p_{\bV}(Q_Y) \right) \right| \notag\\
    &\leq \frac12 \sum_{Q_Y\in\mP_n(\mY)} \ \sum_{y^n\in\mT_{Q_Y}} \ \sum_{\bV\in\mP_n(\mX|Q_Y)} 
        \omega(Q_{XY}) \ \mathbb{E}_{\mC} \left|\frac{k_{\bV}(y^n)}{M} - p_{\bV}(Q_Y)\right| \notag\\
    &\overset{a}{\leq} \sum_{Q_Y\in\mP_n(\mY)} \ \sum_{y^n\in\mT_{Q_Y}} \ \sum_{\bV\in\mP_n(\mX|Q_Y)} 
        \omega(Q_{XY}) \left( p_{\bV}(Q_Y) - \frac12 p_{\bV}(Q_Y) \min\{Mp_{\bV}(Q_Y), 1\} \right) \notag\\
    &\overset{b}{=} \sum_{Q_Y\in\mP_n(\mY)} \ \sum_{y^n\in\mT_{Q_Y}} 
        \left( P_Y^n(y^n) - \frac12 \sum_{\bV\in\mP_n(\mX|Q_Y)} \omega(Q_{XY}) \ p_{\bV}(Q_Y) \min\{Mp_{\bV}(Q_Y), 1\} \right) \notag\\
    &= 1 - \frac12 \sum_{Q_Y\in\mP_n(\mY)} \ \sum_{\bV\in\mP_n(\mX|Q_Y)} |\mT_{Q_Y}|\ \omega(Q_{XY}) \ p_{\bV}(Q_Y) \min\{Mp_{\bV}(Q_Y), 1\}.
\end{align}
where $(a)$ follows from Lemma \ref{lemma-binomial} because $k_{\bV}(y^n) \sim \text{Binomial}(M,p_{\bV}(Q_Y))$ according to its definition \eqref{kv}; and we identify $P_Y(y^n)$ in $(b)$ from \eqref{Py-type}. Furthermore, from \eqref{wQXY} and \eqref{pV(y)} we have
\begin{align*}
    |\mT_{Q_Y}| \ \omega(Q_{XY}) 
    &\overset{a}{\geq} (n+1)^{-|\mathcal{Y}|} \ 2^{-n[D(V\|W|Q_X) + H(V|Q_X) - H(Q_Y)]} \\
    &= (n+1)^{-|\mathcal{Y}|} \ 2^{-n[D(Q_{XY}\|P_{XY}) - D(Q_{XY}\|P_X Q_Y)]}, \\
    p_{\bV}(Q_Y) 
    &\overset{b}{\geq} (n+1)^{-|\mathcal{X}||\mathcal{Y}|} \ 2^{-n[D(Q_Y\bV\|P_X) + H(Q_Y\bV) - H(\bV|Q_Y)]} \\
    &= (n+1)^{-|\mathcal{X}||\mathcal{Y}|} \ 2^{-nD(Q_{XY}\|P_X Q_Y)},
\end{align*}
where $(a)$ and $(b)$ follow from general properties of types.
Then \eqref{EcTV} simplifies to
\begin{align*}
\EcTVinEq
    &\leq 1 - \frac12 (n+1)^{-\left(|\mathcal{X}|+1\right)|\mathcal{Y}|} \sum_{Q_{XY}\in\mP_n(\mX\mY)} 2^{-n\left[D(Q_{XY}\|P_{XY}) + |D(Q_{XY}\|P_X Q_Y) - R|^+ \right]} \\
    &\!\!\!\!\!\!\!\!\!\!\!\!\!\!\!\!\!\!
     \leq    1 - \frac12 (n+1)^{-\left(|\mathcal{X}|+1\right)|\mathcal{Y}|} \exp_2 \left\{ -n \min_{Q_{XY}\in\mP_n(\mX\mY)} \left[D(Q_{XY}\|P_{XY}) + \big|D(Q_{XY}\|P_X Q_Y) - R\big|^+ \right] \right\}.
\end{align*}
Taking $n\to\infty$ completes the proof of \eqref{thm-rc-a}.

Next, we prove the dual form \eqref{thm-rc-b}. Pick any $P_X\in\mS$ and introduce the corresponding joint distribution $P_{XY} = P_X W_{Y|X}$. In line with our notation convention, we also define the associated backward channel $\bW_{X|Y} := P_{XY}/P_Y \in \mP(\mX|\mY)$. For any other distribution $Q_{XY}$, it is straightforward to verify the following identities.
\begin{align}
    D(Q_{XY}\|P_{XY}) &= D(\bV\|\bW|Q_Y) + D(Q_Y\|P_Y), \label{D1} \\
    D(Q_{XY}\|P_X Q_Y) &= D(\bV\|\bW|Q_Y) + \iQ, \label{D2}
\end{align}
where $\iQ$ is defined in Definition \ref{def-iQ}. It is noteworthy that $\iQ$ satisfies the following property \cite[Cor.~1]{yagli2019exact}, which can be viewed as a corollary of \eqref{cuff-minQy}:
\begin{equation}\label{cuff-minV}
    \min_{\bV\in\mP(\mX|\mY)} \left\{D(\bV\|\bW|Q_Y) + \lambda\iQ\right\} 
    = - \mathbb{E}_{Q_Y} \left[ \log \left(\mathbb{E}_{\bW_{X|Y}}[2^{-\lambda \iota_{X;Y}}|Y] \right) \right]
\end{equation}
for any $\lambda\in\mathbb{R}$. Now consider the following chain of transformations:
\begin{align*}
\min_{Q_{XY}\in\mP(\mX\mY)} &\left\{D(Q_{XY}\|P_{XY}) + \big|D(Q_{XY}\|P_X Q_Y) - R\big|^+ \right\} \\
    &= \min_{Q_Y\in\mP(\mY)} \ \min_{\bV\in\mP(\mX|\mY)} \ \max_{\lambda\in[0,1]} 
        \big\{D(Q_{XY}\|P_{XY}) + \lambda\left[D(Q_{XY}\|P_X Q_Y) - R\right] \big\} \\
    &\overset{a}{=} \min_{Q_Y\in\mP(\mY)} \ \min_{\bV\in\mP(\mX|\mY)} \ \max_{\lambda\in[0,1]}
        \left\{ D(Q_Y\|P_Y) + (1+\lambda)D(\bV\|\bW|Q_Y) + \lambda \left[ \iQ - R \right] \right\} \\
    &\overset{b}{=} \min_{Q_Y\in\mP(\mY)} \ \max_{\lambda\in[0,1]} \ \min_{\bV\in\mP(\mX|\mY)} 
        \left\{ D(Q_Y\|P_Y) + (1+\lambda)D(\bV\|\bW|Q_Y) + \lambda \left[ \iQ - R \right] \right\} \\
    &\overset{c}{=} \min_{Q_Y\in\mP(\mY)} \ \max_{\lambda\in[0,1]}
        \left\{ D(Q_Y\|P_Y) - (1+\lambda) \mathbb{E}_{Q_Y} \left[ \log \left(\mathbb{E}_{\bW_{X|Y}} \big[2^{-\frac{\lambda}{1+\lambda} \iota_{X;Y}}|Y\big] \right) \right] - \lambda R \right\} \\
    &\overset{d}{=} \max_{\lambda\in[0,1]} \ \min_{Q_Y\in\mP(\mY)}
        \left\{ D(Q_Y\|P_Y) - (1+\lambda) \mathbb{E}_{Q_Y} \left[ \log \left(\mathbb{E}_{\bW_{X|Y}} \big[2^{-\frac{\lambda}{1+\lambda} \iota_{X;Y}}|Y\big] \right) \right] - \lambda R \right\} \\
    &\overset{e}{=} \max_{\lambda\in[0,1]} 
        \left\{-\log\left( \mathbb{E}_{P_Y} \left[ \mathbb{E}_{\bW_{X|Y}}^{1+\lambda} \big[2^{-\frac{\lambda}{1+\lambda} \iota_{X;Y}}|Y\big]\right] \right) - \lambda R \right\} \\
    &\overset{f}{=} \max_{\lambda\in[0,1]}
        \left\{ \lambda \left[ I_{\frac{1}{1+\lambda}}(P_X; W_{Y|X}) - R \right] \right\} \\
    &\overset{g}{=} \max_{\alpha\in[\frac12,1]}
        \left\{ \frac{1-\alpha}{\alpha} \left[ I_\alpha(P_X; W_{Y|X}) - R\right] \right\},
\end{align*}
where $(a)$ follows from \eqref{D1} and \eqref{D2}; $(b)$ follows from Lemma \ref{lem-sion} in Appendix \ref{app-opt}: the expression is continuous, and is convex in $\bV$ and linear in $\lambda$ so we can swap $\min_{\bV}$ and $\max_\lambda$; and $(c)$ follows from \eqref{cuff-minV}. In $(d)$, Lemma \ref{lem-sion} is invoked again: it is evident to observe that the expression is convex in $Q_Y$; however, its concavity in $\lambda$ is not obvious. We need to return to the right-hand side of $(b)$: fix $\bV$, the objective expression is linear in $\lambda$ (a trivial case of being concave), so the minimization over $\bV$ yields a concave function of $\lambda$. $(e)$ is a direct result of \eqref{cuff-minQy}; $(f)$ follows from Lemma \ref{lem-Ia-yagli}; and $(g)$ is obtained by setting $\alpha:= \frac{1}{1+\lambda}$. This completes the proof of \eqref{thm-rc-b}.
\end{proof}

\section{Proofs of Some Lemmas and Propositions}

\subsection{Properties of $\Ga(s,Q_X,R)$}\label{app-fs}

\begin{lemma}
\label{lem-fs}
Fix any $Q_X\in\mP(\mX)$ and $R\geq0$. For simplicity, write $f(s) := \Ga(s,Q_X,R)$, which is defined in \eqref{Ga-s-Qx-R}.
Then $f(s)$ is a continuous, non-increasing function of $s$. It is strictly decreasing on the interval $[0,s_0]$ for some $s_0\geq0$ and constant on $[s_0,\infty)$.
\end{lemma}

\begin{proof}
For convenience, define the objective function in $f(s)$ to be
$g(V) := D(Q_X V\|P_Y) + \big|I(Q_X;V) - R\big|^+ 
    = \max \big\{ D(Q_X V\|P_Y), D(V\|P_Y|Q_X) - R \big\}.$
Then $f(s) = \min_{V\in\mP(\mY|\mX): D(V\|W|Q_X)\leq s} g(V)$.
    
We first show continuity. Clearly, $g(\cdot)$ being the maximum of two continuous functions, is continuous. Also, $s\twoheadrightarrow\{D(V\|W|Q_X)\leq s\}$ a continuous, compact, and non-empty correspondence. Thanks to Lemma \ref{lem-berge} in Appendix \ref{app-opt}, $f(s)$ is continuous. 

Next, the non-increasing property of $f(s)$ is straightforward: the feasible set $D(V\|W|Q_X)\leq s$ expands as $s$ grows, and thus the minimization value can never increase. When $s=0$, the only feasible $V$ is $V=W$ and then $f(0) = D(Q_X W\|P_Y) + \big|I(Q_X;W) - R\big|^+ \geq 0$. On the other hand, note that $D(V\|W|Q_X)$ is bounded when $V\ll W$, so when $s\geq s_0 := \max_{V\in\mP(\mY|\mX): V\ll W} D(V\|W|Q_X)$, we have
$$ f(s) = \min_{V\in\mP(\mY|\mX): V\ll W} \left[ D(Q_X V\|P_Y) + \big|I(Q_X;V) - R\big|^+ \right], $$
which is a constant. Therefore, $f(s) \equiv f(s_0)$ on $[s_0,\infty)$. 

Now, before moving to the decreasing property on $[0,s_0]$, we prove the convexity of $f(s)$. Take $s_1, s_2\geq0$. Suppose that the optimizers for $f(s_1)$ and $f(s_2)$ are $V_1$ and $V_2$, respectively; that is,
$$ f(s_i) = g(V_i), \quad D(V_i\|W|Q_X)\leq s_i, \quad i = 1,2. $$
For any $t\in[0,1]$, take $s_t = t s_1 + (1-t) s_2$ and $V_t = t V_1 + (1-t) V_2$. By the convexity of $D(\cdot\|W|Q_X)$, we have 
$ D(V_t\|W|Q_X) = D \big([t V_1 + (1-t) V_2] \big\| W \big|Q_X \big) 
    \leq tD(V_1\|W|Q_X) + (1-t) D(V_2\|W|Q_X) 
    \leq t s_1 + (1-t) s_2 = s_t. $
Now since $D(V_t\|W|Q_X)\leq s_t$, we further have
\begin{align}
    f(s_t) &= \min_{V\in\mP(\mY|\mX): D(V\|W|Q_X)\leq s_t} g(V) \notag
    \leq g(V_t) = g\big([t V_1 + (1-t) V_2]\big) \notag\\
    &\overset{a}{\leq} tg(V_1) + (1-t)g(V_2) 
    = t f(s_1) + (1-t) f(s_2), \label{es-convex}
\end{align}
where $(a)$ follows from the convexity of $g(\cdot)$. To see this, observe that both $D(Q_X V\|P_Y)$ and $D(V\|P_Y|Q_X)$ are convex in $V$. Thus, $g(\cdot)$ being the maximum of two convex functions, is also convex. \eqref{es-convex} thereby indicates the convexity of $f(\cdot)$.

Equipped with convexity and the non-increasing monotonicity, $f(s)$ must be strictly decreasing on the interval $[0,s_0]$. It is also noteworthy that on this interval, the optimization of $V$ is achieved at the boundary where $D(V\|W|Q_X) = s$. Similar arguments can be found in the proof of \cite[Cor.~10.4]{csiszar2011information}. 
\end{proof}

\subsection{Proof of Proposition \ref{prop-nl-exact}}\label{app-nl-exact}
\begin{proof}
We only show the expression for $\El^\nl(R)$, while that for $E_\rc^\nl(R)$ follows from analogous reasoning. We apply the method in \cite[Cor.~10.4]{csiszar2011information} and first show the convexity of $\Eu^\nl(\cdot)$.
Let $Q_{Y1},Q_{Y2}$ be the optimizer for $\Eu^\nl(R_1)$ and $\Eu^\nl(R_2)$, respectively, i.e.,
$$ \Eu^\nl(R_i) = D(Q_{Yi}\|P_Y), \quad
D(Q_{Yi}\|P_Y) + H(Q_{Yi}) \geq R_i, \quad i = 1,2. $$
Take any $t\in[0,1]$. Note that $D(Q_Y\|P_Y) + H(Q_Y)$ is linear in $Q_Y$. Then $D(Q_Y\|P_Y) + H(Q_Y)\geq R$ implies that $D(tQ_Y\|P_Y) + H(tQ_Y)\geq tR$ for the unnormalized distribution $tQ_Y$. Thus,
\begin{align*}
    \El^\nl \left(tR_1 + (1-t)R_2\right)
    &= \min_{Q_Y\in\mP(\mY): D(Q_Y\|P_Y) + H(Q_Y) > tR_1 + (1-t)R_2 } D(Q_Y\|P_Y) \\
    &\leq D\left( \left[tQ_{Y1} + (1-t)Q_{Y2}\right] \|P_Y\right) \\
    &\overset{a}{\leq} t D(Q_{Y1}\|P_Y) + (1-t) D(Q_{Y2}\|P_Y) \\
    &= t \Eu^\nl(R_1) + (1-t) \Eu^\nl(R_2),
\end{align*}
where $(a)$ follows from the convexity of $D(\cdot\|P_Y)$. So far, we have shown that $\Eu^\nl(R)$ is convex in $R$.

Moreover, observe that $\Eu^\nl(R)$ is non-decreasing in $R$ and non-negative. Due to its convexity, $\Eu^\nl(R)$ is strictly increasing in $R$ in the interval where it is finite and positive. Then the optimization of $Q_Y$ must be achieved at the boundary:
\begin{equation}\label{Eu-nl-boundary}
    \Eu^\nl(R) = \min_{Q_Y\in\mP(\mY): D(Q_Y\|P_Y) + H(Q_Y) = R} D(Q_Y\|P_Y),
    \quad \text{when } R\leq H_{-\infty}(P_Y).
\end{equation}
Hence,
\begin{align*}
    \El^\nl(R) 
    &= \min_{Q_Y\in\mP(\mY)} \left\{ D(Q_Y\|P_Y) + \big| R - D(Q_Y\|P_Y) - H(Q_Y) \big|^+ \right\} \\
    &= \min_{R'\leq H_{-\infty}(P_Y)} \ \min_{Q_Y\in\mP(\mY): D(Q_Y\|P_Y) + H(Q_Y) = R'} \left\{ D(Q_Y\|P_Y) + \big| R - D(Q_Y\|P_Y) - H(Q_Y) \big|^+ \right\} \\
    &= \min_{R'\leq H_{-\infty}(P_Y)} \left\{ \Eu^\nl(R') + \big| R - R' \big|^+ \right\}
    \overset{a}{=} \min_{R'\leq R} \left\{ \Eu^\nl(R') + \big| R - R' \big|^+ \right\},
\end{align*}
where $(a)$ follows from the fact $\Eu^\nl(\cdot)$ is monotone increasing. 
\end{proof}

\subsection{Proof of Lemma \ref{lem-Gl0}}
\label{app-lem-Gl0}
\begin{proof}
Let $P_X^*$ be the minimizer of $\minI$, i.e., $\minI = I(P_X^*;W)$.

For \textit{(i)}, clearly $\Gl(R)$ is non-negative. Recall the expression of $\Ga(s,Q_X,R)$ in \eqref{Ga-s-Qx-R}. Take $Q_X=P_X^*$. Then $\Ga(s,P_X^*,R) \equiv 0$ for all $s\geq0$, because $V=W$ is the optimizer. Then $\Gl(R) \leq \min_{s\geq0} \max\big\{s, \Ga(s,P_X^*,R)\big\} = 0$, and hence $\Gl(R) = 0$.

For \textit{(ii)}, suppose not. Owing to \eqref{maxs=mins}, there exists some $Q_X^*\in\mP(\mX)$ such that $\Ga(s,Q_X^*,R) \equiv 0$ for all $s\geq0$. Take $s=0$, the only optimizer is $V^*=W$ and thus $0 = \Ga(0,Q_X^*,R) = D(Q_X^* W\|P_Y) + \big|I(Q_X^*;W) - R\big|^+$, indicating that $Q_X^*\in\mS$ and $I(Q_X^*;W)\leq R$. This contradicts the condition $R < \minI$.
\end{proof}

\subsection{Proof of Lemma \ref{lem-Gu0}}\label{app-lem-Gu0}
\begin{proof}
When $\minI \leq R \leq \maxI$, we can write $R = I(P_X;W)$ for some $P_X\in\mS$, since $I(\cdot \ ;W)$ is continuous. Taking $Q_{XY} = P_X W_{Y|X}$, we have $\Gl(R) \leq D(P_Y\|P_Y) + \big|R - \iota(P_X W_{Y|X})\big|^+ = 0$ because $\iota(P_X W_{Y|X}) = I(P_X;W) = R$. On the other hand, $\Gu(R)$ is non-negative, so $\Gu(R)=0$.
\end{proof}

\subsection{Proof of Lemma \ref{lem-nl-Euinf}}\label{app-lem-nl-Euinf}
\begin{proof}
Observe that 
$$ \max_{Q_Y\in\mP(\mY)} \big[ D(Q_Y\|P_Y) + H(Q_Y) \big] 
= \max_{Q_Y\in\mP(\mY)} \left[ -\sum_y Q_Y(y)\log P_Y(y) \right]
= - \min_y \log P_Y(y) 
= H_{-\infty}(P_Y), $$
where the last equality follows from Remark \ref{rmk-Hinf}. 
Therefore, when $R>H_{-\infty}(P_Y)$, the exponent $\Eu^\nl(R) = \infty$ is not well-defined. In fact, the bad set $\mB$ in \eqref{badset} becomes empty and hence we can only apply a trivial bound $\TVinText\geq0$, yielding an infinite exponent.
\end{proof}

\subsection{Proof of Lemma \ref{lem-N1N2}}\label{app-lem-N1N2}

\begin{proof}
Let $Q_Y^*:= \text{argmax}_{Q_Y\in\mP(\mY): D(Q_Y\|P_Y) + H(Q_Y) = R} H(Q_Y)$. Note that the optimization here has an equality constraint $D(Q_Y\|P_Y) + H(Q_Y) = R$. 

Observe that $R-D(Q_Y\|P_Y)$ is concave in $Q_Y$, with its maximum attained at $Q_Y = P_Y$. However, the feasible set of $N_2(R)$ does not contain $P_Y$. Then the maximization of $R-D(Q_Y\|P_Y)$ occurs at the boundary of the linear equation $D(Q_Y\|P_Y) + H(Q_Y) = R$, yielding that
\begin{align*}
    N_2(R) 
    &= \max_{Q_Y\in\mP(\mY): D(Q_Y\|P_Y) + H(Q_Y) = R} \left[R - D(Q_Y\|P_Y)\right] \\
    &= \max_{Q_Y\in\mP(\mY): D(Q_Y\|P_Y) + H(Q_Y) = R} H(Q_Y) 
    = H(Q_Y^*).
\end{align*}

Similarly, $H(Q_Y)$ is also concave in $Q_Y$, but the feasible set for $N_1(R)$ expands with $R$. This implies that the maximizer for $N_1(R)$ lies either on the boundary $D(Q_Y\|P_Y) + H(Q_Y) = R$ or at the uniform distribution $Q_Y = 1/|\mY|$. In summary,
$N_1(R) = H(Q_Y^*)$ or $N_1(R) = \log|\mY|$, and therefore $N_1(R) \geq N_2(R)$.
\end{proof}

\subsection{Proof of Lemma \ref{lem-equiv}}\label{app-lem-equiv}

\begin{proof}
For $(a)$, let $P_{Y^n,\hY^n,I}$ denote the joint distribution of the source sequence $Y^n$, the reconstruction sequence $\hY^n$, and the random index $I$ taking values in the set $\{1,\dots,M\}$; namely,
$$ P_{Y^n,\hY^n,I}(y^n,\hy^n,i) 
:= P_Y^n(y^n) \mathbbm{1}\{i = \mE(y^n)\} \ \mathbbm{1}\{\hy^n = \mD(i)\}. $$
Its marginal on $I$ determines a message distribution 
\begin{equation}\label{qi-source}
    q(i) := P_I(i) = \sum_{y^n} P_Y^n(y^n) \mathbbm{1}\{i = \mE(y^n)\}.
\end{equation}
We choose $q$ as the message distribution, and $w^{-1}\circ\mD$ as the covering code $\mC$. Since $w$ is a bijection, $w^{-1}$ is well-defined. The induced distribution in \eqref{tPyQ} then becomes
\begin{align*}
    \tPyQ(y^n) 
    &= \sum_{i=1}^M q(i) \sum_{x^n} \mathbbm{1}\{y^n = w(x^n)\} \mathbbm{1}\{x^n = w^{-1}(\mD(i))\} \\
    &= \sum_{i=1}^M \sum_{y^n} P_Y^n(y^n) \mathbbm{1}\{i = \mE(y^n)\} \sum_{x^n} \mathbbm{1}\{y^n = w(x^n)\} \mathbbm{1}\{w(x^n) = \mD(i)\} \\
    &= \sum_{\hy^n} P_{\hY^n}(\hy^n) \mathbbm1\{y^n = \hy^n\}.
\end{align*}
Here $P_{\hY^n}$ is the marginal of $P_{Y^n,\hY^n,I}$. Also, note that $P_Y^n$ is naturally the marginal of $P_{Y^n,\hY^n,I}$. We obtain that
\begin{align*}
    \left\| \tPyQ - P_Y^n \right\|_1 
    &= \left\| \sum_{\hy^n} P_{\hY^n}(\hy^n) \mathbbm1\{Y^n = \hy^n\} - P_Y^n \right\|_1 \\
    &\leq \left\|P_{\hY^n} \mathbbm1\{Y^n = \hY^n\} - P_{Y^n,\hY^n} \right\|_1 \\
    &= \sum_{y^n,\hy^n} \left| P_{\hY^n}(y^n) \mathbbm1\{y^n = \hy^n\} - P_{Y^n,\hY^n}(y^n,\hy^n) \right| \\
    &= \sum_{y^n = \hy^n} \left[ P_{\hY^n}(y^n) - P_{Y^n,\hY^n}(y^n,\hy^n) \right] + \sum_{y^n\neq\hy^n} P_{Y^n,\hY^n}(y^n,\hy^n) \\
    &= 1 - \Pr\{Y^n=\hY^n\} + \Pr\{Y^n\neq\hY^n\} 
    = 2\Pe,
\end{align*}
where the inequality follows from the monotonicity of the total variation.

For $(b)$, let $\mC:\mM\to\mX^n$ be the covering code. Define 
\begin{align*}
    \mJ &= \{y^n\in\mY^n: 
    \text{there exists } i\in\mM \text{ such that } y^n = w(\mC(i)) \} 
    = \{y^n\in\mY^n: \tPyQ(y^n) > 0 \}
\end{align*}
We choose our encoder and decoder in the lossless source coding scheme to be 
\begin{align*}
    \mE(y^n) &= 
        \begin{cases}
            i \text{ for some } i\in\mM \text{ such that } y^n = w(\mC(i)) & \text{if } y^n\in\mJ \\
            0 & \text{if } y^n\notin\mJ
        \end{cases}, \\
    \mD(i) &= 
        \begin{cases}
            w(\mC(i)) & \text{if } i\in \mM\cap\mE(\mY^n) \\
            \text{declare error} & \text{if } i = 0 
        \end{cases}.
\end{align*}
It is noteworthy that the given soft covering code $\mC$ may contain repeated codewords. As a result, for some $y^n\in\mJ$, there exists more than one index $i$ such that $y^n = w(\mC(i))$, and we just select one representative index among them as our source coding encoder $\mE(y^n)$. Consequently, only $y^n\in\mJ$ be reconstructed correctly. We obtain that
\begin{align*}
    \left\| \tPyQ - P_Y^n \right\|_1 
    \geq \sum_{y^n\notin\mJ} P_Y^n(y^n) 
    = \Pe.
\end{align*}
Observe that the lossless source coding scheme contains $M+1$ indices (some of them may be unused if $\mC$ contains repeated codewords), and the rate is $\frac{1}{n}\log(M+1)\to R$ for sufficiently large $n$. 
\end{proof}

\subsection{Proof of Lemma \ref{lem-Eu0}}\label{app-lem-Eu0}
\begin{proof}
Fix $P_X\in\mS$ and define
$$ \Eu(R, P_X) := \min_{V\in\mP(\mY|\mX): \iota(P_X V_{Y|X})\geq R} D(V\|W|P_X). $$
Let $V^*$ be its optimizer, i.e., $\Eu(R, P_X) = D(V^*\|W|P_X)$ with $\iota(P_X V^*_{Y|X})\geq R$.
When $R\leq I(P_X;W)$, we have $V^*= W$ and $\Eu(R, P_X) = 0$.
When $R> I(P_X;W)$, note that $\iota(P_X W_{Y|X}) = I(P_X;W)$, so $V^*\neq W$ and hence $\Eu(R, P_X) > 0$.
When $R > \max_{V\in\mP(\mY|\mX)} \iota(P_X V_{Y|X})$, $V^*$ does not exist; we can only apply a trivial bound $\frac12 \big\| \tPyQ - P_Y^n \big\|_1 \geq 0$ and the exponent is $\Eu(R, P_X) = \infty$.

In view that $\Eu(R) = \max_{P_X\in\mS} \Eu(R, P_X) $, the above turning points, including the one from zero to nonzero value, and the one from finite to infinite value, of $\Eu(R)$ are obtained via minimizing over $P_X\in\mS$.
\end{proof}

\section{Optimization Theorems}\label{app-opt} 

In this appendix, we list two optimization theorems used in this paper for the reader’s reference.

\begin{lemma}[Sion minimax theorem\cite{komiya1988elementary}]\label{lem-sion}
    Let $\mX$ be a compact and convex subset of a linear vector space, and let $\mY$ be a convex subset of a linear vector space. Let $F \colon \mX\times\mY \to \mathbb{R}$ be an real-valued function such that
    \begin{itemize}
        \item [1)] The function $F(\cdot,y) \colon \mX\to \mathbb{R} $ is lower semi-continuous and quasi-convex
        on $\mX$ for every $y\in\mY$;
        \item [2)] The function $F(x,\cdot) \colon \mY\to \mathbb{R} $ is upper semi-continuous and quasi-concave
        on $\mY$ for every $x\in\mX$.
    \end{itemize}
    Then
    $$ \min_{x\in\mX} \sup_{y\in\mY} F(x,y) 
    = \sup_{y\in\mY} \min_{x\in\mX} F(x,y). $$
\end{lemma}

\begin{lemma}[Berge maximum theorem\text{\cite[Item~17.31]{aliprantis2006infinite}}]\label{lem-berge}
    Let $\mX$ and $\mY$ be two topological spaces and $\varphi\colon\mX\twoheadrightarrow\mY$ be a continuous and compact correspondence. Suppose $f\colon\mathrm{Graph}(\varphi)\to\mathbb{R}$ is continuous. Then the value function $m(x) = \max_{y\in\varphi(x)}f(x,y)$ is continuous.
\end{lemma}


\bibliographystyle{IEEEtran}
\bibliography{ref}

\end{document}